%% file: main.tex
\documentclass[11pt,letterpaper,final]{article}
\input{header}

\usepackage{enumitem}
\usepackage{parskip}
\usepackage[notcite,notref]{showkeys}

\makeatletter
\def\thm@space@setup{%
  \thm@preskip=\parskip \thm@postskip=0pt
}
\makeatother

\newcommand{\newkappa}{\kappa(C^*)}

\newcommand{\LSred}{\LS'_{\text{red}}}
\newcommand{\LSblue}{\LS'_{\text{blue}}}
\newcommand{\lsd}{D_0}
\newcommand{\lsdres}{\smash{\widehat{D}_0}}

\newcommand{\err}{\blue{\delta}}

\usepackage{setspace}

\title{A Better-than-$2$ Approximation for $k$-Median\\
  Part II: Stable Instances}

\author{
Rares-Darius Buhai\thanks{EPFL, Lausanne and supported by the Swiss State Secretariat for Education, Research and Innovation (SERI) under contract number MB22.00054.}
\and
Vincent Cohen-Addad\thanks{Google Research, New York NY.}
\and
Henry Fleischmann\thanks{Computer Science Department, Carnegie Mellon
  University, Pittsburgh PA. Work supported by a National Science Foundation Graduate
Research Fellowship Program under Grant No(s).\ DGE2140739 and DGE2631988. }
\and
Anupam Gupta\thanks{Computer Science Department, Courant Institute of
  Mathematical Sciences, New York University, New York, and Google Research. Work supported in part by NSF awards CCF-2422926 and CCF-2608359.}
\and 
Antoni Jub\'es Monforte\thanks{EPFL, Lausanne}
\and 
Madhusudhan Reddy Pittu\thanks{Computer Science Department, Courant Institute of
  Mathematical Sciences, New York University, New York, and  H.\ Milton
  Stewart School of Industrial and Systems Engineering, Georgia Tech,
  Atlanta GA.}
\and 
Miltiadis Stouras\footnotemark[1]
\and 
Ola Svensson\footnotemark[1]
\and 
Lukas Vogl\footnotemark[1]
}

\date{}
\begin{document}

\maketitle

\begin{abstract}
  We consider the $k$-median problem on general metric spaces, subject
  to the condition that the instances are $\beta$-stable (i.e., the
  optimal solution using $k-1$ centers has a cost at least $(1+\beta)$
  times that of the optimal solution with $k$ centers). We give a
  $(1.89+\eps)$-approximation algorithm for such instances running in
  time $\exp(1/(\beta \times \poly(\eps)))\cdot \poly_{\eps}(n)$. In
  particular, for $\beta =\poly(\eps)/\log n$ and $\eps$ a tiny constant, this is a polynomial-time algorithm. This improves
  on previous work showing a $(2 + \eps)$-approximation algorithm with
  similar running time \cite{2approx}. Combined with the pseudo-approximation algorithm in our companion work~\cite{companion-paper}, this  
  results in a better-than-$2$ approximation algorithm for general $k$-median instances.

  Our main technical contribution is a novel nuanced analysis of local
  search that allows us to isolate a subset of clients on
  which we lose essentially nothing compared to the optimal solution,
  while losing an approximation factor of $3$ on the rest of the
  clients. We complement this with an algorithm based on sampling
  and partial enumeration, achieving a better-than-$2$
  approximation on the rest of the clients. Trading off these two algorithms yields our final
  approximation. 
\end{abstract}

\newpage 

\tableofcontents

\input{intro}

\input{prelims}

\input{outline}

\input{merged-structure}
\input{new-submod}

\input{new-local}




{\small
\bibliographystyle{alpha}
\bibliography{references}
}

\newpage
\appendix

\input{app-local-search}
\input{app-submod}
\input{cost-bd-lemma}

\end{document}

%% file: header.tex
\usepackage[T1]{fontenc}
\usepackage{babel}
\usepackage{verbatim}
\usepackage{bm}
\usepackage{amsmath,amstext,amsthm,amssymb}
\usepackage{float}
\usepackage{xspace}
\usepackage{graphicx}
\usepackage{mathtools}
\usepackage{subcaption,array}
\usepackage{wrapfig}
\usepackage{booktabs}
\usepackage{makecell}
\usepackage{array}
\usepackage{multirow}
\usepackage[usenames,dvipsnames]{xcolor}
\usepackage{boxedminipage,wrapfig}
\usepackage{pifont}%

\usepackage[letterpaper,top=1in,bottom=2cm,left=1in,right=1in,marginparwidth=1.75cm,footskip=1cm]{geometry}
\usepackage[font={small,it}]{caption}
\usepackage{amsmath, amssymb}
\usepackage{amsthm}
\usepackage{comment}
\usepackage{mdframed}
\usepackage{xspace}
\usepackage{multicol}
\usepackage{tabularx}
\usepackage{makecell}

\usepackage[ruled,vlined,linesnumbered,algonl,procnumbered]{algorithm2e}
\SetEndCharOfAlgoLine{}
\SetKwComment{Comment}{\footnotesize$\triangleright$\ }{}
\SetCommentSty{mycommfont}

\usepackage{enumitem}
\usepackage{nicefrac}

\usepackage{thmtools} 
\usepackage{thm-restate}

\newtheorem{theorem}{Theorem}[section]
\newtheorem{lemma}[theorem]{Lemma}

\newtheorem{corollary}[theorem]{Corollary}
\newtheorem{proposition}[theorem]{Proposition}

\newtheorem{claim}[theorem]{Claim}
\newtheorem*{claim*}{Claim}
\newtheorem{definition}[theorem]{Definition}

\usepackage{thmtools}

\allowdisplaybreaks

\declaretheoremstyle[
    bodyfont=\normalfont  %
]{upright}

\usepackage{graphicx}
\usepackage[colorlinks=true, allcolors=blue]{hyperref}

\newcommand{\eps}{\varepsilon}
\newcommand{\sse}{\subseteq}

\newcommand{\calA}{\mathcal{A}}
\newcommand{\calB}{\mathcal{B}}

\newcommand{\calD}{\mathcal{D}}

\newcommand{\calG}{\mathcal{G}}
\newcommand{\calH}{\mathcal{H}}

\newcommand{\calL}{\mathcal{L}}
\newcommand{\calM}{\mathcal{M}}

\newcommand{\calO}{\mathcal{O}}

\newcommand{\calS}{\mathcal{S}}

\DeclareMathOperator{\core}{core}
\DeclareMathOperator{\sopt}{\mathsf{opt}}
\DeclareMathOperator{\poly}{poly}
\DeclareMathOperator{\avg}{avg}
\DeclareMathOperator{\per}{per}
\DeclareMathOperator{\deficit}{def}
\DeclareMathOperator{\excess}{ex}
\DeclareMathOperator{\cost}{cost}

\DeclareMathOperator{\early}{early}
\DeclareMathOperator{\late}{late}
\DeclareMathOperator{\test}{test}

\renewcommand{\AA}{\mathbb{A}}
\newcommand{\BB}{\mathbb{B}}

\newcommand{\EE}{\mathbb{E}}
\newcommand{\FF}{\mathbb{F}}

\newcommand{\LL}{\mathbb{L}}
\newcommand{\MM}{\mathbb{M}}

\newcommand{\OO}{\mathbb{O}}

\newcommand{\RR}{\mathbb{R}}
\renewcommand{\SS}{\mathbb{S}}
\newcommand{\TT}{\mathbb{T}}
\newcommand{\UU}{\mathbb{U}}

\newcommand{\XX}{\mathbb{X}}
\newcommand{\YY}{\mathbb{Y}}
\newcommand{\ZZ}{\mathbb{Z}}

\newcommand{\OPT}{\mathbb{O}}
\newcommand{\LS}{\LL}

\usepackage[colorinlistoftodos,textsize=tiny,textwidth=2cm,color=green!50!gray]{todonotes}

\newlength\hwhatcwidth%
\newlength\hwhatuwidth%
\newcommand*\halfwidehat[1]{%
  \settowidth\hwhatcwidth{\ensuremath{#1}}%
  \settowidth\hwhatuwidth{\ensuremath{\mathrm{#1}}}%
  \addtolength\hwhatuwidth{-\hwhatcwidth}%
  \makebox[0pt][l]{%
    \kern\dimexpr0.25\hwhatcwidth-0.667\hwhatuwidth\relax%
    \ensuremath{\widehat{\vphantom{#1}\rule{0.5\hwhatcwidth}{0pt}}}%
  }#1%
}%

\usepackage{todonotes}

\newcommand{\nf}{\nicefrac}

\newcommand{\eat}[1]{}
\renewcommand{\emptyset}{\varnothing}

\newcommand{\blue}[1]{{\color{blue}#1}}

\usepackage[nameinlink]{cleveref}

\crefname{observation}{observation}{observations}
\Crefname{observation}{Observation}{Observations}
\Crefname{proposition}{Proposition}{Propositions}

\Crefname{algocf}{Algorithm}{Algorithms}
\crefname{algocfline}{line}{lines}
\Crefname{invariant}{Invariant}{Invariants}
\Crefname{claim}{Claim}{Claims}
\Crefname{subclaim}{Subclaim}{Subclaims}
\Crefname{table}{Table}{Tables}

\newcommand{\mypara}[1]{\medskip\noindent\textbf{#1}}

\newcommand{\margindef}[1]{\marginpar{~~~~{\footnotesize(\blue{#1})}}}

\newcommand{\proj}{\mathrm{proj}}

\newcommand{\olamearly}{\OO^{\early}_{\Lambda}}
\newcommand{\olamlate}{\OO^{\late}_{\Lambda}}

%% file: intro.tex
\section{Introduction}
\label{sec:introduction}

In the metric $k$-median problem, we are given a metric space
$(D \cup \FF, d)$ on a point set consisting of a set $\FF$ of
potential facility locations and a set of clients/demands $D$. The
goal is to find a subset $\SS \sse \FF$ of $k$ facilities such that
the $k$-median cost $\sum_{p\in D} d(p, \SS)$ is minimized. The
problem is $\mathsf{APX}$-hard, and obtaining a near-optimal
approximation guarantee has long been a central problem in the field
of approximation algorithms. A remarkable recent breakthrough showed a
polynomial-time $(2+\eps)$-approximation algorithm~\cite{2approx};
moreover, we cannot do better than a $(1+2/e)$-approximation unless
$\mathsf{P} = \mathsf{NP}$~\cite{GuhaK98,JMMSV03}.

Given the $\mathsf{APX}$-hardness result, it is natural to ask: \emph{can we get
  better approximations to the optimal solution in beyond-worst-case
  settings?} One such class of $k$-median instances are \emph{stable
  instances}, formalized by Ostrovsky, Rabani, Schulman, and
Swamy~\cite{ORSS12}: these are instances for which reducing the number
$k$ of available centers from $k$ to $k-1$ increases the optimal cost
substantially. Formally, given a metric space $(D \cup \FF, d)$,
define $\sopt_\ell$ to be the optimal cost of $\ell$-median on this
instance. Naturally, fewer centers result in a larger cost, and hence
$\sopt_k \leq \sopt_{k-c}$ for any $c \geq 0$.  An instance is called
\emph{$\beta$-stable} if $\sopt_{k-1} \geq (1+\beta)\sopt_k$. 
\emph{Can we beat the worst-case hardness results for $\beta$-stable
  instances?} Indeed, Awasthi, Blum, and Sheffet~\cite{ABS10} answered
this question in the affirmative, giving a $(1+\eps)$-approximation
in time $(nk)^{\poly(\nf1{\eps\beta})}$; Cohen-Addad and
Schwiegelshohn~\cite{cs17} gave an algorithm with runtime
$n^{\beta^{-1}\eps^{-O(1)}}$.

Since the runtime depends on $n^{1/\beta}$, the algorithms are not
polynomial-time for sub-constant $\beta$. This leads to the next
question: \emph{can we get a fixed-parameter PTAS: a
  $(1+\eps)$-approximation in time (say)
  $\exp\{(\beta\eps)^{-O(1)}\} \poly_\eps(n)$?} Sadly, this seems unlikely:
\cite{2approx} observe that getting a better-than-$(1+2/e)$ approximation
for all $\beta$-stable instances cannot be done (assuming the GAP-ETH
conjecture) in any time $f(\beta, \eps) \poly(n)$, based on the hardness
results from~\cite{AGKLL19}. This leads us to refine the question:
what approximation factor can be achieved in FPT-PTAS time?

Complementing their hardness, \cite{2approx} also gave a
$(2+\eps)$-approximation algorithm with runtime
$\exp(\beta^{-1} \eps^{-O(1)}) \poly_\eps(n)$. We build on and simplify their
ideas, and add several new ones, to achieve our main result for stable
instances.
\begin{restatable}[Main Theorem]{theorem}{mainResult}\label{thm:main-result}
  Let $(D \cup \FF, d)$ be a $\beta$-stable instance of $k$-median and
  $\eps \in (0, \nf{1}{200})$. There exists an algorithm running in
  $\exp\{O(\beta^{-1} \poly(1/\eps))\} n^{\poly(1/\eps)}$ time which with
  probability $1- \eps$ outputs $\SS \subseteq \FF$ with $|\SS| = k$
  such that
  \[ \cost(\SS) \leq (1.89 +
    O(\eps))\sopt_k. \]
\end{restatable}

Moreover, in a companion paper~\cite{companion-paper}, we show how
\Cref{thm:main-result} can be combined with a pseudo-approximation
algorithm opening $k + O(\eps^{-O(1)} \log n)$ centers with cost
$(1.95 + \eps)\sopt_k$ to get the first \emph{polynomial-time}
better-than-$2$-approximation algorithm for arbitrary $k$-median
instances:
\begin{theorem}
  \label{thm:kmed-approx}
  There is a randomized polynomial-time $1.95$-approximation algorithm for
  $k$-median.
\end{theorem}

\subsection{Technical Overview}

Our algorithm starts with the following observation: the $k$-median
function $\cost(\SS)$ is a \emph{supermodular} decreasing function of
the set $\SS$---so given any \emph{base solution} $\SS_0$, the
\emph{improvement function} $\cost(\SS_0) - \cost(\SS_0 \cup \SS)$ is
increasing and submodular. Now we can use submodular function
maximization (with side constraints) on this improvement function to
find a set $\AA$ with $|\AA| = k$ such that
\begin{gather}
  \cost(\AA \cup \SS_0) \leq (1-\nf1e) \, \cost(\OPT \cup \SS_0) +
  (\nf1e) \, \cost(\SS_0). \label{eq:tradeoff}
\end{gather}
Note that $\AA \cup \SS_0$ has too many centers. So if we could choose
$\SS_0$ to ensure: (a)~$\cost(\SS_0) \leq 3\sopt$, and (b)~the
solution $\AA \cup \SS_0$ could be converted into a solution only
using $\AA$ with little increase in cost---e.g., if the clusters
corresponding to centers in $\SS_0$ in the clustering for
$(\AA \cup \SS_0)$ were empty---we would get the optimal result
$\cost(\AA) \leq (1+2/e)\, \sopt$. How can we find a good base
solution $\SS_0$? This approach was previously used by \cite{AGKLL19}
who guess a client in each optimal cluster, along with its distance to
the optimal center, create a ``dummy'' center at that distance, and
use the set of $k$ dummies as the base solution. However, this
approach runs in $f(k)\poly(n)$ time for exponentially growing $f$,
whereas we want polynomial-time solutions.

\paragraph{Step 1: Sync, Substitute, and Sample.}
We start off with the approach used in \cite{2approx}: we use
$\beta$-stability to infer that the local search solution is already
very ``close to'' an optimal solution.  Indeed, consider the
local-search solution $\LS$ using $O(1/\eps)$-sized swaps, which is
known to be a $(3+\eps)$-approximation~\cite{arya2001local}. But that
is not all: \cite{cs17} show that all but
$s_0 = O(\eps^{-2}\beta^{-1})$ of the centers in $\LS$ can be matched
(``synced'') to the optimal centers, such that substituting these
centers for the corresponding optimal centers incurs very small cost
(\Cref{lem:few-not-in-sync}). If we knew the identity of these
``unsynced'' $s_0$-many centers in $\LS$, we could fix the remainder
and try to find replacements for just the few bad ones. In fact, we do
not need to replace all the unsynced centers: if an optimal center
$c^*$ has a center $c$ in $\LS$ at distance somewhat less than the
average cost of a client in $c^*$'s cluster, we can just use $c$ to
substitute for $c^*$ with little increase in cost. There are other
substitutions possible as well. Our approach, like in \cite{2approx},
is to first sample clients and then centers in $\LS$ (both with probability
proportional to their costs in the local-search solution), performing
some more substitutions along the way to get our main structural
result (\Cref{thm:structural}) which shows the existence of a
partition of the optimal centers $\OPT$ into four parts
$\OPT_S \cup \OPT_{\Lambda} \cup \OPT_P \cup \OPT_R$:
\begin{itemize}
\item \emph{Substituted centers} ($\OPT_S$): the clients in these
  centers' clusters can be reassigned to a nearby center already in
  $\LS$ while increasing their cost by only a factor of roughly
  $1+ \eta \approx 1.89$. Almost all clusters fall into this case.
\item \emph{Centers with dummies ($\OPT_\Lambda$) or backups
    ($\OPT_P$)}: for these centers, we can explicitly find ``dummy''
  centers to add to the instance, or we are promised substitute
  centers in $\LS$ which can serve as good ``backups.''  Dummies and
  backups act as part of the ``low-cost base solution'' $\SS_0$ in the
  sketch at the start of this section.
\item \emph{``Remaining'' centers} ($\OPT_R$): the remaining clusters
  have several somewhat nearby centers in $\LS$ (their
  \textit{satellites}), but the local-search clusters for these
  satellites are too inexpensive to detect via sampling, and
  too spread out to permit cost-efficient substitution.
\end{itemize}

As we mentioned, the dummy centers are explicitly constructed by two
sampling routines: the first samples \emph{clients}
(\Cref{thm:per-client-sampling}) and the second samples \emph{centers}
(\Cref{thm:expensive}). Crucially, this sampling is effective because
$\OPT$ and $\LS$ nearly agree on most centers.  Combining this with
the submodular maximization ideas from \cite{2approx}, we can get that
for some carefully chosen parameter $\eta < 1$ and some
$D_{\kappa} \subseteq D$, we have
\begin{gather}
  \cost(\SS) \leq (1 + \eta)\cost_{\OPT}(D \setminus D_{\kappa}) + 2
  \cost_{\OPT}( D_{\kappa}) + O(\eps) \sopt.   \label{eq:mainblah} 
\end{gather}
I.e., for a ``difficult'' set of clients $D_\kappa$, we lose a factor
of $2$ on the cost they incur in $\OPT$, despite getting a
better-than-$2$ factor on the rest of the clients. The clients $D_{\kappa}$
are precisely the clients in clusters centered in $\OPT_R$ which are
close to the center (i.e., within its ``core''). If we were able to
sample these clients, they would give good dummy centers, so they must be
invisible to sampling. Since we can sample according to client costs
in $\LS$, these clients must have had exceptionally low cost in $\LS$.

This part of the algorithm follows the general approach of the
$(2+\eps)$-approximation~\cite{2approx}, but extends it in two ways:
(a)~we have to exercise care to get better-than-$2$-approximations on
most of the clients, and (b)~we use a color-coding idea to
substantially simplify a center-deletion step and the submodular
optimization portion. Nonetheless, we are left with the worry that
most of the clients may lie in $D_\kappa$ and incur a factor-$2$
loss.

\paragraph{Step 2: Local Search to the Rescue.} Our final novel
ingredient is to show that the local search solution $\LS$ actually
performs better-than-advertised on these clients in $D_\kappa$. We
extend the standard analysis giving a $(3+\eps)$-approximation to give
the following nuanced bound:
\begin{gather}
  \cost(\LS) \leq 3\cost_{\OPT}(D \setminus D_{\kappa}) + \cost_{\OPT}(D_{\kappa})
  + O(\eps) \sopt.  \label{eq:mainLS}
\end{gather}
I.e., we lose nothing on the ``difficult clients'' $D_\kappa$, while
losing the usual factor of three on the rest. The ideas behind this
analysis are the main technical innovation of the paper, and we
elaborate on these in the next section. Actually, to optimize the
approximation factor all the way down to $1.89$, we use a more
fine-grained form of this bound. Before jumping into that discussion,
let us remark that even this weaker local search cost bound yields a
better-than-two approximation: we can balance the two bounds from
\eqref{eq:mainblah} and \eqref{eq:mainLS} to get a
$\left(\frac{5 - \eta}{3 - \eta} + O(\eps)\right)$-approximation,
which is strictly less than two for $\eta < 1$.

\subsubsection{The Nuanced Local Search Guarantee}
\label{sec:LS}

Let us give a sense of our nuanced local search guarantee for
\emph{$\ell$-swap local search}, and how we use its properties in our
work. The algorithm starts from an arbitrary solution $\SS$ of $k$
centers and checks if there exist subsets $S\sse \SS, T \sse \FF$ with $T \cap \SS = \emptyset$ and $|S| = |T| \leq \ell$ so that the $(S,T)$-swap gives a cost reduction  $\cost((\SS \setminus S) \cup T) < \cost(\SS)$. If so, it updates $\SS \gets (\SS \setminus S) \cup T$, and continues. Arya et al.~\cite{arya2001local} show that any locally optimal solution for this process is a $(3+2/\ell)$-approximation. Using a standard idea of starting from some reasonable solution and taking a local swap only if it improves the cost by a $(1+\eps/k)$ factor ensures a $(3+2/\ell+\eps)$-approximation in polynomial time.

\setlength{\intextsep}{-4pt}%
\begin{wrapfigure}{R}{0.45\textwidth}
  \centering
  \includegraphics[width=0.35\textwidth]{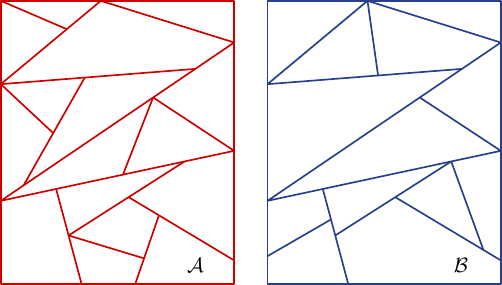}
  \caption{Cross-free clusterings: each pair of blue and red clusters
    are either disjoint, or one is contained within the other.}
  \label{fig:cross-free} 
\end{wrapfigure}

\mypara{A nuanced analysis.}
In our work we will require a more careful bound on the local-search
guarantee. We start by defining a condition for two
clusterings---which we should think of as the local search and optimal
clusterings---to be ``mutually compatible'': this is when no clusters
``cross''. 

\begin{definition}[Cross-Free
  Maps/Clusterings] \label{defn:cross-free} A pair of clusterings $\calA, \calB$ of the client
  set $D$ is called \emph{cross-free} if for any clusters
  $A \in \calA, B \in \calB$, either $A \cap B = \emptyset$, or else
  $A \sse B$ or $B \sse A$.  By extension, a pair of $k$-median
  solutions and assignments $(\AA, \mu_{\AA}),(\BB, \mu_\BB)$ is
  called cross-free if the corresponding clustering pair induced by
  them is cross-free.
\end{definition}

Now suppose that the optimal solution $\OPT$ and the local-search solution $\LS$
were miraculously cross-free (under the closest-facility
assignments $\pi_{\OPT}$ and $\pi_{\LS}$). We show that, surprisingly,
this structure is enough to show that the local search solution has
$\cost(\LS) \leq (1 + O(\eps)) \sopt$. Crucially for our application,
this cost bound can be generalized to $(1 + O(\eps)) \sopt$ plus a
term capturing the ``distance'' between the clustering induced by
$\pi_{\LS}$ and another assignment $\mu_{\AA} : D \to \LS$ ensuring
cross-freeness. We state the formal guarantees below.

\begin{theorem} \label{thm:simpler-ls} Let $(D \cup \FF, d)$ be an
  instance of $k$-median and $\eps > 0$. Let $\LS \subseteq \FF$ with
  $|\LS| = k$ be a solution for local search with $O(1/\eps)$-sized
  swaps, and let $\BB$ be another benchmark solution with $|\BB| =
  k$. Given arbitrary assignments $\mu_\AA: D \to \LS$ and
  $\mu_\BB: D \to \BB$ such that $(\LS, \mu_\AA), (\BB, \mu_\BB)$ is
  cross-free, the clustering induced by the local-search solution
  $(\LS, \pi_\LS)$ has cost
  \begin{equation}
     \cost(\LS) \leq \cost(\mu_\BB) + 
    (\cost(\mu_\AA) - \cost(\LS)) + O(\eps) (\cost(\mu_\AA)+\cost(\mu_\BB)).
  \end{equation}
\end{theorem}

Note that \Cref{thm:simpler-ls} bounds the cost of the local
search solution $\LS$ using the nearest-neighbor map $\pi_\LS$ in
terms of the cost of a clustering around $\LS$ given by another map
$\mu_\AA$, and an arbitrary benchmark solution $(\BB, \mu_\BB)$, as
long as the solution pair $(\LS, \mu_\AA), (\BB, \mu_\BB)$ is
cross-free. We defer its proof to
\S\ref{sec:local-revisited}.

It is illustrative to see how \Cref{thm:simpler-ls} %
implies a $(3 + \eps)$-approximation in general.
\begin{corollary}[(Weaker) Version of \protect{\cite[Theorem
    4.3]{arya2001local}}]
  \label{cor:3-approx} Let $(D \cup \FF, d)$ be an
  instance of $k$-median and $\eps \in (0, 1/200)$. Let $\LS \subseteq \FF$ be a
  local search solution with $|\LS| = k$ with $O(1/\eps)$ sized
  swaps. Then, $\cost(\LS) \leq (3 + O(\eps)) \sopt$.
\end{corollary}
\begin{proof}
  Set $\BB = \OPT$ and $\mu_\BB$ as the closest facility map
  $\pi_\OPT$. Set $\mu_\AA : D \to \LS$ where $\mu_\AA(p) = c$ if
  $p \in C^* \in \calO$ and $c = \arg \min_{c' \in \LS} d(c^*, c')$. That
  is, we assign the clients in cluster $C^*$ to the center $c \in \LS$
  which is closest to $c^*$; this ensures cross-freeness. Moreover,
  for any center $c^* \in \OPT$, its distance to the closest center in
  $\LS$ is $d(c^*, \LS) \leq \avg_{\OPT}(C^*) + \avg_{\LS}(C^*)$, by
  averaging and the triangle inequality. Summing over all clients, we
  get $\cost(\mu_{\AA}) \leq \cost(\LS) + 2 \sopt$. Finally, applying
  \Cref{thm:simpler-ls} gives us that
  \begin{align*}
    \cost(\LS) &\leq 3 \sopt + O(\eps)(\cost(\LS) + 2\sopt).
  \end{align*}
  Rearranging gives $ \cost(\LS) \leq (3 + O(\eps))\sopt$.
\end{proof}

This proof set $\mu_\BB$ to be the optimal map
and $\mu_\AA$ to be a natural map ensuring cross-freeness.  In
our proof, we will carefully hand-craft more elaborate maps
$\mu_\AA, \mu_\BB$ which give a $1$-approximation on the subset $D_\kappa$
of clients, to get the guarantee in~(\ref{eq:mainLS}). 

We will construct $\mu_\BB$ to be a near-optimal assignment to
$\OPT$. Then, certain clusters in the clustering $\calB$ induced by
$\mu_\BB$ will be (roughly speaking) partitioned into local search
clusters induced by the nearest-neighbor map $\pi_{\LS}$. The set of
clients in these \emph{partitioning} local search clusters is
$D_{\kappa}$. Moreover, the centers of these partitioning
local search clusters have the special property that they are not the
closest center in $\LS$ to non-partitioned $B \in \calB$. We can then
construct $\mu_\AA$ by setting $\mu_\AA(p) = \pi_{\LS}(p)$, for all
$p \in D_{\kappa}$, and assigning remaining non-partitioned
$B \in \calB$ consistently to good substitute local search clusters, à
la~\Cref{cor:3-approx}. This naturally results in cross-free maps with
$\cost(\mu_\AA) \leq 2 \cost_{\OPT}(D \setminus D_{\kappa}) +
\cost_{\LS}(D)$. Combining this bound with the near-optimality of $\mu_{\BB}$
and~\Cref{thm:simpler-ls} then gives (\ref{eq:mainLS}). The
discussion above is a slight simplification: the actual proof uses an
extension of \Cref{thm:simpler-ls} applied to a subset $D_0$ of
clients. 

\subsection{Related Work}
\label{sec:related-work}


The first constant-factor approximation algorithm for $k$-median was
due to Charikar, Guha, Tardos, and
Shmoys~\cite{CharikarGTSD99}. Shortly thereafter, Jain and Vazirani
introduced a primal-dual framework based on Lagrangian relaxation for
metric facility location and $k$-median~\cite{JV01}; subsequent
dual-fitting analyses of facility location produced improved
Lagrangian-multiplier-preserving (LMP) algorithms, which became a
standard route to $k$-median approximations~\cite{JMS02,JMMSV03}. In
parallel, Arya et al.\ analyzed local search, showing that local
optima under $p$-swaps achieve a $3+2/p$ approximation, and hence a
$(3+\varepsilon)$-approximation for  $p = O(1/\eps)$~\cite{arya2001local}.
Charikar and Li gave a dependent LP-rounding
scheme~\cite{Charikar-Li12} to get a constant-factor approximation.

LP-rounding and pseudo-approximations subsequently advanced the state
of the art: Li and Svensson introduced the pseudo-approximation
paradigm~\cite{LiS16} to break past the factor-of-$3$ barrier, and
Byrka et al.\ refined the bi-point rounding
framework~\cite{ByrkaPRST17}. More recent work has revisited both
major paradigms: Cohen-Addad et al.\ improved the local-search
analysis~\cite{CGHOS22}; Cohen-Addad, Grandoni, Lee, and
Schwiegelshohn breached the $2$-LMP barrier for facility
location~\cite{CGLS23}; and Gowda et al.\ further improved bi-point
rounding~\cite{GPST23}. As mentioned above, the best approximation
factor currently known for general metric $k$-median is due to another
breakthrough, getting a factor of $(2+\varepsilon)$~\cite{2approx}.
Thereafter, another $2+\varepsilon$-approximation was obtained by
Byrka et al.~\cite{ByrkaGHLWW26} using an iterative rounding
framework. There has been comparatively less progress on the
hardness-of-approximation, where the classical hardness threshold is
$1+2/e \approx 1.736$~\cite{GuhaK98,JMS02}, leaving a tantalizing gap
between the best known upper and lower bounds.


\subsection*{AI Acknowledgment and Concurrent Work}

The first result for stable instances obtained by the authors was a
$(1.95+\eps)$-approximation without AI assistance; it improved over
the $2$-approximation result from \cite{2approx} by using a nuanced
local search combined with a more careful sampling approach; the
second part of the algorithm using submodular maximization closely
followed the steps in \cite{2approx}. Giving this draft to Gemini 3.1
Pro led it to suggest the color-coding approach to simplify the second
part of the argument. Thereafter, GPT-6 Astra suggested the use of
$\olamearly$ and $\olamlate$, and cluster-dependent $\kappa(C^*)$
values to improve the approximation factor to its current value of
$1.89+\eps$. These models, as well as Claude Opus~4.6 and Gemini
Argon, were used for proof-checking and to obtain feedback on the
presentation.

In concurrent work announced on October 6, 2026, OpenAI presented
two different better-than-two approximations for general 
$k$-median instances~\cite{OAI:The-Approximation-Threshold-for-Metric-k-Median-September-24-2026,OAI:Single-Exponential-Recovery-and-Bounded-Price-Strictness-for-Metric-k-Median-September-24-2026}. The first extends the iterative rounding-based $2$-approximation of~\cite{ByrkaGHLWW26} to achieve an optimal $(1 + 2/e + \eps)$-approximation. The other instead builds on the stable instance/pseudo-approximation framework of~\cite{2approx}, thus sharing many similarities with this work and its companion paper. Although this latter algorithm only gives a $(2 -\sigma)$-approximation for general instances of $k$-median, for some absolute constant $\sigma > 0$, it does achieve a $(1 + 2/e + \eps)$-approximation algorithm for stable instances. 
Finally, through personal communication, we learned that Jaroslaw Byrka, Yuhao Guo, Yang Hu, Shi Li, Yongho Shin, Phanu Vajanopath, Chengzhang Wan and Zaixuan Wang have independently acheived a $2-\sigma$-approximation algorithm for a small but constant $\sigma >0$. Their approach further develops the iterative randomized rounding framework of~\cite{ByrkaGHLWW26}, providing an alternative route to a better-than-$2$ approximation.



%% file: prelims.tex
\section{Preliminaries} \label{sec:prelims}

\textbf{$k$-median.} Given an instance $(D \cup \FF,
d)$ for $k$-median, we want to find a subset of $k$ facilities/centers
$\SS \in \binom{\FF}{k}$ minimizing 
$\sum_{p\in D} d(p, \SS)$. Let $\OPT$ be an optimal choice for $\SS$, having cost $\sopt$. \margindef{$\OPT, \sopt$} 

\mypara{Assignments and clusterings.}
For any subset $\SS\subseteq \FF$ of open facilities, an
\emph{assignment} $\mu: D\rightarrow \SS$ is a map from clients to
these facilities. This assignment $\mu$ induces a clustering
$\calS_\mu = \{\mu^{-1}(c)\}_{c\in \SS}$, which is the partition of
$D$ induced by $\mu^{-1}$. We will interchangeably think of
$\calS_\mu$ as a partition $\{S_c\}_{c \in \SS}$ of the point set $D$
and as a collection $\{(S_c, c)\}$ of (set, center) pairs. In general,
we will use blackboard-bold letters such as $\AA, \BB$ to denote sets
of centers; $\mu_\AA, \mu_\BB$ to denote corresponding assignments;
and calligraphic symbols $\calA, \calB$ to denote the associated
clusterings. When the clustering is clear from context, we will
usually write lowercase $c$ (or $c'$ or $c^*$) to denote a facility
and uppercase $C$ (or $C'$ or $C^*$) to denote the cluster in the
clustering centered at that facility. When there is potential for
confusion, e.g., when $\calA$ and $\calB$ are both clusterings that
share a center $c$, we instead write $A_c$ (resp., $B_c$) to denote
the cluster in $\calA$ (resp., $\calB$) centered at $c$. Given a set
of centers $\SS$, one canonical assignment is the \emph{projection} or
\emph{nearest-center assignment}, which is denoted by $\pi_\SS$; i.e.,
$\pi_\SS(p) \coloneqq \arg\min_{c \in \SS} d(p,c)$. \margindef{$\pi_\SS(\cdot)$} (W.l.o.g., we can
assume that distances are distinct and hence the minimizer is unique.)

\mypara{Stability.}  An instance $(D \cup \FF, d)$ of $k$-median is
called \emph{$\beta$-stable} if
$\sopt_{k-1} \geq (1 + \beta) \sopt_k$.  In $\beta$-stable instances,
the clustering produced by local search (even with unit-sized swaps)
approximates most clusters in the optimum solution extremely well. %
We explore the direct consequences of stability in more detail in the
proof of our main structural theorem (\S\ref{sec:structural-thm});
e.g., a trivial consequence of stability is that each optimal cluster
is non-empty.

\mypara{Dummies.}
We augment the metric space with auxiliary facilities $\Lambda$ during the algorithm: i.e., the new set of 
facilities is $\FF'=\FF\cup \Lambda$. For each such new center $\lambda \in \Lambda$, there
is a ``leader'' $\ell_{\lambda} \in D \cup \FF$ and a distance $r_\lambda$; we
explicitly set the distance
$d(\ell_{\lambda}, \lambda)=r_\lambda$, and the remaining
distances to the new center $\lambda$ are induced by the
triangle inequality: 
$d(q, \lambda) \coloneqq d(q, \ell_{\lambda}) + r_\lambda$
and $d(\lambda, \lambda') \coloneqq r_{\lambda} + d(\ell_{\lambda}, \ell_{\lambda'}) + r_{\lambda'}$.

\mypara{Assignment cost.} 
For any assignment $\mu: D \to \FF \cup \Lambda$ and a subset of clients $D' \subseteq D$, define \margindef{$\cost$}
\[ \textstyle
\cost_{\mu}(D') := \sum_{p \in D'} d(p, \mu(p)).\]
The average cost of $D' \sse D$ is \margindef{$\avg$}
\( \textstyle \avg_{\mu}(D') := \frac{\cost_{\mu}(D')}{|D'|},\) whenever
$D' \neq \emptyset.$ We usually consider $\avg_{\mu}(D')$ where
$D' = C$ is a cluster in some clustering. 
When the assignment is via the projection map (closest-facility map) $\pi_\SS$ of
$\SS$, we write $\cost_{\SS}(D')$ (and similarly
$\avg_{\SS}(D')$). We denote $\cost_{\mu}(D)$ as $\cost(\mu)$; for
$\mu$ being the projection assignment $\pi_{\SS}$, we denote
$\cost_{\mu}(D)$ as $\cost(\SS)$.

\mypara{Percentile Costs.} Let $C^*$ be a non-empty optimal cluster centered at $c^* \in \OPT$, and let $\eps \in [0,1)$. We define
$\per(C^*)$ \margindef{$\per$} as the \emph{$\eps$-percentile
  distance in $C^*$}: at least $\eps |C^*|$ clients in $C^*$ have
$d(p, c^*) \leq \per(C^*)$ and at least $(1 - \eps) |C^*|$
clients have $d(p, c^*) \geq \per(C^*)$. Notice that
$(1 - \eps)|C^*| \per(C^*) \leq |C^*|\avg_{\OPT}(C^*)$ and so
$\per(C^*) \leq \frac{1}{1- \eps}\avg_{\OPT}(C^*)$. When
$\eps \leq \nf12$, this implies
\begin{equation}
\per(C^*) \leq (1 + 2\eps)\avg_{\OPT}(C^*). \label{eq:per-bd} 
\end{equation}

\mypara{Cluster Cores.} Given the optimal clustering $\calO$, a
cluster $C^*$ centered at $c^*$, and a function
$\kappa: \calO \to [0,1]$, define the core \margindef{core} of a
cluster $C^*$ to be 
\begin{equation*}
  \core_\kappa(C^*) := \{p \in C^* \mid d(p, c^*) \leq (1 + \newkappa) \per(C^*)\},
\end{equation*}
i.e., the points that are at most $(1+\newkappa)$ times the
$\eps$-percentile distance of $C^*$. We may also specify a constant
function (e.g., a scalar value like $0$) for $\kappa$.


\mypara{Parameters and constants.}
We use the following parameters throughout the argument.
\begin{itemize}[nosep]
    \item $\eps \in (0,1/200)$: the accuracy/error parameter.
    \item $\beta:$ the stability parameter of the instance.
    \item $s_0 := \eps^{-2} \beta^{-1}$: \margindef{$s_0$} the bound on the number of non-synchronized optimal clusters (\Cref{lem:few-not-in-sync}). This also bounds the sizes of various ``special'' sets of centers.
    \item $\eta := 89/100$: \margindef{$\eta$} the separation/closeness threshold used as both a sufficient condition for substitution and establishing the far-from-$\LS$ condition of~\Cref{thm:structural}. Note that $\nf1\eta - 1 \geq 6\eps$.
\end{itemize}


%% file: outline.tex
\section{Deducing our Main Result}
\label{sec:deducing-main}

Our main result is the following better-than-$2$ approximation for
$\beta$-stable instances of $k$-median.

\mainResult*

Fix $\eps > 0$. The starting point for our algorithm is a local search solution
$\LL \subseteq \FF$ with $|\LL| = k$ arising from 
$O(1/\eps)$-sized swaps; let $\calL$ denote the clustering induced by
the nearest-center map $\pi_{\LL}$. Let $\OPT \subseteq \FF$ be an
optimum solution with $|\OPT| = k$, and let the nearest-center map
$\pi_{\OPT}$ induce the corresponding clustering $\calO$.
We define some notation to simplify terms in the following discussion:
for a client $p$, let $C^*$ be the optimal cluster containing it, and
let
\begin{gather}
  \err_p := \avg_\OPT(C^*) + \avg_\LS(C^*). \label{def:error}
\end{gather}
Since the local search solution $\LS$ is a constant-factor
approximation, we get that $\sum_p \err_p = O(\sopt)$, and hence
errors of the form $\eps \err_p$ only contribute $O(\eps)$ terms to
the approximation.

\input{supporting-sections/structural-main-result}

\input{supporting-sections/color-coding}

\input{supporting-sections/cost-of-local-search}

\subsection{Win-win Argument and Final Approximation Guarantee}
\label{sec:wrapup}

To conclude \Cref{thm:main-result}, we combine this local-search bound
with the submodular maximization bound from
\Cref{thm:simple-submod}.  To optimize our final approximation factor, it will be convenient to have cluster-wise bounds; to that end, for each cluster $C^* \in \calO$, define
\[
x(C^*) :=\frac{\per(C^*)}{\avg_\OPT(C^*)},
\qquad
y_\kappa(C^*) :=
\frac{\cost_\OPT(\core_\kappa(C^*))
      +|\core_\kappa(C^*)|\per(C^*)}{\cost_\OPT(C^*)}.
\]
The parameter $y_\kappa(C^*)$ captures the savings from the core in the local-search
bound, relative to $\cost_\OPT(C^*)$. With this notation, the two
aggregate cost bounds from \Cref{thm:simple-submod} and \Cref{thm:applied-local-search-bd} have the following
coefficients of $\cost_\OPT(C^*)$, where we combine $\OPT_P$ with $\OPT_{\Lambda}^{\mathrm{late}}$ for simplicity:
\\
\begin{table}[htbp]
\centering
\(
\begin{array}{c|c|c}
 &\LS&\TT\\ \hline
\OPT_S&1+\eta&1+\eta\\[4pt]
\olamearly&3&1+2x(C^*)/e\\[4pt]
\olamlate\cup\OPT_P&1+x(C^*)&1+2(1+\newkappa)x(C^*)/e\\[4pt]
\OPT_R&2+x(C^*)-y_\kappa(C^*)&1+x(C^*)
\end{array}
\)
\caption{Comparison of local search and submodular maximization cluster costs}
\label{tab:winwin-bounds}
\end{table}
\\ \\
The table hides error terms totaling
$O(\eps)\sopt$ in both columns. Balancing the costs from the two columns, we give weight $a:=3/25$ to the former and
$1-a$ to the latter.
The following lemma shows that, with this choice of $a$, we can define $\newkappa$ as required for \Cref{thm:structural} and achieve an overall $1 + \eta +O(\eps)$ approximation guarantee.
\begin{restatable}{lemma}{winwin}\label{lem:winwin-numerical}
Let $a=3/25$, $\eta=89/100$, and $\eps\in(0,1/200)$.
For each optimal cluster $C^*$, there exists
$\newkappa\in[1/\eta-1,1/2]$ such that
\eqref{eq:radius-diameter-condition} holds whenever
$\per(C^*)\geq\eta\avg_\OPT(C^*)/(1+\eps)$, and the $(a,1-a)$ convex combination of the two entries in
the row of \Cref{tab:winwin-bounds} containing $c^*$
is at most $1+\eta+O(\eps)$.
\end{restatable}
The proof of \Cref{lem:winwin-numerical} appears in
\Cref{sec:winwin-technical-proof}.
Applying the lemma and summing over the clusters $C^* \in \calO$ gives
\[
\min\{\cost(\LS),\cost(\TT)\}
\leq a\cost(\LS)+(1-a)\cost(\TT)
\leq\left(1 + \eta +O(\eps)\right)\sopt.
\]
We return the cheapest solution among $\LS$ and the candidates
produced by the algorithm. Amplifying the randomized routines,
rescaling $\eps$, and using $\eta = 0.89$ gives a $(1.89+\eps)$-approximation
with probability at least $1-\eps$, in time
$\exp\{O(\beta^{-1}\poly(1/\eps))\}\poly_{\eps}(n)$.

In the following sections, we give the details missing from above. We
prove the main structural results \Cref{thm:structural} and \Cref{thm:struct2}, showing the
existence of many substitutions, backups, and dummies in
\Cref{sec:structural-thm}. Then we present the submodularity-based algorithm and prove \Cref{thm:simple-submod}
in \Cref{sec:submodular-section} and conclude with the nuanced local search
guarantee of \Cref{thm:local-search-bd} in \Cref{sec:local-revisited}.


%% file: supporting-sections/structural-main-result.tex
\subsection{Structural Theorem}
\label{sec:structural-theorem}

\begin{figure}[t]
  \centering
  \includegraphics[width=3in]{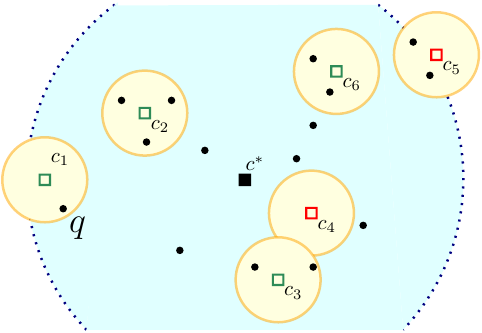}
  \caption{The dotted circle shows the ball of radius
    $(1 + \newkappa + 2\eps)\per(C^*)$ around $c^*$; the dots are clients
    in $C^*$, and the orange circles have radius $\eps\,\per(C^*)$. The green centers are the satellites $\Gamma(c^*)$; 
    $c_5$ is too far to be a satellite, and $c_4$ does not have any
    witness clients from $C^*$.}
  \label{fig:sat-def}
\end{figure}

We begin with a structural theorem giving low-cost substitutions of
centers in $\OPT$ by centers in $\LL$, as well as dummy and
backup centers supporting other centers in $\OPT$. We state the result
in two parts. For the first part, 
define a non-negative ``core-radius'' $\newkappa \in [\nf1\eta - 1,1)$ \margindef{$\kappa(\cdot)$} for each optimal cluster $C^*$. We choose $\newkappa$ so that, whenever
$\per(C^*)\geq\eta\avg_\OPT(C^*)/(1+\eps)$, i.e., $\per(C^*)$ is quite large relative to $\avg_{\OPT}(C^*)$, we have 
\begin{equation}
\label{eq:radius-diameter-condition}
(2+\newkappa)\per(C^*)
\geq(2+6\eps)\avg_\OPT(C^*).
\end{equation}
This bound will be useful in simplifying some of our structural
conditions. (The quantity $\kappa(C^*)$ depends on the unknown optimal
clustering, and we use it only for the analysis.) Also define the set of \emph{satellites}
for any center $c^* \in \OPT$ with corresponding cluster
$C^* \in \calO$ to be
\begin{equation}
  \Gamma(c^*) := \Big\{c \in \LS \mid d(c, c^*) \leq (1 + \newkappa +
  2\eps)\per(C^*), d(c, C^*) \leq 
  \eps\,\per(C^*) \Big\}. 
  \label{eq:Gamma}
\end{equation}
Informally, these are local search centers $c \in \LS$ which are
``close'' to $c^*$, and moreover, which have some ``witness'' client
from the corresponding optimal cluster $C^*$ close to them. (See~\Cref{fig:sat-def}.)

\begin{restatable}{thm}{structural}
  \label{thm:structural} Let $(D \cup \FF, d)$ be a $\beta$-stable
  instance of $k$-median and $\eps \in (0, 1/200)$. There exists an
  algorithm running in $\poly(n)$
  time which, with probability $\exp\{-O(s_0 \log(\eps^{-1})) \}$, outputs a subset of the
  local search solution $\LS' \subseteq \LS$ and a collection of
  \emph{dummy} centers $\Lambda$, such that there exists a partition
  $\OPT = \OPT_\Lambda \sqcup \OPT_P \sqcup \OPT_R \sqcup \OPT_{S} $
  with the following properties.
  \begin{enumerate}[label = (\alph*), nosep]
  \item \emph{(All special sets are small)} We have $|\OPT_{\Lambda}|
    + |\OPT_P| + |\OPT_R| \leq s_0$.\label{cond:struct-special-small} 
  \item \emph{(Centers with corresponding dummy centers)} \label{cond:struct-dummies} There
    exists a partition $\OPT_{\Lambda} = \olamearly \sqcup \olamlate$ and a dummy center map
    $\varphi : \OPT_{\Lambda} \to \Lambda$ such that $\varphi$ is a
    bijection, and for each center $c^* \in \OPT_{\Lambda}$ and client
  $p \in C^* \in \calO$:
    \begin{align*}
      d(c^*, \ell_{\varphi(c^*)})
      &\leq d(\ell_{\varphi(c^*)},
        \varphi(c^*)) = r_{\varphi(c^*)} \\
      &\leq (1 +
        \newkappa\,\mathbf{1}_{(c^* \in \olamlate)}) \per(C^*) + O(\eps \err_p).
        \end{align*}
  \item \emph{(Centers with backups)}  \label{cond:struct-backups} There exists $\LS_P \subseteq \LS'$ and a surjective map
    $\psi : \OPT_P \to \LS_P$ such that, for each $c^* \in \OPT_P$ and $p \in C^* \in \calO$:
    \begin{equation*}
      d(c^*, \psi(c^*)) \leq (1 + \newkappa) \per(C^*) + O(\eps \err_p).
    \end{equation*} 

  \item \emph{(Substituted centers)} \label{cond:struct-subs} There exists a substitution map
    $\sigma : \OPT_S \sqcup \OPT_R \to \LS'$ such that $\LS_P
    \subseteq \LS_\sigma := \sigma(\OPT_S \sqcup \OPT_R) \subseteq
    \LS'$, and
    \begin{enumerate}[label = (\roman*), nosep]
    \item Replacing centers in $\OPT_S$ with centers in $\LS$ according
      to the substitution map $\sigma$ increases the cost by at most an 
      approximately $(1 + \eta)$ factor: 
      \begin{equation*}
        \sum_{c^* \in \OPT_S} \sum_{p \in C^*} d(p, \sigma(c^*)) \leq (1 + \eta) \sum_{c^* \in \OPT_S} \cost_{\OPT}(C^*) + O(\eps) \sopt.
      \end{equation*} \label{cond:struct-substitution-strong}
    \item
      \label{cond:struct-substitution-weak}
For every $c^*\in\OPT_R$, we have
$\sigma(c^*)\in\Gamma(c^*)$ and
$d(c^*,\sigma(c^*))\leq(1+\eps)\per(C^*)$.
Consequently,
\[
\sum_{c^*\in\OPT_R}\sum_{p\in C^*}d(p,\sigma(c^*))
\leq\sum_{c^*\in\OPT_R}
\bigl(\cost_\OPT(C^*)+|C^*|\per(C^*)\bigr)
+O(\eps)\sopt.
\]
    \item \label{item:unique-preimage} For all $c^* \in \OPT_R$, we have $\Gamma(c^*) \subseteq \LS'$, and 
      $\sigma(c^*) \neq \sigma(c^{**})$ for any $c^{**} \in (\OPT_R
      \cup \OPT_S) \setminus \{c^*\}$.
    \end{enumerate}
  \end{enumerate}
\end{restatable}
Running the algorithm of \Cref{thm:structural} a total of $\exp\{O(s_0
\log{\eps^{-1}}) \}$ times ensures that at least one output has all
desired guarantees with probability $(1 - \eps)$. It will be
convenient to have notation for the partition of $D$ induced by the
partition of $\OPT$ from
\Cref{thm:structural}. Henceforth,\margindef{$D_\Lambda, D_P$}
\margindef{$D_R, D_S$} define client sets $D_\Lambda, D_P, D_R$, and
$D_S$ containing the clients assigned to $\OPT_\Lambda$, $\OPT_P$,
$\OPT_R$, and $\OPT_S$ respectively in the optimal
assignment. Moreover, call a client $p$ \emph{late} if $\pi_\OPT(p)
\in \olamlate$, and define the natural partition $D_\Lambda =
D^{\early}_\Lambda \cup D^{\late}_\Lambda$.

\begin{figure}[t]
  \centering
  \includegraphics[width=4in,page=1]{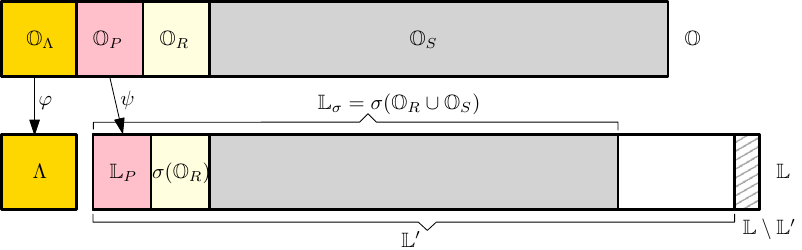}
  \caption{Maps from \Cref{thm:structural}: Substitution map $\sigma$
    has image $\LS_\sigma := \sigma(\OPT_S \sqcup \OPT_R)$ with
    $\LS_P \sse \LS_\sigma$, backup map $\psi$ maps $\OPT_P$ to
    $\LS_P$, dummy center map $\varphi$ is a bijection between
    $\OPT_\Lambda$ and dummy centers $\Lambda$. Note: we have the
    inclusions $\LS_P \sse \LS_{\sigma} \sse \LS' \sse \LS$, and the
    set $\LS \setminus \LS'$ is discarded immediately. }
  \label{fig:structural-maps}
\end{figure}

\Cref{fig:structural-maps} illustrates the partition of $\OPT$
together with the substitution map $\sigma$, the backup map $\psi$,
and the dummy center map $\varphi$ produced by \Cref{thm:structural}. The proof of \Cref{thm:structural} is deferred
to~\Cref{sec:structural-thm}, and goes as follows. 

\textbf{Step I:} Using $\beta$-stability
and~\Cref{lem:few-not-in-sync}, most clusters in $\OPT$ can be
``synced'' with clusters in the local search solution $\LS$, allowing
for very efficient substitution.  Some other cluster centers $c^*$ in
$\OPT$ can be substituted by nearby centers $c$ in $\LS$ (at distance
less than their average cluster distance).

\textbf{Step II:} Next, we find dummy centers for many of the
remaining $\OPT$ clusters by sampling clients according to their costs
in $\LS$. (This sampling procedure yields many additional structural
guarantees, which allow us to do some more ``clean-up''
substitutions.) We analyze this client sampling in two phases. The first phase seeks only the
best possible dummy centers. The second round finds additional dummies
with worse but still acceptable guarantees. This is precisely where the distinction between $\OPT_{\Lambda}^{\mathrm{early}}$ and $\OPT_{\Lambda}^{\mathrm{late}}$ arises from: $\OPT_{\Lambda}^{\mathrm{early}}$ are the centers for which we find dummies when optimistically sampling clients. 

\textbf{Step III:} We then do a second sampling step, this time
sampling \emph{cluster centers} according to the cluster costs in
$\LS$, and then according to their costs in a solution with some
redundant centers from $\LS$ removed and some dummy centers
added. (The new solution resulting from redundant center removal is
the set $\LS' \subseteq \LS$ guaranteed by \Cref{thm:structural}.)
During this step, we find some additional dummy centers and show the
existence of backup centers for other optimum clusters. These will %
help us control costs during the submodular maximization subroutine
used later in the algorithm of \Cref{thm:simple-submod}. A final round
of ``clean-up'' substitutions completes the proof. The only
algorithmic steps here are the two sampling procedures.

Let us also record some more structural properties of the output from
the above theorem.
\begin{restatable}{thm}{structuraltwo} \label{thm:struct2}
  Successful runs of the algorithm from~\Cref{thm:structural} also ensure the following properties:
\begin{enumerate}[label = (\alph*)]
   \item \emph{(Redundant centers are
      hidden in the mixed solution)} \label{cond:struct2-redund-hidden} Define $\MM' := \LS'
    \cup \Lambda$ and $\calM'$ as the clustering corresponding to
    $\pi_{\MM'}$. Then,
    \begin{equation*}
      \sum_{c \in \Gamma(\OPT_R) \cup (\LS' \setminus \LS_{\sigma})} \cost_{\MM'}(M_c) < \eps^3 \sopt,
    \end{equation*}
    where $M_c$ is the cluster in $\calM'$ centered at $c$.
    \item \emph{(Far from $\LS$)} For all $c^* \not\in \OPT_S$ and $p \in C^* \in \calO$,  $d(c^*,
      \LS) > \eta \avg_{\OPT}(C^*) + 2\eps \avg_{\LS}(C^*)$. \label{cond:struct2-not-close-new} 
          \item  \emph{(Hidden cores)} $\sum_{c^* \in \OPT_P \cup \olamlate}
      \cost_{\LS}(\core_{6\eps}(C^*)) < \eps^5
      \sopt$. \label{cond:struct2-hidden-core-gen} 
    \item \emph{(Remaining
      centers)} \label{cond:struct2-props-remaining} The set of remaining centers $\OPT_R$ satisfies the following.
    \begin{enumerate}[label = (\roman*), nosep]
    \item  \emph{($\OPT_R$ Hidden cores)} $\sum_{c^* \in \OPT_R}
      \cost_{\LS}(\core_\kappa(C^*)) < \eps^5
      \sopt$, and, for all $c^* \in \OPT_R$, we have
      $\cost_{\LS}(\core_\kappa(C^*)) < \eps^3
      \cost_{\OPT}(C^*).$ \label{cond:struct2-hidden-core-R} 
    \item \emph{(Satellites are hidden in $\calL$)} 
      $ \sum_{c\in \Gamma(\OPT_R)} \cost_{\LS}(C)
      < \eps^3 \sopt.$ \label{cond:struct2-hidden-satellites} 
    \item \emph{($\per(C^*) \approx \avg_{\OPT}(C^*)$)}  For all $c^*
      \in \OPT_R$, $\frac{1}{1 + 2\eps} \per(C^*) \leq
      \avg_{\OPT}(C^*) \leq \frac{(1 + \eps)}{\eta} \per(C^*)$. 
      \label{cond:struct2-per-approx-avg}
    \end{enumerate}
    \end{enumerate}
\end{restatable}

\subsection{The Dangerous Set $D_\kappa$}
\label{sec:d_kappa}

After applying~\Cref{thm:structural}, each center in $\OPT$ has either
a corresponding dummy center in $\Lambda$, a backup center in $\LS_P$,
or a substitute center in $\LS_{\sigma}$. Most of these will incur a
cost less than a factor of 2; the one exception will be the clients in
the set \margindef{$D_\kappa$}
\begin{gather}
D_\kappa:=\bigcup_{c^*\in\OPT_R}\core_\kappa(C^*) \subseteq D_R. 
\end{gather}
This ``dangerous set'' will be an important character in the rest of
the discussion. The following claim gives a sense of why $D_\kappa$
arises naturally:
\begin{claim}
  \label{clm:dangerous}
  The cost of sending clients in $D_R$ to their substitutes is:
  \begin{gather}
    \sum_{c^*\in\OPT_R}\sum_{p\in C^*}d(p,\sigma(c^*))
    \leq 2\cost_\OPT(D_\kappa)
    +(1+\eta)\cost_\OPT(D_R\setminus D_\kappa)
    +O(\eps)\sopt.
  \end{gather}
\end{claim}

\begin{proof}
  Using 
  \Cref{thm:structural}\ref{cond:struct-subs}\ref{cond:struct-substitution-weak},
  and splitting the summation over clients in $D_\kappa$ and $D_R
  \setminus D_\kappa$ gives:
  \begin{align}
    \sum_{c^*\in\OPT_R}\sum_{p\in C^*}d(p,\sigma(c^*))
    &\leq \cost_\OPT(D_\kappa)  + \sum_{c^*\in\OPT_R}
      |\core_\kappa(C^*)| \per(C^*) \notag \\
    &~~~~+ \cost_\OPT(D_R \setminus D_\kappa)
    +
    \sum_{c^*\in\OPT_R}
    |C^* \setminus D_\kappa| \per(C^*)
      +O(\eps)\sopt \label{eq:dangermouse}
  \end{align}
  To bound the $D_\kappa$ terms in~(\ref{eq:dangermouse}), consider any cluster $C^*$. The
  definition of percentile implies that
  \[
    |\core_\kappa(C^*)|\per(C^*) \leq \cost_\OPT(\core_\kappa(C^*))
    +\eps|C^*|\per(C^*),
  \]
  since all but $\eps|C^*|$ of the clients from the core are at
  distances at least $\per(C^*)$ from $c^*$. Finally, the
  $\eps|C^*|\per(C^*)$ terms summed over all centers in $\OPT_R$
  contribute another $O(\eps) \sopt$, so 
  \begin{gather}
    \cost_\OPT(D_\kappa)  + \sum_{c^*\in\OPT_R}
    |\core_\kappa(C^*)| \per(C^*) \leq 2\cost_\OPT(D_\kappa)  +
    O(\eps)\sopt. \label{eq:dangermouse-1}
  \end{gather}
  To bound the non-$D_\kappa$ terms in (\ref{eq:dangermouse}), 
fix $c^* \in \OPT_R$
  and $p \in C^* \setminus D_{\kappa}$. Since
  $p \not\in \core_{\kappa}(C^*)$,
  $\per(C^*) \leq \frac{d(p, c^*)}{1 + \newkappa}$ and hence
  \begin{equation}
    d(p, c^*) + \per(C^*) \leq \frac{(2 + \newkappa)d(p, c^*)}{1 + \newkappa} \leq (1 + \eta) d(p, c^*), \label{eq:dangermouse-2}
  \end{equation}
  with the last inequality following from
  $\newkappa \geq \nf1\eta - 1$. Summing~(\ref{eq:dangermouse-2})
  over $p \in D_R \setminus D_{\kappa}$, and substituting
  into~(\ref{eq:dangermouse}) along with (\ref{eq:dangermouse-1})
  completes the proof.
\end{proof}


In the rest of the paper, we pursue a win-win analysis: (i)~If the
cost contribution of clients in $D_\kappa$ is a small fraction of
$\sopt$, we can compensate for this factor-of-$2$ loss by the smaller
loss on the remaining clients. (ii)~On the other hand, if their cost
contribution is large, we show that local search already gives an
excellent overall approximation.


%% file: supporting-sections/color-coding.tex
\subsection{Color-Coding and Submodular Maximization}
\label{sec:color-coding-submod}

\Cref{thm:structural} and \Cref{clm:dangerous} show the \emph{existence} of a substitution map
$\sigma$, a backup map $\psi$, and an auxiliary center map $\varphi$
that loses a factor of $2$ on $D_\kappa$, but less on the
rest. However, we are not yet done, because: (a)~$\MM' := \LS' \cup
\Lambda$ \margindef{$\MM'$} may have more than $k$ centers, and
(b)~the dummy centers $\Lambda$ are imaginary and need to be replaced
by real centers in $\FF$, plus we need to add some more optimal
centers corresponding to the backups. We do this via color-coding and
submodular maximization.

To begin, we define \margindef{$z$}
$z := (|\LS'| + |\OPT_\Lambda \cup \OPT_P| - k)_+$ as the ``excess''
number of centers we would have if we were to add in
$\OPT_\Lambda \cup \OPT_P$ to replace the dummies and account for the
backups. We choose an arbitrary ``deletion'' set \margindef{$\ZZ^*$}
$\ZZ^* \sse \LS' \setminus \LS_\sigma$ of size $|\ZZ^*| = z$. (To see
that such a set exists, note that
$|\LS_\sigma| = |\sigma(\OPT_R \cup \OPT_S)| \leq |\OPT_R \cup \OPT_S| =
k - |\OPT_\Lambda \cup \OPT_P|$, and hence $\LS'\setminus \LS_\sigma$
has size at least $z$.)

\begin{figure}[t]
  \centering
  \includegraphics[width=4in,page=2]{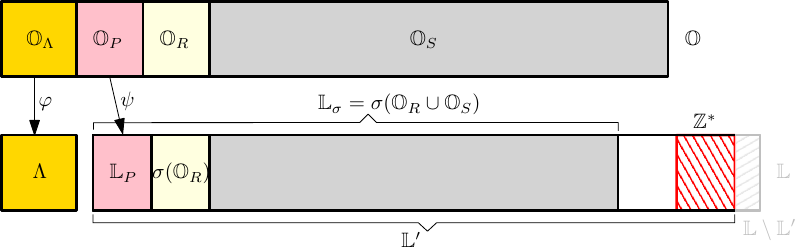}
  \caption{The set $\ZZ^*$.}
  \label{fig:structural-maps-2}
\end{figure}

Define the ``target'' solution to be \margindef{$\MM^*$}
$\MM^* := (\LS' \setminus \ZZ^*) \cup \OPT_P \cup \OPT_\Lambda$  and $\MM^*_{\Lambda}$ its counterpart with $\OPT_P \cup \OPT_\Lambda$ replaced by $\Lambda$. 
However, we do not know the ideal deletion set $\ZZ^*$, nor the set
$\OPT_P \cup \OPT_\Lambda$ we would like to add, so we try to find
(an approximation to these) via two tools: (a) color-coding, and
(b)~submodular maximization.

\textbf{Color-Coding.}
First, we color the centers in $\LS'$ red or blue uniformly at
random. Define the good event $\calG$ to be the following:
\begin{enumerate}[nosep]
\item centers in the deletion set $\ZZ^*$ are all colored red, and
\item centers in $\pi_{\LS'
    \setminus \ZZ^*}(\ZZ^*)$ are all colored blue.
\end{enumerate}
Since each of the sets in the definition of $\calG$ is of size at most
$s_0$, we have that $\Pr[\calG] \geq 2^{-2s_0}$.  We can just repeat
the above coloring step $\exp(O(s_0))$ times and take the best
solution among all the runs. So we assume henceforth that event
$\calG$ happens; let $\LSred$ and $\LSblue$ be the red and blue
centers in $\LS'$. The advantage of this coloring is that we can now
restrict the search for the deletion set $\ZZ^*$ to within the red
centers, and we will be sure of maintaining good reassignments for the clients in clusters centered in $\ZZ^*$
(which we have forced to be colored blue).

\textbf{Submodular Maximization.} Next, we define a monotone
non-increasing \emph{supermodular} function $g := g_1 + g_2$, where
\begin{enumerate}[label = (\roman*),nosep]
\item
  $g_1(\YY) := \cost(\MM' \cup \YY)$, and 
\item
  $g_2(\YY) = \min_{\ZZ \subseteq \LSred : |\ZZ| = z} \sum_{c
    \in \ZZ} [\cost_{\MM'}(M_c) + |M_c|\, d(c, \LSblue \cup \YY \cup \Lambda)]$.
\end{enumerate}
Observe: the set $\ZZ^*$ is a candidate in the above minimization.
These functions $g_1$ and $g_2$ are designed to capture the costs of
clients assigned in $\calM'$ to the non-deleted centers from
$\LS' \setminus \ZZ^*$, and to the deleted centers from $\ZZ^*$,
respectively. (Some points are counted in both terms, and we need to
control this excess.)

We now define the submodular ``improvement'' function
$f(\YY) := g(\emptyset) - g(\YY)$, and apply the
$(1 - \nf{1}{e})$-approximation algorithm for submodular maximization
subject to matroid constraints
\cite{calinescu2011maximizing}. Specifically, we cast constraints of
finding a set $\YY$ of size $|\OPT_P| + |\OPT_\Lambda|$, subject to
opening at least one center in the ball
$B(\ell_{\lambda}, r_{\lambda})$ for each $\lambda \in \Lambda$, as
that of finding a base subject to a partition matroid constraint.
We then take the returned solution $\widehat\YY$, and the set
$\widehat\ZZ$ that achieves the minimizer for $g_2(\widehat\YY)$, and
return the $k$-median solution:
\[ \TT := (\LS' \setminus \widehat\ZZ) \cup \widehat\YY. \] Note that, although we do not know $|\OPT_\Lambda| + |\OPT_P|$ a priori, we can easily enumerate all possibilities of $O(s_0)$. To complete this
part of the argument, we relate the cost of solution $\TT$ to
$g(\widehat\YY)$ to show:

\begin{restatable}{thm}{simplesubmod} \label{thm:simple-submod}
  Conditioned on the event $\calG$ (which happens with probability
  $\exp\{-O(\eps^{-2} \beta^{-1})\}$), the solution $\TT \subseteq \FF$
  with $|\TT| = k$ satisfies
  \begin{align}
    \cost(\TT)
    &\leq (1 - \nf1e) \cost(\MM^*) + \nf1e\,
      \cost(\MM^*_{\Lambda}) + O(\eps) \sopt. \label{eq:first-part}
  \end{align}
  Moreover,
\begin{align*}
\cost(\TT)\leq{}&
(1+\eta)\cost_\OPT(D_S)+\cost_\OPT(D\setminus D_S)\\
&+\frac2e\sum_{c^*\in\olamearly}|C^*|\per(C^*)
+\frac2e\sum_{c^*\in\olamlate}(1+\newkappa)|C^*|\per(C^*)\\
&+\frac1e\sum_{c^*\in\OPT_P}(1+\newkappa)|C^*|\per(C^*)
+\sum_{c^*\in\OPT_R}|C^*|\per(C^*)
+O(\eps)\sopt.
\end{align*}
\end{restatable}
We prove this result in \S\ref{sec:submodular-section}.

%% file: supporting-sections/cost-of-local-search.tex
\subsection{The Cost of Local Search}
\label{sec:cost-local-search}

To show part (ii) of the win-win approach, that local-search is a good
approximation for the client set $D_\kappa$, we use an extension of
our nuanced local-search analysis from \Cref{thm:simpler-ls}
based on cross-free pairings. In particular, the cross-free pairing
used for the simple $3$-approximation proof in \Cref{cor:3-approx} was
quite simple, forcing each cluster in $\calB$ to be contained in a
cluster in $\calA$. While we can show a good cost bound for just $D_{\kappa}$ with only \Cref{thm:simpler-ls} and more careful definitions for both $\mu_\AA$ and $\mu_\BB$, to simultaneously preserve our substitution cost bounds we need
a slight extension of cross-freeness. We
now define the notion of cross-freeness \emph{with respect to some
  $\lsd \subseteq D$}, and present the extension of
\Cref{thm:simpler-ls} to this setting of generalized cross-freeness.

\bigskip
\begin{figure}[H]
  \centering
  \includegraphics[width=5in]{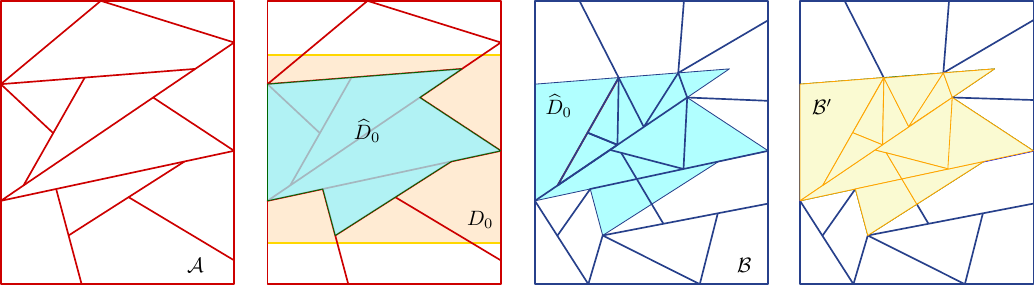}
  \caption{A clustering $\calA$ of $D$ in red; a possible set $\lsd$,
    and the set $\lsdres$, with the clusters $\calA'$ being the
    clusters now colored blue. Another clustering $\calB$ in blue,
    such that $\calA, \calB$ are not cross-free in general. However,
    they are cross-free w.r.t.\ $\lsd$, since each cluster in the
    clustering $\calB'$ obtained by restricting $\calB$ to $\lsdres$
    (given in goldenrod on the right) either contains, or is contained
    in, a cluster of $\calA'$.}
  \label{fig:gen-cross-free}
\end{figure}

\begin{restatable}[Generalized
  Cross-Freeness]{definition}{GenCF} \label{defn:gen-cross-free} Given
  a pair of clusterings $\calA, \calB$ of $D$, and a subset
  $\lsd \subseteq D$, define
  $\calA' := \{A \in \calA \mid A \subseteq \lsd\}$ to be the clusters
  completely contained within $\lsd$, and
  $\lsdres := \cup_{A \in \calA'} A$ to be the points in these
  clusters. Now let $\calB' := \{ B \cap \lsdres \mid B \in \calB \}$
  be the restrictions of these clusters of $\calB$ to the clients in
  $\lsdres$. The clusterings $\calA$ and $\calB$ are \emph{cross-free
    with respect to $\lsd$} if $\calA'$ and $\calB'$ are cross-free.
\end{restatable}


\begin{restatable}[Refined Local Search Bound]
  {thm}{localsearch} \label{thm:local-search-bd} Let $(D \cup \FF, d)$
  be an instance of $k$-median and $\eps > 0$. Let $\LS \subseteq \FF$
  with $|\LS| = k$ be a solution for local search with
  $O(1/\eps)$-sized swaps, and let $\BB \subseteq \FF$ be another
  benchmark solution. Given arbitrary assignments
  $\mu_\AA: D \to \LS$, $\mu_\BB: D \to \BB$ and a subset
  $\lsd \subseteq D$ such that
  (i)~$|\BB| + |\mu_{\AA}(D \setminus \lsd)| \leq k$ and (ii)~the pair
  $(\LS, \mu_\AA), (\BB, \mu_\BB)$ is cross-free with respect to
  $\lsd$, the clustering induced by the local-search solution
  $(\LS, \pi_\LS)$ has cost
  \begin{equation}
     \cost(\LS) \leq \cost_{\mu_\BB}(\lsd) + 
    (\cost(\mu_\AA) - \cost_{\LS}(\lsd)) + O(\eps) (\cost(\mu_\AA)+\cost_{\mu_{\BB}}(D_0)).
  \end{equation}
\end{restatable}

We prove \Cref{thm:local-search-bd} in
\S\ref{sec:local-revisited}. Observe that setting $\lsd = D$ in
\Cref{defn:gen-cross-free} and~\Cref{thm:local-search-bd} gives us
back \Cref{defn:cross-free} and \Cref{thm:simpler-ls}
respectively. Despite generalizing \Cref{thm:simpler-ls}, after defining the right abstraction, the proof is nearly identical (see \S\ref{sec:graph-theor-abstr}). The intuition for \Cref{thm:local-search-bd} is that
$\BB$ is a near-optimal solution on $D_0$ of size
$k - |\mu_{\AA}(D \setminus D_0)|$, but possibly a very bad solution
outside of $D_0$. Next, we present our claim about the existence of
good (generalized) cross-free pairs, which are cognizant of the set
$D_\kappa$.
\begin{restatable}[A ``Nice'' Cross-Free Pair]{lemma}{NicePair}
  \label{lem:new-maps-cost}
  Let $\BB := \olamearly \cup \OPT_R$ and
  $\lsd := D^{\early}_\Lambda \cup D_R$. Then there
  exist mappings $\mu_\AA: D \to \LS, \mu_\BB: D \to \BB$ such that
  the associated clusterings $(\calA, \calB)$ form a cross-free pair
  with respect to $\lsd$ and $|\BB| + |\mu_{\AA}(D \setminus D_0)| \leq k$. Moreover,
  \begin{align*}
    \cost_{\mu_\BB}(D_0) &\leq \cost_{\OPT}(D_0) + O(\eps) \sopt, \text{and} \\
    \cost(\mu_\AA) &\leq  \sum_{c^*\in\olamlate\cup\OPT_P}
 \bigl(\cost_\OPT(C^*)+|C^*|\per(C^*)\bigr) +(1+\eta)\cost_\OPT(D_S)  \\
 &+ \cost_\LS(\lsd)
 +2\cost_\OPT(D_\Lambda^{\early}) +  \cost_\OPT(D_R\setminus D_\kappa)\\
&+\sum_{c^*\in\OPT_R}|C^*\setminus D_\kappa|\per(C^*)
 +O(\eps)\sopt.
  \end{align*}
\end{restatable}
(We prove \Cref{lem:new-maps-cost} in
\S\ref{sec:app-proofs-local-search}.) We set $D_0 := D^{\early}_\Lambda \cup D_R$ here because this is precisely the set of clients in clusters without efficient substitutes in $\LS$. For all other clusters, we simply assign them to their substitutes under $\mu_{\AA}$ without having to worry about ensuring cross-freeness. This explains the $\olamlate\cup\OPT_P$ and $D_S$ terms in the bound on $\cost(\mu_{\AA})$. 

Then, within $D_0$, we set $\mu_{\BB}$ to be the assignment under $\pi_{\OPT}$ except that the $D_0$-clients in local search clusters centered at satellites $\Gamma(\OPT_R)$ are ``absorbed'' into $\OPT_R$ clusters. The assignment is constructed such that, for each $c^* \in \OPT_R$, the cluster centered at $c^*$ absorbs the local search cluster centered at $\sigma(c^*) \in \Gamma(\OPT_R)$. This change ensures that (1) each $\OPT_R$ cluster can be sent to its substitute under $\mu_{\AA}$ while preserving cross-freeness and (2) each client in a satellite cluster can be assigned to its satellite center under $\mu_{\AA}$. The combination of these two properties ensures that the low local search cost of the cores of $\OPT_R$ clusters is preserved under $\mu_{\AA}$ and non-core clients incur the substitution cost guaranteed by $\sigma$. This gives the $\OPT_R$-related contributions in the bound on $\cost(\mu_{\AA})$. Finally, the $D_{\Lambda}^{\mathrm{early}}$ bound follows from straightforwardly bounding the cost of reassigning to the nearest local search center.

Given the existence of this
``nice'' cross-free pair, we can infer a bound on the cost of the
local search which has low cost for the clients in $D_\kappa$, as
desired.
\begin{theorem}[Nuanced Local Search Guarantee]
  \label{thm:applied-local-search-bd}
  We have 
\begin{align*}
\cost(\LS)\leq{}&
(1+\eta)\cost_\OPT(D_S)
+3\cost_\OPT(D_\Lambda^{\early}) \nonumber\\
&+\sum_{c^*\in\olamlate\cup\OPT_P}
 \bigl(\cost_\OPT(C^*)+|C^*|\per(C^*)\bigr) \nonumber\\
&+\cost_\OPT(D_\kappa)+2\cost_\OPT(D_R\setminus D_\kappa)
+\sum_{c^*\in\OPT_R}|C^*\setminus D_\kappa|\per(C^*)
+O(\eps)\sopt.
\end{align*}
\end{theorem}

\begin{proof}
 Applying \Cref{thm:local-search-bd} and \Cref{lem:new-maps-cost}, the
$\cost_\LS(\lsd)$ terms cancel. The desired result then follows from the fact that
$\lsd=D_\Lambda^{\early}\cup D_R$.
\end{proof}



%% file: merged-structure.tex
\newcounter{mycount}
\setcounter{mycount}{0}

\section{Structural Theorem: Proofs of~\Cref{thm:structural}
  and~\Cref{thm:struct2}} \label{sec:structural-thm}

In this section, we prove our main structural result
\Cref{thm:structural}. The proof proceeds by a series of substitution
and sampling steps. Over the course of the algorithm and analysis, we
will accumulate a collection of low cost substitutions of centers in
$\OPT$ by centers in $\LS$.

\paragraph{Notation.} In this section it will often be convenient to
combine maps $f_1 : X_1 \to Y, f_2 : X_2 \to Y, \ldots, f_{r} : X_r
\to Y$ where $X_i \cap X_j = \emptyset$ for $i \neq j$. We write
$[f_1, f_2, \ldots, f_r] : X_1 \sqcup X_2 \sqcup \cdots \sqcup X_r \to Y$ to denote the (coproduct) map applying $f_i$ on inputs in $X_i$.

\subsection{Synchronized and Close Clusters}\label{sec:synchr-clust}

The first two substitutions arise directly from local search on
$\beta$-stable instances, and from optimal centers that have a
local-search center close to them.

\begin{definition}[Synchronized Clusters] \label{def:synced} An
  optimum cluster $C^* \in \calO$ is \emph{synchronized} (or
  \emph{synced}) with respect to $\calL$ if there exists a cluster
  $C' \in \calL$ such that
  $|C' \triangle C^*| \leq \eps \min(|C^*|, |C'|)$. In this case, we
  also say that cluster $C'$ is synced and that $C'$ and $C^*$ are
  $\emph{in-sync}$.
\end{definition}
(In prior work~\cite{2approx}, synchronized clusters were called pure
clusters.) Since $\eps < \nf12$, the in-sync pairs form a matching. By
prior results, most optimum clusters are synced with respect to
$\calL$ in $\beta$-stable instances, and synced clusters can cheaply
substitute one another:

\begin{lemma}[Restated from {\cite[Lemma
    IV.4]{cs17}}] \label{lem:few-not-in-sync} Let $(D \cup \FF, d)$ be
  $\beta$-stable, and $\eps \in (0, \nf1{200})$. Then, at most
  $s_0 = \eps^{-2} \beta^{-1}$ clusters $C^* \in \calO$ are not synced
  with respect to $\calL$.
\end{lemma}



\begin{proposition}[Adapted from {\cite[Lemma 4.4]{2approx}}]
  \label{prop:sync-sub}
  There exists $\OPT_{\mathrm{sync}}\sse \OPT$ with
  $|\OPT\setminus \OPT_{\mathrm{sync}}| \leq s_0 = \eps^{-2}\beta^{-1}$, and a
  \emph{substitution map} $\sigma_{\mathrm{sync}} :
  \OPT_{\mathrm{sync}} \to \LS$---where $\sigma_{\mathrm{sync}}(c^*) =
  c'$ for $(c^*,c')$ being an in-sync pair---such
  that $|\LS \setminus \sigma_{\mathrm{sync}}(\OPT_{\mathrm{sync}})| \leq s_0$, and
  \begin{equation} \label{eq:sync-sub} \sum_{c^* \in \OPT_{\mathrm{sync}}} \sum_{p \in C^*} d(p,
    \sigma_{\mathrm{sync}}(c^*)) \leq (1 + O(\eps)) \sum_{c^* \in \OPT_{\mathrm{sync}}} \cost_{\OPT}(C^*) + O(\eps) 
    \cost(\LS).
  \end{equation}
\end{proposition}

The second substitution map is a trivial one: if an optimal center has an
especially close center in $\LS$, it can be substituted.
For $\eta \in [0,1]$, a cluster $C^* \in \calO$ is called
\emph{$\eta$-close} if the closest center $c'\in \LS$ satisfies
$d(c^*, c') \leq \eta \avg_{\OPT}(C^*) + 2\eps \avg_{\LS}(C^*)$.
The following result just uses %
the triangle inequality.

\begin{proposition}[Close Substitution]
  \label{prop:close} 
  Let 
  $\OPT_{\mathrm{cl}} \sse \OPT \setminus \OPT_{\mathrm{sync}}$ be the centers of
  $\eta$-close unsynced clusters, and define
  $\sigma_{\mathrm{cl}} : \OPT_{\mathrm{cl}} \to \LS$, with $\sigma_{\mathrm{cl}}(c^*) =
  \arg\min_{c' \in \LS} d(c^*, c')$. For each $c^* \in \OPT_{\mathrm{cl}}$ and
  the corresponding cluster $C^*$:
  \begin{equation}
    \sum_{p \in C^*} d(p, \sigma_{\mathrm{cl}}(c^*)) \leq (1 + \eta)
    \cost_{\OPT}(C^*) + 2\eps \cost_\LS(C^*). \label{eq:sync-close}
  \end{equation}
\end{proposition}

Define $\OPT_{S_0} := \OPT_{\mathrm{sync}} \cup \OPT_{\mathrm{cl}}$ as the set of centers with substitutions so far, and $\sigma_0 = [\sigma_{\mathrm{sync}}, \sigma_{\mathrm{cl}}]$\margindef{$\OPT_{S_0}, \sigma_0$} the combined substitution map. Also define $\OPT_1 := \OPT \setminus \OPT_{S_0}$\margindef{$\OPT_1$} as the set of remaining centers in $\OPT$; note that we have $|\OPT_1| \leq s_0 = \eps^{-2} \beta^{-1}$ from \Cref{lem:few-not-in-sync}. 

\subsection{Adding Auxiliary Centers using Client Sampling}
\label{sec:client-sample}
The next step of the algorithm is to augment the metric space by a set
$\Lambda_1$ of auxiliary/dummy centers. These centers will serve as
backup centers for a certain subset of $\OPT$ in the final submodular
optimization step of our algorithm. 

\emph{Intuition.} We sample clients with probability proportional to their current cost (a ``$D$-sampling''
step), so an expensive optimal cluster is likely to have a sampled client land in its core. Using that
client as a \emph{leader} and guessing the right radius, we place a dummy center close to the optimal center
$c^*$, all without knowing $c^*$ itself. To optimize the quality of the sampled leaders, we consider the sampling in two rounds, with the second round extending the core further from the center.

\begin{restatable*}{thm}{dsampling} \label{thm:per-client-sampling}
There exists a polynomial time algorithm outputting a set of
  auxiliary centers $\Lambda_1$, so that with
  probability $\exp\{- O(s_0 \log (\eps^{-1}))\}$ there exists a partition $\OPT_1 = \OPT_{\Lambda_1} \sqcup \OPT_{S_1} \sqcup \OPT_2$ with all of the following guarantees:
\begin{enumerate}[label=(\alph*)]
\item \label{cond:sampled-with-dummies} \emph{(Auxiliary Centers)}
  There exists a partition $\OPT_{\Lambda_1} = \OPT_{\Lambda_1}^{\mathrm{early}} \sqcup \OPT_{\Lambda_1}^{\mathrm{late}}$, and a dummy center map
  $\varphi_1 : \OPT_{\Lambda_1} \to \Lambda_1$ such that $\varphi_1$ is a bijection, and for each $c^* \in \OPT_{\Lambda_1}$ with $C^* \in \calO$, we
  have
    \begin{equation*}
        d(c^*, \ell_{\varphi_1(c^*)}) \leq d(\ell_{\varphi_1(c^*)}, \varphi_1(c^*)) = r_{\varphi_1(c^*)} \leq (1 + \newkappa\,\mathbf{1}_{(c^* \in \OPT_{\Lambda_1}^{\mathrm{late}})} + 8 \eps) \per(C^*) + \eps \avg_{\LL}(C^*).
    \end{equation*}   
    \item \label{cond:rest-substitutable} \emph{(Substitutions)} There exists a \emph{substitution map} $\sigma_1 : \OPT_{S_1} \to \LS$ such that
    \begin{equation*}
       \sum_{c^* \in \OPT_{S_1}} \sum_{p \in C^* } d(p, \sigma_1(c^*)) \leq  \eps \sopt.
    \end{equation*}
    \item \label{cond:client-struct} \emph{(Structural Properties)} The remaining centers in $\OPT_2$ have the following properties:
      \begin{enumerate}[nosep, label=(\roman*)] 
      \item \emph{(Remaining clusters are hidden)} \label{cond:hidden} We have
      \begin{equation*}
        \max\Big\{ \sum_{c^* \in \OPT_{1} \setminus \OPT_{\Lambda_1}^{\mathrm{early}}} \cost_{\LS}(\core_{6\eps}(C^*)), \,\, \sum_{c^* \in \OPT_{1} \setminus \OPT_{\Lambda_1}} \cost_{\LS}(\core_\kappa(C^*)) \Big\}< \eps^{5} \sopt,
    \end{equation*}
      \end{enumerate}
      and, for $c^* \in \OPT_2$ with $C^* \in \calO$:
    \begin{enumerate}[nosep, resume, label=(\roman*)]
         \item \emph{($C^*$ has proportionate core cost)}  \label{cond:proportionate} $\cost_{\LL}(\core_\kappa(C^*)) <
          \eps^3 \cost_{\OPT}(C^*)$.
        \item \emph{($\avg_{\OPT}(C^*) \approx \per(C^*)$)} $\frac{1}{1 + 2\eps} \per(C^*) \leq \avg_{\OPT}(C^*) \leq \frac{(1 + \eps)}{\eta} \per(C^*)$. \label{cond:per-roughly-avg}
        \item \emph{($C^*$ has some candidate leader)} There exists $c' \in \LL$ with $d(c^*, c') \leq (1 + \eps)\per(C^*)$ and $d(c', C^*) \leq \eps \per(C^*)$. \label{cond:cand-leader} 
    \end{enumerate}
\end{enumerate}    
\end{restatable*}
After a successful application of
\Cref{thm:per-client-sampling}, the subset
$\OPT_2 = \OPT_1 \setminus (\OPT_{\Lambda_1} \cup \OPT_{S_1})$ \margindef{$\OPT_2$} is
precisely the set of optimal centers still without a substitution or
auxiliary map.

\subsection{Satellites of Remaining Clusters are Spread Out}
Next, we show that $\OPT_2$ has a useful additional structural property. Recall~(\ref{eq:Gamma}), where we defined the set of
satellites for an optimal center $c^* \in \OPT$ (with
corresponding cluster $C^* \in \calO$) as:
\begin{equation*}
  \Gamma(c^*) := \Big\{c \in \LS \mid d(c, c^*) \leq (1 + \newkappa +
  2\eps)\per(C^*), d(c, C^*) \leq 
  \eps\,\per(C^*) \Big\}. 
\end{equation*}
(See \Cref{fig:sat-def}.) The clients in $C^*$ close to satellites in
$\Gamma(c^*)$ are called \emph{witnesses}. We now show that, for $c^* \in \OPT_2$, $\Gamma(c^*)$ has a non-trivial diameter. We will use this fact in our next sampling step when estimating radii for new dummy centers.

\begin{lemma}\label{lem:no-dense-cores}
For each $c^*\in\OPT_2$, the diameter of $\Gamma(c^*)$ is at least
$\eps\per(C^*)$.
\end{lemma}

\begin{proof}
Suppose otherwise. Choose $c'\in\Gamma(c^*)$ with
$d(c^*,c')\leq(1+\eps)\per(C^*)$, as guaranteed by
\Cref{thm:per-client-sampling}\ref{cond:client-struct}%
\ref{cond:cand-leader}. We seek a contradiction with our technical condition \Cref{eq:radius-diameter-condition}. As a reminder, this says that, since
$\per(C^*)\geq\eta\avg_\OPT(C^*)/(1+\eps)$, our choice of $\newkappa$ ensures
\[
(2+\newkappa)\per(C^*)
\geq(2+6\eps)\avg_\OPT(C^*).
\]
We will show that at least $(1-3\eps)|C^*|/2$ clients lie outside
$\core_\kappa(C^*)$. This suffices: at most $\eps|C^*|$ clients
have distance strictly less than $\per(C^*)$, and each client outside
the core has distance greater than $(1+\newkappa)\per(C^*)$. Hence
\begin{align*}
\cost_\OPT(C^*)
&\geq\bigl((1-\eps)|C^*|
 +\newkappa|C^*\setminus\core_\kappa(C^*)|\bigr)\per(C^*)\\
&\geq\frac{2+\newkappa-5\eps}{2}|C^*|\per(C^*),
\end{align*}
where the last step uses $\newkappa<1$. It follows that
\[
(2+6\eps)\avg_\OPT(C^*)
>(2+\newkappa)\per(C^*),
\]
contradicting \eqref{eq:radius-diameter-condition}. This uses that $6 \eps \avg_{\OPT}(C^*) > 5 \eps \per(C^*)$, since $\avg_{\OPT}(C^*) \geq \frac{1}{1 + 2\eps} \per(C^*)$.

So it suffices to show that at least $(1 - 3\eps)|C^*|/2$ clients belong to $C^* \setminus \core_{\kappa}(C^*)$. Roughly speaking, the argument goes as follows. By optimality of $\OPT$, we know that the solution replacing $c^*$ with $c'$ cannot decrease the cost. On the other hand, the assumption that $\Gamma(c^*)$ has small diameter, combined with the fact that the core of $C^*$ is cheap in $\LS$ by \Cref{thm:per-client-sampling}\ref{cond:client-struct}\ref{cond:proportionate} means that clients in $\core_{\kappa}(C^*)$ reassign to $c'$ almost for free. Then most of the cost of $C^*$ must come from clients outside the core. But, these clients pay at most $(1 + \eps) \per(C^*)$ extra each, by the triangle inequality. So, intuitively, we must have many such clients.

Let us begin by bounding the cost of reassigning $C^*$ to $c'$.
For $p\in\core_\kappa(C^*)$, we have
\begin{equation}\label{eq:diam-point-reassignment}
d(p,c')\leq\eps\per(C^*)+\frac4\eps d(p,\LS).
\end{equation}
Indeed, if $d(p,\LS)\leq\eps\per(C^*)$, then
$\pi_\LS(p)\in\Gamma(c^*)$ and the diameter bound applies.
Otherwise,
$d(p,c')\leq(2+\newkappa+\eps)\per(C^*)<4\per(C^*)$, giving \eqref{eq:diam-point-reassignment}.
Summing and using proportionate core cost gives
\begin{equation}\label{eq:diam-core-reassignment}
\sum_{p\in\core_\kappa(C^*)}d(p,c')
\leq\eps|\core_\kappa(C^*)|\per(C^*)
+4\eps^2\cost_\OPT(C^*).
\end{equation}
For clients outside the core, the triangle inequality gives the bound
$d(p,c')\leq d(p,c^*)+(1+\eps)\per(C^*)$.
Combining this with \eqref{eq:diam-core-reassignment}, we obtain
\begin{align}
\sum_{p\in C^*}d(p,c')
&\leq\cost_\OPT(C^*\setminus\core_\kappa(C^*))
+|C^*\setminus\core_\kappa(C^*)|\per(C^*) + \eps|C^*|\per(C^*)+4\eps^2\cost_\OPT(C^*)\nonumber\\
&\leq\cost_\OPT(C^*\setminus\core_\kappa(C^*))
 +|C^*\setminus\core_\kappa(C^*)|\per(C^*)
 +(\eps+6\eps^2)\cost_\OPT(C^*),
\label{eq:diam-total-reassignment}
\end{align}
where the last step uses
$|C^*|\per(C^*)\leq(1+2\eps)\cost_\OPT(C^*)$.

Replacing $c^*$ by $c'$ cannot improve $\OPT$, by optimality. 
Thus \eqref{eq:diam-total-reassignment} implies
\begin{equation}\label{eq:diam-core-optimality}
\cost_\OPT(\core_\kappa(C^*))
\leq|C^*\setminus\core_\kappa(C^*)|\per(C^*)
+(\eps+6\eps^2)\cost_\OPT(C^*).
\end{equation}
At most $\eps|C^*|$ clients have distance strictly less than
$\per(C^*)$ from $c^*$. Hence
\begin{equation}\label{eq:diam-core-lower}
\cost_\OPT(\core_\kappa(C^*))
\geq\bigl(|\core_\kappa(C^*)|-\eps|C^*|\bigr)\per(C^*).
\end{equation}
Combining \eqref{eq:diam-core-optimality} and
\eqref{eq:diam-core-lower} and adding $|\core_{\kappa}(C^*)| \per(C^*)$ to both sides bounds the size of the core:
\begin{equation}\label{eq:diam-core-size}
2|\core_\kappa(C^*)|\per(C^*)
\leq(1+\eps)|C^*|\per(C^*)+(\eps+6\eps^2)\cost_\OPT(C^*).
\end{equation}

By \Cref{thm:per-client-sampling}\ref{cond:client-struct}%
\ref{cond:per-roughly-avg}, we have
$\cost_\OPT(C^*)\leq(1+\eps)|C^*|\per(C^*)/\eta$.
Substituting in \eqref{eq:diam-core-size} and dividing by
$2\per(C^*)$ gives
\[
|\core_\kappa(C^*)|
\leq\frac{1+\eps+(\eps+6\eps^2)(1+\eps)/\eta}{2}|C^*|
\leq\frac{1+3\eps}{2}|C^*|,
\]
where $(1+6\eps)(1+\eps)/\eta<2$ for
$\eta=89/100$ and $\eps<1/200$.
This implies $|C^* \setminus \core_\kappa(C^*)| \geq \frac{(1 - 3\eps) |C^*|}{2}$, as required. 
\end{proof}

\subsection{Adding Auxiliary Centers using Cluster Sampling}
\label{sec:cluster-sample}
We now sample a few \emph{centers} (and not \emph{clients}) from the
local search solution $\LS$ and use them as leaders to add more
auxiliary centers; this further reduces the number of optimal centers which do
not have substitutions or auxiliary center maps (and imposes more
structure on these remaining centers). Along the way, we also identify
several substitute centers in $\LS$ as backups and some others as
redundant and removable. %
For brevity, define
$\LS_X := \sigma_0(\OPT_{S_0}) \cup \sigma_1(\OPT_{S_1})$ \margindef{$\LS_X$} %
to be the image of the substitution maps defined so far. 
The following theorem outputs $\LS'
\subseteq \LS$ \margindef{$\LS'$} along with another set of
dummies $\Lambda_2$.

\begin{restatable*}[Cluster
  Sampling]{thm}{expensive} \label{thm:expensive} There exists a
  polynomial time algorithm outputting $\LS' \subseteq \LS$ and a set of dummy centers
  $\Lambda_2$ so that with probability
  $\exp\{-O(s_0 \log(\eps^{-1}))\}$ there exists partition
  $\OPT_2 = \OPT_{\Lambda_2} \sqcup \OPT_X \sqcup \OPT_B$ and
  $\LS_X \sse \LS'$, with the
  following guarantees:
  \begin{enumerate}[nosep,label = (\alph*)]
  \item \label{cond:expensive-dummy} \emph{(More Auxiliary Centers)}
    There exists a bijection $\varphi_2: \OPT_{\Lambda_2} \to \Lambda_2$ such that,
    for each $c^* \in \OPT_{\Lambda_2}$:
    \begin{equation*}
      d(c^*, \ell_{\varphi_2(c^*)}) \leq  r_{\varphi_2(c^*)} = d(\ell_{\varphi_2(c^*)}, \varphi_2(c^*)) \leq (1 + \newkappa + 6\eps) \per(C^*).  
    \end{equation*}
  \item \label{cond:BK2-psuedodummy} \emph{(Centers with Protected Backups)} There exists a map $\psi_X : \OPT_X \to \LS_X$ such that, for each $c^* \in \OPT_X$,
    \[ d(c^*, \psi_X(c^*)) \leq (1 + \newkappa + 4\eps)\per(C^*). \]
  \item \label{cond:expensive-rem-cost-bds} \emph{(Remaining Missed
      Clusters are Cheap)} Define
    $\MM' := \LS' \cup \Lambda_1 \cup \Lambda_2$ and $\calM'$ the
    clustering induced by $\pi_{\MM'}$. Then, we have
    $\Gamma(\OPT_B) \sse \LS'$,
    \begin{equation}
      \sum_{c  \in \Gamma(\OPT_B)} \sum_{p \in C} d(p, \LS) < \eps^3 \sopt,  \label{eqtn:B_K^3Lbd}
    \end{equation}
    and
    \begin{equation} 
      \sum_{\substack{M_c \in \calM'  \\ c \in \Gamma(\OPT_B) \cup (\LS' \setminus \LS_X)}} \sum_{p \in M_c} d(p, \MM') < \eps^3 \sopt. 
    \end{equation}
  \item \emph{(Cost Bound)} We have
    $\cost(\MM') \leq (5 + O(\eps)) \sopt$. \label{cond:expensive-post-rem-cost-bd}
  \end{enumerate}   
\end{restatable*}
This is the last ingredient we need in our structural theorem.

\input{concluding-structure}

\input{client-sampling-new}

\input{expensive-sampling}

%% file: concluding-structure.tex
\subsection{Concluding the Proofs of~\Cref{thm:structural} and~\Cref{thm:struct2}}

Concluding~\Cref{thm:structural} and~\Cref{thm:struct2} is now mostly
a matter of collecting properties and parameters. 

\structural*

\begin{proof}[Proof of~\Cref{thm:structural}]
  We already defined $\LS'$ in~\Cref{thm:expensive}. We define $\Lambda := \Lambda_1
  \cup \Lambda_2$, $\OPT_{\Lambda} := \OPT_{\Lambda_{1}} \sqcup
    \OPT_{\Lambda_{2}}$, $\OPT_{\Lambda}^{\mathrm{early}} := \OPT_{\Lambda_1}^{\mathrm{early}}$, $\OPT_{\Lambda}^{\mathrm{late}} = \OPT_{\Lambda} \setminus \OPT_{\Lambda}^{\mathrm{early}}$, $\OPT_S := \OPT_{S_0} \sqcup \OPT_{S_1}$, $\varphi = [\varphi_1, \varphi_2] : \OPT_{\Lambda} \to \Lambda$, and $\sigma_S = [\sigma_0, \sigma_1] : \OPT_S \to \LS'$.
    
    To finish the definition of the partition of $\OPT$ and define the full substitution and backup maps $\sigma$ and $\psi$, we next define $\OPT_P$ and $\OPT_R$ such that $\OPT_P \sqcup \OPT_R = \OPT_X \sqcup \OPT_B$, $\OPT_X \subseteq \OPT_P$, and $\OPT_R \subseteq \OPT_B$. Consider the following greedy algorithm. 
    \begin{enumerate}
        \item Initialize $\OPT_P \leftarrow \OPT_X$, $\OPT_R \leftarrow \emptyset$, and $\psi: \OPT_P \to \LS'$ as $\psi_X : \OPT_X \to \LS_X \subseteq \sigma_S(\OPT_S)$ as in \Cref{thm:expensive}\ref{cond:BK2-psuedodummy}. 
        \item We also maintain an assignment $\sigma_R : \OPT_R \to \LS'$ such that, for each $c^* \in \OPT_R$ and $\sigma := [\sigma_S, \sigma_R] : \OPT_S \cup \OPT_R \to \LS'$, $\sigma(c^*)$ has a unique preimage, as needed for \Cref{thm:structural}\ref{cond:struct-subs}\ref{item:unique-preimage}.
        \item Iterate over $c^* \in \OPT_B$ in an arbitrary order:
        \begin{itemize}
            \item If there exists $c' \in \sigma(\OPT_S \cup \OPT_R)$ with $d(c^*, c') \leq (1 + \eps) \per(C^*)$, then set $\OPT_P \leftarrow \OPT_P \cup \{c^*\}$ and extend $\psi$ by $\psi(c^*) = c' \in \sigma(\OPT_S \cup \OPT_R)$.
            \item Otherwise, since $c^* \in \OPT_B \subseteq \OPT_2$, by \Cref{thm:per-client-sampling}\ref{cond:client-struct}\ref{cond:cand-leader} and \Cref{thm:expensive}\ref{cond:expensive-rem-cost-bds}, there exists some $c' \in \Gamma(c^*) \subseteq \LS'$ with $d(c^*, c') \leq (1 + \eps)\per(C^*)$. In this case we set $\OPT_R \gets \OPT_R \cup \{c^*\}$ and extend $\sigma$ by $\sigma(c^*) = c'$.
        \end{itemize}
    \end{enumerate} 
    This defines our partition $\OPT := \OPT_\Lambda \sqcup \OPT_P \sqcup \OPT_R \sqcup \OPT_S$, the backup map $\psi: \OPT_P \to \LS_P \subseteq \LS'$ (where $\LS_P := \psi(\OPT_P)$), and the substitution map $\sigma: \OPT_S \sqcup \OPT_R \to \LS'$. The algorithm is to run the algorithms of \Cref{thm:per-client-sampling} and \Cref{thm:expensive} one after another, immediately yielding the running-time and probabilistic guarantees. It remains to check the required properties.
  \begin{enumerate}[label=(\alph*), nosep]
  \item Note that $\OPT_{\Lambda}, \OPT_P, \OPT_R$ are disjoint by
    definition and all subsets of
    $\OPT \setminus \OPT_{\mathrm{sync}}$. Hence,
    by~\Cref{lem:few-not-in-sync}, we have
    $|\OPT_{\Lambda}| + |\OPT_P| + |\OPT_R| \leq s_0$. 
  \item This is immediate by the combination of
    \Cref{thm:per-client-sampling}\ref{cond:sampled-with-dummies}
    and~\Cref{thm:expensive}\ref{cond:expensive-dummy}, using $\err_p$ to absorb all of the error terms. This uses $\per(C^*) \leq (1 + 2\eps) \avg_{\OPT}(C^*)$.
  \item This is immediate on $\OPT_X$ by~\Cref{thm:expensive}\ref{cond:BK2-psuedodummy} and follows from our construction of $\psi$ for $\OPT_P \setminus \OPT_X$. (In fact, these latter centers have the even stronger $(1 + \eps) \per(C^*)$ distance upper bound, but we will not need this.) 
  \item By construction of $\sigma$ and \Cref{thm:expensive}\ref{cond:expensive-rem-cost-bds} we have that 
      $\LS_P  \subseteq \LS_\sigma \subseteq \LS'$.
    \begin{enumerate}[label = (\roman*)]
    \item The cost bound follows from the definition of $\sigma$ and
      the bounds from~\Cref{prop:sync-sub}, \Cref{prop:close}, and
      \Cref{thm:per-client-sampling}\ref{cond:rest-substitutable}. This uses $\cost(\LS) \leq 5 \sopt$.   
    \item This is immediate by construction of $\sigma_R$ and the triangle inequality.
\item By \Cref{thm:expensive}\ref{cond:expensive-rem-cost-bds} we have $\Gamma(\OPT_R) \sse \LS'$, and each $\sigma(c^*)$ for $c^* \in \OPT_R$ has a unique preimage by construction of $\sigma_R$. \qedhere
    \end{enumerate}
  \end{enumerate}    
\end{proof}

\structuraltwo*

\begin{proof}[Proof of~\Cref{thm:struct2}]
  We can again check the required properties one-by-one.
  \begin{enumerate}[label = (\alph*), nosep]
  \item This is immediate from the second guarantee
    of~\Cref{thm:expensive}\ref{cond:expensive-rem-cost-bds}, since we defined $\OPT_R$ as a subset of $\OPT_B$ and $\LS_X \subseteq \LS_{\sigma}$.
  \item This is by the definition of close clusters: in fact, for all
    $c^* \in \OPT \setminus \OPT_{S_0}$ we have
    $d(c^*, \LS) > \eta \avg_{\OPT}(C^*) + 2\eps
    \avg_{\LS}(C^*)$. Notably, $\OPT_{S_0} \subseteq \OPT_S$.
  \item Since $\OPT_P \cup \OPT_{\Lambda}^{\mathrm{late}} \subseteq \OPT_1 \setminus \OPT_{\Lambda_1}^{\mathrm{early}}$, this follows from \Cref{thm:per-client-sampling}\ref{cond:client-struct}\ref{cond:hidden}.
  \item The last set of properties all concern $\OPT_R$. Note that
    by definition $\OPT_R \subseteq \OPT_2$.
    \begin{enumerate}[label = (\roman*), nosep]
    \item This follows
      by~\Cref{thm:per-client-sampling}\ref{cond:client-struct}\ref{cond:hidden}
      and
      \Cref{thm:per-client-sampling}\ref{cond:client-struct}\ref{cond:proportionate}. (Here we use that $\OPT_R \subseteq \OPT_2 \subseteq \OPT_1 \setminus \OPT_{\Lambda_1}$.)
    \item This follows by the first bound of~\Cref{thm:expensive}\ref{cond:expensive-rem-cost-bds}.
    \item This follows from \Cref{thm:per-client-sampling}\ref{cond:client-struct}\ref{cond:per-roughly-avg}. \qedhere
    \end{enumerate}
  \end{enumerate}
\end{proof}


%% file: client-sampling-new.tex
\subsection{$D$-Sampling and Dummy Addition: Proof of \Cref{thm:per-client-sampling}}
\label{sec:d-sampling-proof}

Next, we give the proofs for the various supporting theorems from
\S\ref{sec:structural-thm}. Let us recall
\Cref{thm:per-client-sampling}, our result for client sampling.

\dsampling

To prove this result, we show that \Cref{alg:client-sampling}
satisfies the properties with the appropriate probability.

\bigskip
\begin{algorithm} 
\caption{Client sampling and dummy guessing}
\label{alg:client-sampling}
  $\Lambda_1 \gets \emptyset$ \;
  \For{$i = 1, \ldots, 2s_0$}{
    Sample a random client $\ell \in D$ with probability
    $d(\ell, \pi_{\LS}(\ell))/\cost(\LS)$\;   
    w.p. $1/2$, \textbf{continue}\; \label{line:skip-client-sampling}
    $r_0 \gets d(\ell, \LS)$\;
    $r \gets$ random $(1 + \eps)^x r_0$ with $\eps^3 r_0 \leq r \leq r_0 / \eps^7$, $x \in \mathbb{Z}$\;
    $\Lambda_1 \gets \Lambda_1 \cup \{ \lambda \}$, where $\lambda$ is added with $d(\ell, \lambda) = r$ 
    \tcp*{$\ell = \ell_{\lambda}$, $r = r_{\lambda}$} 
  }
\end{algorithm}
\bigskip

Let $H \subseteq D$ be the random subset of clients sampled (hit) in \Cref{alg:client-sampling} for which the algorithm added a corresponding $\lambda \in \Lambda_1$ (i.e., did not skip the iteration in Line~\ref{line:skip-client-sampling}). We will set $\OPT_{\Lambda_1} = \OPT_{\Lambda_1}^{\mathrm{early}} \cup \OPT_{\Lambda_1}^{\mathrm{late}}$, where 
\begin{equation*}
\OPT_{\Lambda_1}^{\mathrm{early}} := \{c^* \in \OPT_1 \mid \core_{6\eps}(C^*) \cap H \neq \emptyset \}, \quad \text{and} \quad \OPT_{\Lambda_1}^{\mathrm{late}} := \{c^* \in \OPT_1 \setminus \OPT_{\Lambda_1}^{\mathrm{early}} \mid \core_{\kappa}(C^*) \cap H \neq \emptyset \}.
\end{equation*} 

We now define an event which yields our desired property and happens with reasonable probability.

\begin{lemma} \label{lem:good-sampling}
Let $\calG$ be the event that
\begin{enumerate}[label = (\Alph*),nosep]
    \item \label{cond:good-no-repeats}
Every $\lambda\in\Lambda_1$ has its leader in
$\core_\kappa(C^*)$ for some $c^*\in\OPT_1$, and distinct
dummy centers have leaders in distinct optimal clusters.
     \item \label{cond:good-radius-guess} For each $\ell_{\lambda} \in H$, where $\ell_{\lambda} \in C^*$, we have
    \begin{equation*}
        d(\ell_\lambda, c^*) \leq r_{\lambda} \leq (1 + \newkappa \mathbf{1}_{(c^* \in \OPT_{\Lambda_1}^{\mathrm{late}})} + 8 \eps) \per(C^*) + \eps \avg_{\LL}(C^*).
    \end{equation*}
    \item \label{cond:first-core-bd} $\sum_{c^* \in \OPT_1 \setminus \OPT_{\Lambda_1}^{\mathrm{early}}} \cost_{\LL}(\core_{6\eps}(C^*)) < \eps^{5} \sopt$. 
    \item \label{cond:good-hit-all} $\sum_{c^* \in \OPT_1 \setminus \OPT_{\Lambda_1}} \cost_{\LL}(\core_{\kappa}(C^*)) < \eps^{5} \sopt$. 
\end{enumerate}
Then, $\calG$ occurs with probability at least $\exp\{- O(s_0  \log (\eps^{-1}))\}$. 
\end{lemma}
\begin{proof}
Let $\widetilde{H}$ denote the set $H$ at some iteration of \Cref{alg:client-sampling}. Define $\OPT_{\widetilde{H}}$ analogously to $\OPT_{\Lambda_1}$ but with respect to $\widetilde{H}$. For $c^* \in \OPT_1$ define its \emph{punched core} as
\begin{equation*}
\widehat{\core}_{6\eps}(C^*) := \{p \in \core_{6 \eps}(C^*) \mid d(p, \LL) \geq \eps^{6} \per(C^*)\}.   
\end{equation*} 
Similarly define $\widehat{\core}_{\kappa}(C^*)$.
Assume that~\ref{cond:good-no-repeats} holds for the dummy centers
added so far. It is maintained by the following event.
\begin{itemize}
    \item Let $\ell\in C^*$ be the sampled client.
    If $c^*\notin\OPT_1$, the iteration is skipped on
    line~\ref{line:skip-client-sampling}.
    Otherwise, the iteration is skipped if
    $c^*\in\OPT_{\widetilde H}$ or
    $\ell\notin\widehat{\core}_\kappa(C^*)$.
    If $i\leq s_0$ and
    $\ell\notin\widehat{\core}_{6\eps}(C^*)$,
    the iteration is also skipped.
    In all other cases, the iteration is not skipped.
\end{itemize}
This happens with probability $1/2$ on each iteration. Next, suppose that~\ref{cond:good-radius-guess} holds up to the current iteration, with respect to $\widetilde{H}$. We bound the probability it is maintained in this iteration, given the iteration is not skipped from the event above. Since the above event holds, $r_0 \geq \eps^{6} \per(C^*)$. If $i \leq s_0$, we also have $\ell \in \core_{6\eps}(C^*)$; otherwise we just have $\ell \in \core_{\kappa}(C^*)$. Call $\ell$ early if $\ell\in\core_{6\eps}(C^*)$, and late
otherwise. Since each optimal cluster has at most one accepted
leader, this agrees with the classification of $c^*$ in
$\OPT_{\Lambda_1}^{\mathrm{early}}$ or
$\OPT_{\Lambda_1}^{\mathrm{late}}$. We have
\[
d(\ell,c^*)\leq
\begin{cases}
(1+6\eps)\per(C^*) & \text{if $\ell$ is early},\\
(1+\newkappa)\per(C^*) & \text{if $\ell$ is late}.
\end{cases}
\] If $r_0 > \eps^{-1} \per(C^*)$, then, for all $p \in \core_{\kappa}(C^*)$, 
\begin{equation*}
d(p, \LS) \geq d(\ell, \LS) - d(p, \ell) \geq r_0 - 4 \per(C^*) \geq \frac{r_0}{2}.    
\end{equation*}
In particular, since $|\core_{\kappa}(C^*)| \geq \eps |C^*|$, we get $\cost_\LS(C^*) > \frac{\eps}{2} r_0 |C^*|$. As such, $\eps \avg_\LS(C^*) > \frac{\eps^2}{2} r_0$. In either case, there exists some $r = (1 + \eps)^x r_0$ such that $d(\ell, c^*) \leq r \leq (1 + \newkappa\mathbf{1}_{(\ell \text{ late})}+ 8\eps) \per(C^*) + \eps \avg_\LS(C^*)$ with $x \in \mathbb{Z}$ and $\eps^3 r_0 \leq r \leq \eps^{-7}r_0$. A random guess of $(1 + \eps)^x r_0$ in this interval around $r_0$ then meets the required guarantees with probability at least $\frac{\eps}{20 \log (\eps^{-1})}  \geq \eps^2 / 5$. The arguments for \ref{cond:first-core-bd} and \ref{cond:good-hit-all} are similar. We ensure  \ref{cond:first-core-bd} and \ref{cond:good-hit-all} by bounding the probability that the following event holds.  
\begin{itemize}
    \item If $i \leq s_0$ and $\sum_{c^* \in \OPT_1 \setminus \OPT_{\widetilde{H}}} \cost_{\LS}(\core_{6\eps}(C^*)) \geq \eps^{5} \sopt$, then $\ell$ is sampled so that $\ell \in \widehat{\core}_{6\eps}(C^*)$ with $c^* \in \OPT_1 \setminus \OPT_{\widetilde{H}}$.
    \item If $i > s_0$ and $\sum_{c^* \in \OPT_1 \setminus \OPT_{\widetilde{H}}} \cost_{\LS}(\core_{\kappa}(C^*)) \geq \eps^{5} \sopt$, then $\ell$ is sampled so that $\ell \in \widehat{\core}_{\kappa}(C^*)$ with $c^* \in \OPT_1 \setminus \OPT_{\widetilde{H}}$.
\end{itemize}
Assuming the first event ensuring~\ref{cond:good-no-repeats} holds, each ``if'' condition can only hold in at most $|\OPT_1| \leq s_0$ iterations. Thus~\ref{cond:first-core-bd} holds after the first $s_0$
iterations and remains valid because
$\OPT_{\Lambda_1}^{\mathrm{early}}$ only grows.
Similarly, \ref{cond:good-hit-all} holds after the second
$s_0$ iterations. Now, if $i \leq s_0$ and $\sum_{c^* \in \OPT_1 \setminus \OPT_{\widetilde{H}}} \cost_{\LS}(\core_{6\eps}(C^*)) \geq \eps^{5} \sopt$, then 
\begin{align*}
\sum_{c^* \in \OPT_1 \setminus \OPT_{\widetilde{H}}} \cost_{\LS}(\widehat{\core}_{6\eps}(C^*)) &\geq  \sum_{c^* \in \OPT_1 \setminus \OPT_{\widetilde{H}}} (\cost_{\LS}(\core_{6\eps}(C^*)) - |C^*| \eps^{6} \per(C^*)) \\
&\geq (\eps^{5} - (1 + 2\eps)\eps^{6}) \sopt \geq \eps^{6} \sopt.
\end{align*}
The probability of sampling a good $\ell$ if the condition holds is then at least 
\begin{equation*}
    \frac{\eps^{6} \sopt}{\cost(\LS)} \geq \frac{\eps^{6} \sopt}{5 \sopt} = \frac{\eps^{6}}{5}.
\end{equation*}
An identical argument gives the same lower bound on the probability in the case when $i > s_0$. Combining these bounds shows that
$\Pr[\calG] \geq (\eps^{8}/(50))^{2s_0} = \exp\{- O(s_0 \log
(\eps^{-1}))\}$.
\end{proof}
Assume that $\calG$ holds. Define
\begin{align*}
\OPT_{S_1}&:=\left\{c^*\in\OPT_1\setminus\OPT_{\Lambda_1}:
  \cost_\LS(\core_\kappa(C^*))\geq
  \eps^3\cost_\OPT(C^*)\right\},\\
\OPT_2&:=\OPT_1\setminus(\OPT_{\Lambda_1}\cup\OPT_{S_1}),
\end{align*}
and set $\sigma_1(c^*):=\pi_\LS(c^*)$ for $c^*\in\OPT_{S_1}$.
For each $c^*\in\OPT_{\Lambda_1}$, define $\varphi_1(c^*)$
to be the unique dummy center whose leader belongs to $C^*$.
By~\ref{cond:good-no-repeats} and the definition of
$\OPT_{\Lambda_1}$, this gives a bijection
$\varphi_1:\OPT_{\Lambda_1}\to\Lambda_1$. The auxiliary-center and hiddenness guarantees then follow immediately from
\Cref{lem:good-sampling}, and proportionate core cost holds by definition. (These are \Cref{thm:per-client-sampling}\ref{cond:sampled-with-dummies}, \ref{cond:client-struct}\ref{cond:hidden}, and \ref{cond:client-struct}\ref{cond:proportionate}, respectively.)
We verify the remaining guarantees below.

Since $\core_0(C^*)$ contains at least $\eps|C^*|$ clients, each within
distance $\per(C^*)$ of $c^*$, averaging the triangle inequality gives
\begin{align}
\label{eq:cleanup-nearest-center}
d(c^*,\LS)
&\leq \frac{1}{|\core_0(C^*)|} \sum_{p \in \core_0(C^*)} (d(c^*, p) + d(p, \LS)) \leq \per(C^*)+
  \frac{\cost_\LS(\core_0(C^*))}{|\core_0(C^*)|} \nonumber \\
&\leq \per(C^*)+
  \frac{\cost_\LS(\core_\kappa(C^*))}{\eps|C^*|}.
\end{align}

Now, the definition of $\OPT_{S_1}$ and
\Cref{lem:good-sampling}\ref{cond:good-hit-all} imply
\begin{equation}
\sum_{c^*\in\OPT_{S_1}}\cost_\OPT(C^*)
\leq\eps^{-3}\sum_{c^*\in\OPT_{S_1}}
  \cost_\LS(\core_\kappa(C^*))
<\eps^2\sopt. \label{eq:opts1-bd}    
\end{equation}

By the triangle inequality, definition of $\sigma_1$, \eqref{eq:cleanup-nearest-center},
and the trivial upper bound on $\per(C^*)$ from \eqref{eq:per-bd}, we get 
\begin{align*}
\sum_{c^*\in\OPT_{S_1}}\sum_{p\in C^*}d(p,\sigma_1(c^*))
&\leq \sum_{c^*\in\OPT_{S_1}}\sum_{p\in C^*}(d(p,c^*) + d(c^*, \LS)) \\
&\leq \sum_{c^*\in\OPT_{S_1}}\left(
  \cost_\OPT(C^*)+|C^*|\per(C^*)+
  \eps^{-1}\cost_\LS(\core_\kappa(C^*))\right) \\
  &\leq \sum_{c^*\in\OPT_{S_1}}\left(
  \cost_\OPT(C^*)+(1 + 2\eps)|C^*|\avg_{\OPT}(C^*)+
  \eps^{-1}\cost_\LS(\core_\kappa(C^*))\right) \\
  &= \sum_{c^*\in\OPT_{S_1}}\left(
  (2 + 2\eps)\cost_\OPT(C^*)+
  \eps^{-1}\cost_\LS(\core_\kappa(C^*))\right).
\end{align*}
Then, applying \Cref{eq:opts1-bd} to the first term and \Cref{lem:good-sampling}\ref{cond:good-hit-all} to the second, we get a bound of $((2 + 2\eps)\eps^2 + \eps^4) \sopt < \eps \sopt$, as needed for \Cref{thm:per-client-sampling}\ref{cond:rest-substitutable}.

It remains to show \Cref{thm:per-client-sampling}\ref{cond:client-struct}\ref{cond:per-roughly-avg} and \ref{cond:client-struct}\ref{cond:cand-leader}. For this, let $c^*\in\OPT_2$. Since $c^*\in\OPT_1$ is not close to $\LS$ (as in \Cref{prop:close}),
\eqref{eq:cleanup-nearest-center} and the definition of $\OPT_2$ give
\[
\eta\avg_\OPT(C^*)<d(c^*,\LS)
<\per(C^*)+\eps^2\avg_\OPT(C^*).
\]
Together with the trivial upper bound on $\per(C^*)$ from \eqref{eq:per-bd}, this yields
\[
\frac{\per(C^*)}{1+2\eps}
\leq\avg_\OPT(C^*)
<\frac{\per(C^*)}{\eta-\eps^2}
\leq\frac{1+\eps}{\eta}\per(C^*),
\]
giving \ref{cond:client-struct}\ref{cond:per-roughly-avg}. Finally, choose $q\in\core_0(C^*)$ minimizing $d(q,\LS)$
and set $c':=\pi_\LS(q)$. Then
\[
d(q,c')
\leq\frac{\cost_\LS(\core_\kappa(C^*))}{\eps|C^*|}
<\eps^2\avg_\OPT(C^*)
<\frac{\eps^2}{\eta-\eps^2}\per(C^*)
\leq \eps \per(C^*).
\]
Thus $d(c^*, c') \leq d(c^*, q) + d(q, c') \leq (1 + \eps) \per(C^*)$, proving \ref{cond:client-struct}\ref{cond:cand-leader}.
The inequalities above use $\eta\geq1/2$ and $\eps<1/200$.
This completes the proof of \Cref{thm:per-client-sampling}.

%% file: expensive-sampling.tex
\subsection{Expensive Cluster Sampling: Proof of \Cref{thm:expensive}}
\label{sec:proof-expens-clust-sampl}

We start this section by bounding the cost of the solution we will be
sampling from in \Cref{thm:expensive}. The constant in the proposition
is very loose, but this looseness is not relevant and only changes the
final success probability.
\begin{proposition}
  \label{prop:cost-bound-for-exceptional}
  We have 
  \begin{equation*}
   \sum_{c^* \in \OPT \setminus \OPT_2} \sum_{p \in C^*} d(p, \LS_X \cup \Lambda_1)
    \leq 5 \sum_{c^* \in \OPT \setminus \OPT_2} \cost_{\OPT}(C^*) + O(\eps)\sopt.     
  \end{equation*}
\end{proposition}
\begin{proof}[Proof of \Cref{prop:cost-bound-for-exceptional}]
Note that $\OPT \setminus \OPT_2 = \OPT_{S_0} \cup \OPT_{S_1} \cup \OPT_{\Lambda_1}$. For $c^* \in \OPT_{\Lambda_1}$, by \Cref{thm:per-client-sampling}\ref{cond:sampled-with-dummies}, there exists $\lambda \in \Lambda_1$ (namely $\lambda = \varphi_{1}(c^*)$) such that $d(c^*, \lambda) \leq (2 + 2\newkappa + 16\eps) \per(C^*) + 2\eps \avg_\LS(C^*)$. So, using that $\per(C^*) \leq (1 + 2 \eps) \avg_{\OPT}(C^*)$ and $\newkappa \leq 1$,
\begin{align*}
 \sum_{c^* \in \OPT_{\Lambda_1}} \sum_{p \in C^*} d(p, \Lambda_1) &\leq   \sum_{c^* \in \OPT_{\Lambda_1}} \sum_{p \in C^*} (d(p, c^*) + d(c^*, \Lambda_1)) \\
 &\leq  \sum_{c^* \in \OPT_{\Lambda_1}} ((3 + 2\newkappa + 25\eps) \cost_{\OPT}(C^*) + 2 \eps \cost_\LS(C^*)).
\end{align*}
Then, combining the above with the substitution cost bounds from \Cref{prop:sync-sub,prop:close} and \Cref{thm:per-client-sampling}\ref{cond:rest-substitutable} and the fact that $\cost(\LS) \leq 5 \sopt$, we get
\begin{align*}
  \sum_{c^* \in \OPT \setminus \OPT_2} \sum_{p \in C^*} d(p, \LS_X \cup \Lambda_1) &\leq  \sum_{c^* \in \OPT_{\mathrm{sync}}} \cost_{\OPT}(C^*) + (1+\eta) \sum_{c^* \in \OPT_{\mathrm{cl}}} \cost_{\OPT}(C^*) \\
&\quad+   5\sum_{c^* \in \OPT_{\Lambda_1}} \cost_{\OPT}(C^*) + O(\eps) \sopt \\
&\leq 5 \sum_{c^* \in \OPT \setminus \OPT_2} \cost_{\OPT}(C^*) + O(\eps)\sopt.  \qedhere
\end{align*}
\end{proof}

We can now prove \Cref{thm:expensive}, which shows how to add further
auxiliary centers by sampling ``expensive'' cluster centers from the
local search solution, and using them as leaders.

\expensive

\newcommand{\remls}{\LS^{-}}
\begin{algorithm}[t]
\SetCommentSty{textit}
\caption{Iterative cluster sampling and dummy creation}
\label{alg:expensive-sampling+dummy}
 $\Lambda_2, \remls \gets \emptyset$\;
 \For{$i \in \{0,1\}$}{
 \For{$j \in  [s_0]$}{ \label{line:inner-big-while}
      w.p.\ $\nf12$, \textbf{continue}\; \label{line:skip}
      \lIf{$i = 0$}{sample $\ell \in \LS \setminus \remls$ w.p. $\propto$ $\cost_{\LS}(C_{\ell})/\cost(\LS)$, where $C_{\ell} \in \calL$ \label{line:sample1}
      }
      \Else{
      Sample $\ell \in \LS \setminus \remls$ w.p. $\propto$ $\sum_{p \in M_{\ell}} d(p, \ell) /\cost((\LS \setminus \remls) \cup \Lambda_1 \cup \Lambda_2)$, $M_\ell$ in clustering of  $(\LS \setminus \remls) \cup \Lambda_1 \cup \Lambda_2$\; \label{line:sample2}
      }

      \lIf{$|\remls| \leq s_0$}{
      w.p. $\nf{1}{2}$, $\remls \gets \remls \cup \{\ell\}$,  and \textbf{continue} \label{line:r-add1}}
      \tcp{$\ell$ is a candidate leader}
      \While{$|\Lambda_2| + |\remls| \leq 2s_0$}{ \label{line:dummy-guessing-loop}
      $c' \gets \arg \min_{c'' \in \LS \setminus (\remls \cup \{\ell\})} d(\ell, c'')$\;
      $b \gets$ True w.p. $\nf{1}{2}$, False o.w.\label{line:b-ready?}\;
      \If{$b$}{
      $r_0 \gets d(\ell,c')$ \tcp*{$d(\ell, \LS \setminus (\remls \cup \{\ell\}))$ is a good estimate of radius}
      \While{$|\Lambda_2| + |\remls| \leq 2s_0$}{ \label{line:radius-guessing-loop}
      $r \gets$ random $(1 + \eps)^x r_0$ with $\eps^3 r_0 \leq r  \leq r_0/\eps^3$, $x \in \mathbb{Z}$\;
      $\Lambda_2 \gets \Lambda_2 \cup \{\lambda\}$, where $\lambda$ is added with $d(\ell, \lambda) = r$\label{line:dummy-addition} 
      \tcp*{$\ell = \ell_{\lambda}$, $r = r_{\lambda}$} 
      w.p. $\nf{1}{2}$ \textbf{break}\;
      }
      $\remls \gets \remls \cup \{\ell\}$, and \textbf{break} \label{line:r-add3}\; 
      }
      \lElse{
      $\remls \gets \remls \cup \{c'\}$ \label{line:r-add2}
      }
      }
    }
}
  \Return $\Lambda_2$ and $\LS' := \LS \setminus \remls$.
\end{algorithm}

\begin{proof}
Define $\OPT_X \subseteq \OPT_2$ as the subset of centers $c^* \in \OPT_2$ such that there exists $c' \in \LS_X$ with $d(c^*, c') \leq (1 + \newkappa + 4\eps)\per(C^*)$. Then define $\psi_X : \OPT_X \to \LS_X$ by setting $\psi_X(c^*)$ equal to some $c'$ as above for each $c^* \in \OPT_X$. By construction this satisfies \ref{cond:BK2-psuedodummy}.

Next we show that, given this definition of $\OPT_X$ and $\psi_X$, \Cref{alg:expensive-sampling+dummy} satisfies the
  guarantees of~\Cref{thm:expensive}. The algorithm is intuitively
  simple. Repeat many rounds of sampling centers in $\LS$ according to
  their cost in some appropriately chosen solution (lines
  \ref{line:sample1} and \ref{line:sample2}). The goal of sampling is
  (1) to detect redundant centers from $\LS$ which we can remove from
  our final solution, and (2) to detect centers which are ``good
  leaders'' insofar as they can be used to give dummy centers for
  remaining centers in $\OPT_2 \setminus \OPT_X$. If the centers in
  cases (1) or (2) make up a substantial portion of the cost in the
  solution we are sampling from then we have a good probability of
  sampling one such cluster. We can guess which case it falls under,
  and we do so in line \ref{line:r-add1}. In case (1), we just go
  ahead and delete the center, adding it to $\remls$. Otherwise, in
  case (2), we use the distance of the sampled center (candidate
  leader) $\ell$ to its nearest center $c'$ as a radius estimate (for
  constructing dummy centers). The distance $d(\ell, c')$ may not be a
  good estimate; we guess whether this is the case, with the bit $b$
  being \texttt{False} if it is a poor estimate. If indeed the nearest center is much too close to $\ell$, we can safely delete $c'$ (line \ref{line:r-add2}). Since $\ell$ and $c'$ are so close they are roughly equivalent in their capacity to serve outstanding centers in $\OPT_2 \setminus \OPT_X$.
  
  We repeat this process of guessing whether $d(\ell, c')$ is a good radius estimate until eventually it becomes one. At that point, we repeatedly guess the correct radii and add dummies for every remaining center in $\OPT_2 \setminus \OPT_X$ for which $\ell$ is a good leader (lines \ref{line:radius-guessing-loop} and \ref{line:dummy-addition}). Afterwards we remove $\ell$ since it is then no longer needed (line \ref{line:r-add3}).

  We now proceed with the formal analysis, first showing
that~\Cref{alg:expensive-sampling+dummy} runs in polynomial
  time. The two nested for loops have $2s_0$ iterations total (Line~\ref{line:inner-big-while}). For the
  dummy center guessing loop (Line~\ref{line:dummy-guessing-loop}), in
  each iteration of the loop, if $b$ is set to True, $|\Lambda_2|$ and
  $|\remls|$ each increase by at least $1$. Otherwise, if $b$ is set to
  False, $|\remls|$ increases by $1$. Either way,
  $|\Lambda_2| + |\remls|$ increases by at least $1$, so this loop occurs at
  most $O(s_0)$ times over the course of the whole algorithm. The last loop
  is the radius guessing loop (Line~\ref{line:radius-guessing-loop})
  in which multiple dummy centers might be added from the same
  $\ell$. In each iteration of this loop $|\Lambda_2|$ increases by $1$,
  so this loop occurs at most $O(s_0)$ times over the entire course of the
  algorithm. Each individual line of the algorithm can be implemented
  in polynomial time, yielding the desired runtime.

  It will be convenient to also define:
  \begin{equation*}
    \Gamma_{\eps}(c^*) := \Big\{c \in \LS \mid d(c, c^*) \leq (1 + \newkappa +
    3\eps)\per(C^*), d(c, C^*) \leq
    2\eps\cdot \per(C^*) \Big\}. 
  \end{equation*}
  We now show that the output of~\Cref{alg:expensive-sampling+dummy}
  has all of the desired properties with probability
  $\exp\{-O(s_0 \log(\eps^{-1}))\}$. Over the course of the algorithm, only for the purposes
  of analysis, we maintain a current candidate partition
  $\OPT_2 \setminus \OPT_X= \widetilde{\OPT}_{\Lambda_2} \sqcup
  \widetilde{\OPT}_B$ and a bijection $\varphi_2: \widetilde{\OPT}_{\Lambda_2} \to \Lambda_2$ as
  follows. We will explicitly set $\widetilde{\OPT}_{\Lambda_2}$ and set $\widetilde{\OPT}_B$ as the complement (in $\OPT_2 \setminus \OPT_X$). Initially $\widetilde{\OPT}_{\Lambda_2} = \emptyset$. The final values of $\widetilde{\OPT}_{\Lambda_2}$ and $\widetilde{\OPT}_B$ are $\OPT_{\Lambda_2}$ and $\OPT_B$. For \emph{valid}
  instantiations of the algorithm, we maintain two invariants:
    \begin{enumerate}[label=(\Roman*)]
  \item \label{inv:B_K^1} Whenever $\lambda$ is added to $\Lambda_2$ such
    that $d(\ell, \lambda) = r$, we move some
    center $c^*_0$ from $\widetilde{\OPT}_B$ to $\widetilde{\OPT}_{\Lambda_2}$ and update
    $\varphi_2$ accordingly with $\varphi_2(c^*_0) = \lambda$ with
    $\ell_{\lambda} = \ell$ and $r_{\lambda} = r$. Each time this occurs, we have $d(c^*_0, \ell) \leq d(\ell, \lambda) \leq (1 + \kappa(C^*_0) + 6\eps) \per(C^*_0)$.
  \item \label{inv:redundant} At the beginning and end of each iteration of the for loop on line~\ref{line:inner-big-while}, $\LS_X \cup \Gamma(\widetilde{\OPT}_B)\subseteq \LS \setminus \remls.$
  \end{enumerate}
  By \Cref{prop:sync-sub} and the definition of $\LS_X$, we have
  $|\LS\setminus\LS_X|\leq s_0$. Thus, we maintain that
  $|\Lambda_2|\leq |\OPT_2\setminus\OPT_X|\leq s_0$ and $|\remls|\leq s_0$.
  In the good-event analysis below we also show that centers from $\LS_X$ are never deleted, so none of the size checks prevents a good iteration from being completed.
  Note that invariant \ref{inv:B_K^1} implies \ref{cond:expensive-dummy} and invariant~\ref{inv:redundant} immediately gives $\Gamma(\OPT_B)\subseteq\LS'$ at termination.

  We next prove the cost and probability bounds. For this, we begin with two helper claims. First we show that if the satellites of two clusters overlap, they have comparable values of $\per(\cdot)$.
 \begin{claim} \label{cla:comparable-per}
Suppose that $c^*, c_1^* \in \OPT_2$ satisfy $\Gamma_{\eps}(c^*) \cap \Gamma_{\eps}(c_1^*) \neq \emptyset$. Then, $\per(C^*)/3 \leq \per(C_1^*) \leq 3 \per(C^*)$.     
 \end{claim}
 \begin{proof}
Let $c \in \Gamma_{\eps}(c^*) \cap \Gamma_{\eps}(c_1^*)$. By definition of $\Gamma_{\eps}(\cdot)$, there exist witness clients $p, p_1 \in D$ such that $p \in C^* \in \calO$, $p_1 \in C_1^* \in \calO$, $d(p, c) \leq 2\eps \per(C^*)$, and $d(p_1, c) \leq 2 \eps \per(C^*_1)$. We also know that $d(c^*, c) \leq (1 + \newkappa + 3\eps)\per(C^*)$ and $d(c^*_1, c) \leq (1 + \kappa(C_1^*) + 3\eps)\per(C_1^*)$. Since $\calO$ is induced by the nearest-center assignment $\pi_{\OPT}$ on $\OPT$, we must have that $d(p, c^*) \leq d(p, c_1^*)$ and $d(p_1, c_1^*) \leq d(p_1, c^*)$. Now, assume for the sake of contradiction that $\per(C_1^*) < \per(C^*)/3$. Then, $d(c_1^*, p) \leq d(c_1^*, c) + d(c,p) \leq (1 + \kappa(C_1^*) + 3 \eps)\per(C_1^*) + 2\eps \per(C^*) < (2/3 + 3\eps) \per(C^*)$. On the other hand, we know that $d(c^*, p) \geq d(c^*, c) - d(p, c) \geq (\eta/(1+2\eps) - 2\eps) \per(C^*)$ since $c^* \not\in \OPT_{S_0}$ and $\per(C^*)\leq(1+2\eps)\avg_{\OPT}(C^*)$. This is a contradiction since $\eta \geq 0.7$ and $\eps<1/200$. Exchanging $C^*$ and $C_1^*$ proves the other inequality.
 \end{proof}
 We now use invariant~\ref{inv:redundant} to conclude \ref{cond:expensive-post-rem-cost-bd}. We actually prove a stronger statement, which we will require in our probability analysis.
\begin{claim} \label{cla:intermediate-cost-bd}
Assuming the invariants, before and after every iteration of the for loop on line \ref{line:inner-big-while}, we have $\cost((\LS \setminus \remls) \cup \Lambda_1 \cup \Lambda_2) \leq (5 + 50 \eps) \sopt < 6 \sopt$. 
\end{claim}
\begin{proof}
 The invariant \ref{inv:redundant} ensures that $\LS_X \subseteq \LS \setminus \remls$, so the cost bound from \Cref{prop:cost-bound-for-exceptional} still applies. By definition, $\psi_X(\OPT_X) \subseteq \LS_X$. Then, combining this with the triangle inequality, the distance bound from invariant \ref{inv:B_K^1}, and invariant~\ref{inv:redundant} we get the following: for each $p \in C^* \in \calO$ with $c^* \in \OPT_2$, there exists some $c \in (\LS \setminus \remls) \cup \Lambda_1 \cup \Lambda_2$ with $d(p, c) \leq d(p, c^*) + (2 + 2\newkappa + 12\eps) \per(C^*)$. Summing the costs of these clients and using $\per(C^*) \leq (1 + 2 \eps) \avg_{\OPT}(C^*)$ for each $C^* \in \calO$ and combining with \Cref{prop:cost-bound-for-exceptional}, we get $\cost((\LS \setminus \remls) \cup \Lambda_1 \cup \Lambda_2) \leq (5 + 50 \eps) \sopt < 6 \sopt$, as needed.
\end{proof}

  It remains to show that, with probability $\exp\{-O(s_0\log(\eps^{-1}))\}$, an instantiation of the algorithm is valid (i.e., satisfies the invariants) and $\OPT_B$
  satisfies the cost bound requirements of \ref{cond:expensive-rem-cost-bds} at termination.

  The following events are called \emph{good}. Assume the invariants hold inductively at the beginning of the current outer iteration.
\begin{enumerate}
    \item If either:
    \begin{enumerate}
        \item $i = 0$ and $\sum_{c \in \Gamma(\widetilde{\OPT}_B)} \sum_{p \in C} d(p, \LS) < \eps^3 \sopt$, or
        \item $i = 1$ and $\sum_{c \in \Gamma(\widetilde{\OPT}_B) \cup (\LS \setminus (\LS_X \cup \remls))} \sum_{p \in M_c} d(p, (\LS \setminus \remls) \cup \Lambda_1 \cup \Lambda_2) < \eps^3 \sopt$,
    \end{enumerate}
    then we skip the entire iteration via Line~\ref{line:skip}. Otherwise we do not skip the iteration. (This occurs with probability $1/2$.)
    \item If $i = 0$ and the iteration is not skipped, we sample $\ell \in \Gamma(\widetilde{\OPT}_B)$. This occurs with probability at least $\eps^3/5$ since $\cost(\LS) \leq 5 \sopt$.
    \item If $i = 1$ and the iteration is not skipped, we sample $\ell \in \Gamma(\widetilde{\OPT}_B) \cup (\LS \setminus (\LS_{X} \cup \remls))$ with probability at least $\eps^3/6$ since $\cost((\LS \setminus \remls) \cup \Lambda_1 \cup \Lambda_2) < 6 \sopt$ by \Cref{cla:intermediate-cost-bd}. \label{cond:prob-sampling2}

    \item After either sampling step, we add $\ell$ to $\remls$ on
    line~\ref{line:r-add1} if and only if
    $\ell\notin\Gamma(\widetilde{\OPT}_B)$. In this case,
    $\ell\notin\LS_X$ by the sampling event, so both invariants are
    maintained. If $\ell\in\Gamma(\widetilde{\OPT}_B)$, we instead
    continue to the dummy guessing loop. In either case, the desired outcome occurs
   with probability $1/2$.

    \item Otherwise, the sampled $\ell$ belongs to $\Gamma(\widetilde{\OPT}_B)$. Note that, by definition of $\OPT_X$, $\Gamma_{\eps}(\widetilde{\OPT}_B) \cap \LS_X = \emptyset$, so deleting $\ell$ (as we will eventually do) will not delete a center in $\LS_X$.
    Fix $c_0^*\in\widetilde{\OPT}_B$ with $\ell\in\Gamma(c_0^*)$.
    In each iteration of the dummy guessing loop, write $r_0=d(\ell,c')$.
    The bit $b$ is set to True if and only if there exists
    $c^*\in\widetilde{\OPT}_B$ with $\ell\in\Gamma_\eps(c^*)$ and
    $r_0\geq\eps^2\per(C^*)$. Each required outcome has probability $1/2$.
    We verify the two cases.
    \begin{enumerate}
        \item If $b$ is False, then $r_0<\eps^2\per(C_0^*)$.
        By the triangle inequality, $c'\in\Gamma_\eps(c_0^*)$,
        so $c'\notin\LS_X$. Moreover, if
        $c'\in\Gamma(c^*)$ for any $c^*\in\widetilde{\OPT}_B$,
        \Cref{cla:comparable-per}, applied at $c'$, gives
        $\per(C_0^*)\leq3\per(C^*)$. Consequently,
        $d(\ell,c')<3\eps^2\per(C^*)\leq\eps\per(C^*)$.
        The witness for $c'\in\Gamma(c^*)$ and the triangle
        inequality then show that $\ell\in\Gamma_\eps(c^*)$.
        Thus every remaining cluster whose satellite is deleted
        can use $\ell$ as a leader.

        To see that $b$ is eventually True, note that
        $\Gamma(c_0^*)$ has diameter at least
        $\eps\per(C_0^*)$ by \Cref{lem:no-dense-cores}. There is therefore some
        $s\in\Gamma(c_0^*)$ with
        \[
          \eps\per(C_0^*)/2\leq d(\ell,s)<5\per(C_0^*).
        \]
        Invariant~\ref{inv:redundant} ensures that $s$ is present at
        the beginning of this outer iteration. It is never deleted
        while $b$ is False, since each such deletion has distance
        less than $\eps^2\per(C_0^*)$ from $\ell$.
        Hence the loop eventually reaches a True outcome, and at
        that point $r_0\leq d(\ell,s)<5\per(C_0^*)$.
        All preceding deletions avoid $\LS_X$, so the size checks
        do not prevent this.

        \item When $b$ is True, the definition of the good
        outcome and \Cref{cla:comparable-per}, applied at $\ell$,
        give, for every $c^*\in\widetilde{\OPT}_B$ with
        $\ell\in\Gamma_\eps(c^*)$,
        \[
          \frac{\eps^2}{3}\per(C^*)\leq r_0\leq15\per(C^*).
        \]
        Also,
        \[
          \frac{\eta}{1+2\eps}\per(C^*)<d(c^*,\ell)
          \leq(1+\newkappa+3\eps)\per(C^*),
        \]
        since $c^*\notin\OPT_{S_0}$. For $\eps<1/200$, these bounds
        ensure that rounding $d(c^*,\ell)$ up to a power
        $(1+\eps)^x r_0$ gives a radius in
        $[\eps^3r_0,\eps^{-3}r_0]$.

        Let $c^*$ be any center in $\widetilde{\OPT}_B$ with $\ell \in \Gamma_{\eps}(c^*)$. We implicitly set $\varphi_2(c^*) = \lambda$ in this iteration. The random guess of $r$ is such that $d(c^*, \ell) \leq r \leq (1 + \eps) d(c^*, \ell)$. This happens with probability at least $\frac{\eps}{20 \log (1/\eps)}$. Also note that, since $\ell \in \Gamma_{\eps}(c^*)$, we have $r \leq (1 + \eps)(1 + \newkappa + 3\eps)\per(C^*) \leq (1 + \newkappa + 6\eps) \per(C^*)$. Thus it is valid to add $c^*$ to $\widetilde{\OPT}_{\Lambda_2}$.

        The radius guessing loop continues if and only if
        there is still some $c^*\in\widetilde{\OPT}_B$ with
        $\ell\in\Gamma_\eps(c^*)$; each required outcome has
        probability $1/2$. Each dummy is matched to a distinct
        center, preserving invariant~\ref{inv:B_K^1}.
        Once all these centers have been moved to
        $\widetilde{\OPT}_{\Lambda_2}$, we break the radius guessing
        loop and delete $\ell$ on line~\ref{line:r-add3}.
    \end{enumerate}
    Every cluster that lost a satellite during this outer
    iteration has now been moved to $\widetilde{\OPT}_{\Lambda_2}$.
    No center in $\LS_X$ was deleted. This proves
    invariant~\ref{inv:redundant} at the end of the iteration.
    Throughout its processing, $|\remls|\leq|\LS\setminus\LS_X|\leq s_0$
    and $|\Lambda_2|\leq|\OPT_2|\leq s_0$, so all the required
    iterations satisfy the size checks.
\end{enumerate}

We now prove \ref{cond:expensive-rem-cost-bds}. In each phase, a
good iteration is skipped precisely when the corresponding cost
bound holds. Otherwise, it deletes at least one previously
undeleted center of $\LS\setminus\LS_X$. Since there are at most
$s_0$ such centers, either the bound holds before the phase ends,
or all these centers have been deleted by its end.  In the latter
case, both sums are zero, since then invariant~\ref{inv:redundant} and
the definition of $\OPT_X$ give
\[
\Gamma(\widetilde{\OPT}_B)
\subseteq \LS\setminus(\LS_X\cup\remls)=\emptyset.
\]
Thus each phase establishes its required bound. The first bound
continues to hold during the second phase because
$\widetilde{\OPT}_B$ only shrinks.

Each random choice has a success probability of at least $\eps^4$, and
the total number of random choices is at most $O(s_0)$. The bound on
the number of random choices comes from the range of $j$ and the fact
that random choices in the dummy guessing while loop
(Line~\ref{line:dummy-guessing-loop}) are accompanied by increases to
$|\Lambda_2| + |\remls|$ which is bounded by $2s_0$. Hence, the
probability of a good event happening every time it is possible is at
least $\exp\{-O(s_0\log(\eps^{-1}))\}$. This completes the proof of
\Cref{thm:expensive}.
\end{proof}

%% file: new-submod.tex
\section{Reduction to Submodular Maximization: Proof of \Cref{thm:simple-submod}}
\label{sec:submodular-section}

Recall that we have $\LS' \sse \LS$, whose centers have been colored
red/blue. Our desired solution is to add $\OPT_\Lambda \cup \OPT_P$ to
$\LS'$ and drop a set $\ZZ^* \sse \LS' \setminus \LS_\sigma$ (which is
also contained in $\LSred$, by the color-coding) of size
\margindef{$z$} $z := (|\LS'| + |\OPT_\Lambda \cup \OPT_P| - k)_+$
from it, to get the target solution
\[ \MM^* := (\LS' \setminus \ZZ^*) \cup \OPT_\Lambda \cup \OPT_P. \]
This is computationally intractable; however, we now use submodular
maximization to find a set $\widehat\YY$ to add, and
$\widehat\ZZ \sse \LSred$ with $|\widehat\ZZ| = z$ to drop, such that
the solution $\TT := (\LS' \setminus \widehat\ZZ)\cup \widehat\YY$
satisfies
\[ \cost(\TT) \leq (1-\nf1e) \cost(\MM^*) + \nf1e
  \cost(\MM_\Lambda^*) + \eps \, \sopt, \] for an ``intermediate'' solution
$\MM_\Lambda^* := (\LS' \setminus \ZZ^*) \cup \Lambda$.  We finally
bound the expression on the right, 
thereby finishing the proof of~\Cref{thm:simple-submod}.

\subsection{The Submodular Function: Definition and Optimization}
\label{sec:def-submod}

To define the ``improvement'' function $f$, we begin with some
helper functions:
\begin{enumerate}
\item The first function $g_1$ \margindef{$g_1$} is simply the
  cost
  \begin{gather}
    g_1(\YY) := \cost(\MM' \cup \YY).
  \end{gather}
  Recall that $\MM' := \LS' \cup \Lambda$. The function $g_1$ is
  monotone non-increasing and supermodular due
  to~\Cref{cor:dist-supmod}.
  
\item The second function $g_2$ \margindef{$g_2$} is defined as follows:
  \begin{gather}
    g_2(\YY)=\min_{\ZZ\subseteq\LSred: 
        |\ZZ|=z}
    \sum_{c\in\ZZ}\left[\cost_{\MM'}(M_c)
      +|M_c|\,d(c,\LSblue\cup\Lambda\cup\YY)\right]. \label{eq:g2-defn}
  \end{gather}
  This function has the minimum-of-minimum form given in
  \Cref{cla:minofmin}, and hence is monotone supermodular. Also, given
  any $\YY$, the function is efficiently computable, since it involves
  picking the smallest $z$ entries.
\end{enumerate}  
Finally, define \margindef{$f, g$} $g(\YY) =g_1(\YY)+g_2(\YY)$ and
$f(\YY)=g(\varnothing)-g(\YY)$. It follows that $f$ is monotone
non-decreasing, non-negative, submodular; moreover, it is efficiently
computable on any $\YY$. We now show that the function $f$ can be
approximately maximized with constraints that render the dummy centers
in $\Lambda$ redundant.

\begin{lemma} \label{lem:lifting-app} There exists a polynomial time
  algorithm that, given values $z$ and $|\OPT_P|$, outputs
  a set $\widehat{\YY} \subseteq
  \FF$ 
such that
\begin{enumerate}[label = (\alph*),nosep]
    \item $f(\widehat{\YY}) \geq (1 - \nf{1}{e})f(\OPT_{\Lambda} \cup \OPT_P)$, \label{cond:lift-submod-approx} 
    \item $|\widehat{\YY}| \leq |\Lambda| + |\OPT_P|$, and \label{cond:lift-submod-size}
    \item $|\widehat{\YY} \cap (\FF \cap B(\ell_{\lambda}, r_{\lambda}))| \geq 1$ for
  all $\lambda \in \Lambda$. \label{cond:lift-submod-balls}
\end{enumerate}
\end{lemma}
\begin{proof}
  We want to find a $(1 - \nf{1}{e})$-approximation of the submodular
  function $f$ according to constraints enforcing opening one center
  per ``dummy ball,'' while also allowing $\OPT_\Lambda \cup \OPT_P$
  as a feasible solution. To do this formally, we work in a lifted
  space. Let
  $\overline{\UU} \subseteq \FF \times (\Lambda \cup \{\bot\})$ be
  defined as follows:
  \begin{itemize}[nosep]
    \item for each $\lambda \in \Lambda$, $(c, \lambda) \in \overline{\UU}$ if and only if $c \in \FF \cap B(\ell_{\lambda}, r_{\lambda})$, and
    \item $(c, \bot) \in \overline{\UU}$ for all $c \in \FF$.
\end{itemize}
Define 
the projection map $\proj: \overline{\UU} \to \FF$, given by $\proj(c, x)
= c$ for any $(c, x) \in \overline{\UU}$; this naturally extends to
$\proj(\overline{\YY}) := \{\proj(c,x) \mid (c,x) \in
\overline{\YY}\}$. We define the \emph{lifted function} $f_\uparrow :
2^{\overline{\UU}} \to \RR_{\geq 0}$ where, for $\overline{\YY}
\subseteq \overline{\UU}$, $f_\uparrow(\overline{\YY}) =
f(\proj(\overline{\YY}))$. By \Cref{cla:lifting-sub}, $f_\uparrow$ is
monotone non-decreasing, non-negative, submodular; moreover, since
both $f$ and the projection operator are efficiently computable, we
get the same for $f_\uparrow$.

We now define our partition matroid $H(\overline{\UU}, \calH)$. A set
$\overline{\YY} \subseteq \overline{\UU}$ belongs to $\calH$ (i.e., is
an independent set of the matroid), if (1) for each
$\lambda \in \Lambda$, there is at most one
$(c, \lambda) \in \overline{\YY}$, and (2) there are at most
$|\OPT_P|$-many $(c, \bot) \in \overline{\YY}$. Since $H$ is a
(partition) matroid, we can approximately maximize $f_\uparrow$ over $\calH$.

Crucially, note that we defined $\Lambda$ so that, for each
$\lambda \in \Lambda$, we have
$d(\varphi^{-1}(\lambda), \ell_{\lambda}) \leq r_{\lambda}$; i.e.,
$\varphi^{-1}(\lambda) \in B(\ell_{\lambda}, r_{\lambda})$. Define the
set
$\overline{\OPT} := \{ (\varphi^{-1}(\lambda), \lambda) \mid \lambda
\in \Lambda\} \cup \{(c^*, \bot) \mid c^* \in \OPT_P\}$; observe that
$\overline{\OPT} \in \calH$ (and is, in fact, a matroid base), and
moreover,
$f_\uparrow(\overline{\OPT}) = f(\proj(\overline{\OPT})) =
f(\OPT_\Lambda \cup \OPT_P)$.  Hence, applying the algorithm
of~\cite{calinescu2011maximizing}, we can efficiently compute a
\emph{base} $\overline{\YY} \in \calH$ such that
\begin{equation}
f(\proj(\overline{\YY})) = f_\uparrow(\overline{\YY}) \geq (1 - \nf{1}{e})  \max_{\overline{\XX} \in \calH}  f_\uparrow(\overline{\XX}) \geq (1 - \nf{1}{e}) f_\uparrow(\overline{\OPT}) = (1 - \nf{1}{e}) f(\OPT_\Lambda \cup \OPT_P). \label{eq:lift-submod-costs}
\end{equation}
We output $\widehat{\YY} := \proj(\overline{\YY})$, so
\Cref{eq:lift-submod-costs} immediately implies property
\ref{cond:lift-submod-approx}. Moreover, the fact that
$\overline{\YY}$ is a base of $H$ means that there is exactly one
$(c, \lambda) \in \overline{\YY}$ for each $\lambda \in \Lambda$.
Then, each such $c$ belongs to
$\proj(\overline{\YY}) = \widehat{\YY}$, so property
\ref{cond:lift-submod-balls} is satisfied as well. Finally, since
there are at most $|\OPT_P|$-many $(c, \bot) \in \overline{\YY}$, we
also have property \ref{cond:lift-submod-size}.
\end{proof}

\subsection{A $k$-Median Solution from Submodularity}
\label{sec:calling-submod-opt}

To get a $k$-median solution, the overall algorithm is simple:
\begin{enumerate}

\item Guess $|\OPT_P|$ (this determines $z$ since $|\OPT_\Lambda| =
  |\Lambda|$ is known). Since $|\OPT_P| \leq s_0$, we can enumerate the $O(s_0)$ guesses and keep the best
  feasible set produced by the steps below.

\item Use the algorithm of~\Cref{lem:lifting-app} to obtain a set
  $\widehat{\YY} \subseteq \FF$ with the given properties. Expanding
  the definition of $f(\cdot) = g(\emptyset) - g(\cdot)$ and
  rearranging, we get
  \begin{gather}
    g(\widehat{\YY}) \leq (\nf{1}{e})\,g(\emptyset) + (1 - \nf{1}{e})\,g(\OPT_\Lambda \cup \OPT_P). \label{eq:g-guarantee}
  \end{gather}

\item Let $\widehat\ZZ$ be the minimizing set in the
  definition of $g_2(\widehat{\YY})$. Return \margindef{$\TT$}
  \begin{gather}
    \TT := (\LS' \setminus \widehat\ZZ) \cup \widehat\YY
  \end{gather}
  as the $k$-median solution. Since
  $|\widehat\ZZ| = z = (|\LS'| + |\OPT_\Lambda \cup \OPT_P| - k)_+$,
  and $|\widehat\YY| \leq |\OPT_P \cup \OPT_\Lambda|$, we are
  guaranteed that $|\TT| \leq k$.
\end{enumerate}

\begin{lemma}\label{lem:g(S) bd}
  The cost of the new solution $\TT$ satisfies:
  \begin{align*}
    \cost(\TT)
    &\leq g(\widehat\YY) 
    \leq (1/e)\cdot g(\emptyset)
      + (1-1/e)\cdot g(\OPT_{\Lambda} \cup \OPT_P). 
  \end{align*}
\end{lemma}
\begin{proof}
  To bound the $k$-median cost of the solution $\TT$:
  \Cref{lem:lifting-app}\ref{cond:lift-submod-balls} ensures that for
  any point $q$ in the metric space and any dummy center
  $\lambda \in \Lambda$, there is a center $c_\lambda \in \widehat\YY$
  satisfying $d(q,c_\lambda) \leq d(q,\lambda)$. Since
  $\widehat\YY \sse \TT$, we have $d(q,\TT) = d(q,\TT \cup
  \Lambda)$. Now we relate the cost to the constituent portions of
  $g$:
  \begin{enumerate}
  \item For client $p \in M_c \in \calM'$ such that the center $c$
    belongs to $\MM' \setminus \widehat\ZZ$, center $c$ still belongs
    to $\TT \cup \Lambda$. This means
    $d(p,\TT) = d(p, \TT \cup \Lambda) \leq d(p, \MM' \cup
    \widehat\YY)$; this can be charged to $g_1(\widehat\YY)$.
    
  \item For $p \in M_c \in \calM'$ with $c \in \widehat\ZZ$, we have
    \[ d(p, \TT) = d(p, \TT\cup \Lambda) \leq d(p, \LSblue \cup \widehat\YY
      \cup \Lambda) \leq 
      d(p, c) + d(c, \LSblue \cup \widehat\YY \cup \Lambda), 
    \]
    where we use that $\LSblue \cup \widehat\YY \sse \TT$, and the
    triangle inequality. The sum of the final expression over the
    clients in $\cup_{c \in \widehat\ZZ} M_c$ is exactly
    $g_2(\widehat\YY)$.
  \end{enumerate}
  Summing, we get $\cost(\TT) \leq g(\widehat\YY)$;
  the second inequality of \Cref{lem:g(S) bd} follows
  from~(\ref{eq:g-guarantee}).
\end{proof}

\begin{lemma} \label{lem:g-bds}
We have 
$g(\emptyset) \leq \cost(\MM^*_{\Lambda}) + O(\eps) \sopt$, and 
 $g(\OPT_{\Lambda} \cup \OPT_P) \leq \cost(\MM^*) + O(\eps) \sopt$.
\end{lemma}
\begin{proof}
  Recall from \S\ref{sec:color-coding-submod} that
  $\ZZ^* \sse \LS'\setminus \LS_\sigma$ is a feasible set for the
  minimization in $g_2$'s definition; indeed, we defined
  $\ZZ^* \sse \LSred$ and $|\ZZ^*| = z$. Consequently, choosing it in
  $g_2$ gives an upper bound:
  \begin{gather}
    g(\YY) \leq \underbrace{\sum_{p \in D} d(p, \LS' \cup \YY \cup
      \Lambda)}_{= g_1(\YY)} \;+\; \underbrace{\sum_{c \in \ZZ^*}
      \sum_{p \in M_c} 
      (d(p,c) + d(c, \LSblue \cup \YY \cup \Lambda))}_{\geq g_2(\YY)}. \label{eq:g-first}
  \end{gather}
  Moreover, we guaranteed in our color-coding step from \S\ref{sec:color-coding-submod} that
  $\pi_{\LS'\setminus \ZZ^*}(c) \in \LSblue$ for all $c \in \ZZ^*$;
  therefore $d(c,\LSblue) = d(c,\LS'\setminus\ZZ^*)$. Using this fact
  (and the triangle inequality to infer
  $d(c,\SS) \leq d(c,p) + d(p,\SS)$) in the second sum
  from~(\ref{eq:g-first}), and splitting its first sum based on
  whether the centers belong to $\ZZ^*$ or not, we get:
  \begin{align}
    g(\YY)
    &\leq \sum_{c \not\in \ZZ^*} \sum_{p \in M_c}
    d(p, (\LS' \setminus \ZZ^*) \cup \YY \cup
    \Lambda) \;+\; \sum_{c \in \ZZ^*}
    \sum_{p \in M_c} 
    (3\,d(p,c) + d(p, (\LS' \setminus \ZZ^*) \cup \YY \cup \Lambda)) \notag
    \\
    &= \cost((\LS' \setminus \ZZ^*) \cup \YY \cup
      \Lambda) + 3\cdot \sum_{c \in \ZZ^*} \cost_{\MM'}(M_c).
  \end{align}
  Now using $\YY = \emptyset$ and $\YY = \OPT_\Lambda \cup \OPT_P$,
  the corresponding first terms are exactly $\cost(\MM_\Lambda^*)$,
  and at most $\cost(\MM^*)$ respectively. Finally, since
  $\ZZ^* \sse \LS'\setminus \LS_\sigma$,
  \Cref{thm:struct2}\ref{cond:struct2-redund-hidden} implies that the
  last term is at most $3\eps^3 \sopt \leq \eps \sopt$, completing the proof.
\end{proof}

\begin{lemma} \label{lem:mixed-bd} 
We have 
\begin{align*}
(1-1/e)\cost(\MM^*)+(1/e)\cost(\MM^*_{\Lambda})\leq{}&
(1+\eta)\cost_\OPT(D_S)+\cost_\OPT(D\setminus D_S)\\
&+\frac2e\sum_{c^*\in\olamearly}|C^*|\per(C^*)
+\frac2e\sum_{c^*\in\olamlate}(1+\newkappa)|C^*|\per(C^*)\\
&+\frac1e\sum_{c^*\in\OPT_P}(1+\newkappa)|C^*|\per(C^*)
+\sum_{c^*\in\OPT_R}|C^*|\per(C^*)
+O(\eps)\sopt.
\end{align*}
\end{lemma}

\begin{proof}
For client $p$, define
  $\Phi_p := (1- 1/e)\, d(p,\MM^*) + (1/e)\,
  d(p,\MM^*_{\Lambda})$.  
  First, consider clients in $D_\Lambda$. Let $\pi_\OPT(p) = c^*$, and
  the corresponding cluster be $C^*$. Since $\OPT_\Lambda \sse \MM^*$,
  we get $d(p, \MM^*) \leq d(p, c^*)$ for any $p \in D_\Lambda$;
  moreover, the triangle inequality gives
  $d(p, \MM^*_\Lambda) \leq d(p, c^*) + d(c^*, \varphi(c^*))$.
  In either case, we can use
  \Cref{thm:structural}\ref{cond:struct-dummies} to bound
  \begin{align}
   \sum_{p\in D_\Lambda}\Phi_p
&\leq\cost_\OPT(D_\Lambda)
 +\frac2e\sum_{c^*\in\OPT_\Lambda}|C^*|r_{\varphi(c^*)} \nonumber \\
 &\leq \cost_\OPT(D_\Lambda) + \frac2e\sum_{c^*\in\olamearly}|C^*|\per(C^*)
+\frac2e\sum_{c^*\in\olamlate}(1+\newkappa)|C^*|\per(C^*) + O(\eps) \sopt.\label{eq:mixed-1}
  \end{align}

  For clients in $D_P$, using the map $\psi$
  instead of $\varphi$ and also
\Cref{thm:structural}\ref{cond:struct-backups} gives 
  \begin{align}
   \sum_{p\in D_P}\Phi_p
&\leq \cost_\OPT(D_P) + \frac1e\sum_{c^*\in\OPT_P}(1+\newkappa)|C^*|\per(C^*) + O(\eps) \sopt.\label{eq:mixed-2}
  \end{align}


  For each client $p$ in $D_R \cup D_S$, the substitution map
  sends its optimal center $c^*$ to some point $\sigma(c^*)$ in
  $\LS_\sigma$, which is disjoint from $\ZZ^*$. Since
  $\LS_\sigma \sse \MM^* \cap \MM_\Lambda^*$,  we can
  bound $\Phi_p \leq d(p, \sigma(c^*))$ for each of these clients. 
  Using~\Cref{thm:structural}\ref{cond:struct-subs}\ref{cond:struct-substitution-strong} gives
  \begin{align}
    \sum_{p \in D_S} \Phi_p &\leq (1 + \eta)
                              \cost_{\OPT}(D_S) + O(\eps) \sopt, \quad
                              \label{eq:mixed-3}
  \end{align}
 whereas 
    \Cref{thm:structural}\ref{cond:struct-subs}\ref{cond:struct-substitution-weak}
    and the triangle inequality give
     \begin{align}
  \sum_{p\in D_R}\Phi_p
&\leq\sum_{c^*\in\OPT_R}\sum_{p\in C^*}d(p,\sigma(c^*))  \leq \cost_{\OPT}(D_R) + \sum_{c^*\in\OPT_R}|C^*| \per(C^*) + O(\eps) \sopt. \label{eq:mixed-4}
  \end{align}
  Summing (\ref{eq:mixed-1})-(\ref{eq:mixed-4}) completes
  the proof of the lemma.
\end{proof}

Combining \Cref{lem:g(S) bd} and \Cref{lem:g-bds} gives the first
inequality in \Cref{thm:simple-submod}. The second inequality follows from \Cref{lem:mixed-bd}, which wraps up this section.


%% file: new-local.tex
\section{Local Search Revisited: Proof of \Cref{thm:local-search-bd}}
\label{sec:local-revisited}

\newcommand{\pair}{\mathrm{pair}}
\newcommand{\unpair}{\mathrm{stdom}}

Next, we prove the nuanced local-search guarantee in
\Cref{thm:local-search-bd}, which shows that if we have a cross-free
pair $(\LS,\mu_{\AA}), (\BB,\mu_{\BB})$ with respect to $\lsd$, where $\LS$ is a solution to
local-search with $O(1/\eps)$-swaps, and $\BB$ is a benchmark
solution, then the cost of the local-search solution (using the nearest-center
map $\pi_\LS$) can be related to the cost of the clustering $\mu_{\BB}$,
plus the gap between the costs of $\mu_{\AA}$ and $\pi_\LS$ (plus lower
order terms).

\localsearch*



It will be useful to recall the definition of generalized
cross-freeness (see \Cref{fig:gen-cross-free}):

\GenCF*

Consider the clusterings $\calA, \calB$ of $D$ induced by
$(\LS,\mu_\AA), (\BB, \mu_\BB)$; since they are cross-free with
respect to $\lsd$, \Cref{defn:gen-cross-free} tells us that
induced clusterings $\calA', \calB'$ of the subset
$\lsdres := \cup_{A \in \calA'} A$ are cross-free. 
(Both $\calA', \calB'$ are partitions of $\lsdres$, though both of them
may contain one or more empty sets.)  Define
$\AA' := \LS \setminus \mu_{\AA}(D \setminus \lsd)$ to be the centers
of clusters in $\calA'$. Assumption~(i)
from~\Cref{thm:local-search-bd} along with
$|\AA'| + |\mu_{\AA}(D \setminus \lsd)| = k$ means
$|\BB| \leq |\AA'|$.

Call a cluster $K \sse \lsdres$ belonging to either of the clusterings
$\calA', \calB'$ \emph{dominating} if it is non-empty and a
disjoint union of clusters of the other clustering; let the
\emph{refinement map} $\Pi(K)$ denote the clusters in the other
clustering that partition $K$. A non-empty cluster is \emph{dominated}
if it is contained within a cluster of the other clustering. 
Cross-freeness implies that a cluster from $\calA'$ or
$\calB'$ is empty, dominating, or dominated. A cluster that is
both dominating and dominated is called \emph{paired}; such a cluster
is non-empty and exists in both clusterings. A non-paired dominating
cluster is called \emph{strictly dominating}.  Let
$\calB'_{\pair} = \calA'_{\pair}$ be the paired clusters; let
$\calA'_{\unpair}$ and $\calB'_{\unpair}$ be the strictly dominating
clusters in $\calA'$ and $\calB'$ respectively. Observe that the
clusters in
$\calA'_{\pair} \sqcup \calA'_{\unpair} \sqcup \calB'_{\unpair}$ form
a partition of $\lsdres = \bigcup_{A \in \calA'} A$.

Overloading notation, for any collection of dominating clusters
$\calA'' \sse \calA_{\unpair}'$, we define
$\Pi(\calA'') := \bigcup_{A_c \in \calA''} \Pi(A_c)$, and define
$\Pi(\calB'')$ similarly for $\calB'' \sse \calB_{\unpair}'$.  We
partition
$\calA_{\unpair}' = \calA_{\mathrm{big}}' \sqcup \calA_{\mathrm{small}}'$
into \emph{big} and \emph{small} clusters, where
$\calA_{\mathrm{big}}'$ contains clusters that dominate at least
$2/\eps$ clusters in $\calB'$. Finally, let the nearest-center map
$\pi_\LS$ for the local search solution $\LS$ induce the clustering
$\calL = \{C_c\}_{c \in \LS}$. 


\subsection{A Graph Theoretic Abstraction}
\label{sec:graph-theor-abstr}

In this section we define a graph theoretic abstraction of the problem
of constructing a suitable collection of swaps for our local search
analysis. We then prove that the existence of ``spanning collections''
of graph swaps implies the existence of collections of $k$-median
swaps that certify the nuanced local-search guarantees from
\Cref{thm:local-search-bd}. We finally show how to construct such a
spanning collection of graph swaps in \Cref{lem:graph-swaps}.

Suppose we have a bipartite graph $G = (L, R, E)$, where
$n = |L| = |R|$, and $G$ is a disjoint union of stars, and has no
isolated vertices. Call a vertex $v \in L \cup R$ a \emph{leaf} if it
is the leaf of some star and a \emph{hub} otherwise. For stars made up
of a single edge, we consider both vertices to be leaves.  Call a star
component in $G$ a \emph{left-star} if its hub belongs to $L$ and a
\emph{right-star} otherwise. Let $H_L \subseteq L$ and
$H_R \subseteq R$ denote the sets of hubs of left and right stars,
which are called left-hubs and right-hubs, respectively. For a vertex
$v$, let $N(v)$ denote its neighbors in $G$; for a set $S$ of vertices
which belongs to the same side of the bipartition, let
$N(S) = \cup_{v \in S} N(v)$. For a hub $v$, let $S(v) = N(v)$ denote
the leaves of its star.  The star centered at $v$ is \emph{big} if
$|S(v)| \geq 2/\eps$, and \emph{small} otherwise. For $L' \subseteq L$
and $R' \subseteq R$ with $|L'| = |R'|$, call the pair $(L', R')$ a
\emph{swap}.  %

\begin{definition}[Spanning Collections of Swaps]
  \label{defn:span-swaps}
  A collection $\calS = \{(L_1, R_1), (L_2, R_2), \cdots\}$ of
  swaps is \emph{spanning} if the following properties hold:
  \begin{enumerate}[label = (\alph*),nosep]
  \item %
    For each swap
    $(L_i, R_i) \in \calS$:
    $|L_i| = |R_i| \leq 12/\eps$. \label{cond:swap-small}
    \item \label{cond:nbrs} If vertex $v \in L_i$, then $N(v) \sse R_i$. 
    \item \label{cond:alive-repped} 
      Except for hubs of large left stars, each other vertex appears
      in some swap in $\calS$. 
    \item \label{cond:left-hubs} For each hub on the left or leaf on
      the right---i.e., for each vertex
      $v \in H_L \cup (R \setminus H_R)$---there exists at most one
      swap $(L_i, R_i) \in \calS$ with $v \in L_i \cup R_i$.
\end{enumerate}
\end{definition}

Intuitively, $L$ and $R$ will correspond to some of the centers in
$\AA' \sse \LS$ and $\BB$ respectively; a swap $(L_i, R_i)$ will
correspond to closing down centers at $L_i$ and opening those at
$R_i$, which is why we want $|L_i|=|R_i|$ and this size to be
bounded. We want clients near centers in $L_i$ to have a new center
close-by, which is why we ask for $N(L_i) \sse R_i$.  All leaves
should appear in some swap, since that allows us to bound the quality
of local-search solutions. Finally, we require left-hubs to appear
in at most one swap; this allows us to compare the costs of
nearest-facility and custom assignments.

The following lemma shows the existence of a distribution over
spanning collections of swaps such that each leaf appears in
slightly more than one swap in expectation.

\begin{lemma} \label{lem:graph-swaps} There exists a distribution
  $\calD$ over spanning collections of swaps such that each leaf
  $v \in L \setminus H_L$ appears in at most $(1 + 8\eps)$ swaps in
  expectation. Formally, we have
  \begin{equation}
    \mathop{\mathbb{E}}_{\calS \in \calD}\bigg[\sum_{(L_i, R_i) \in \calS} \mathbf{1}\{v \in L_i\}\bigg] \leq (1 + 8\eps).
  \end{equation}
\end{lemma}
Before we prove \Cref{lem:graph-swaps}, let us use it to prove
\Cref{thm:local-search-bd}.

\subsection{Proving the Local Search Guarantee}
\label{sec:proving-LS-guarantee}

\begin{proof}[Proof of \Cref{thm:local-search-bd}]
  Given a cross-free pair $(\LS, \mu_{\AA}),(\BB, \mu_{\BB})$ with respect to $\lsd$
  and corresponding clusterings $\calA'$ and $\calB'$, we can define the
  associated intersection graph on $(\AA', \BB)$, with an edge
  $(c,c^*) \in \AA' \times\BB$ if $A_c \cap B'_{c^*} \neq \emptyset$. The two sides use distinct copies
  of any common center. 
  
  We have two issues to address before being able to apply
  \Cref{lem:graph-swaps}. Some clusters might be empty---$B_{c^*}' \in \calB'$ are restrictions of clusters to the points in
  $\lsdres$, and $A_c \in \calA'$ can be empty to begin with---so the graph might have isolated vertices. Moreover, we
  might have $|\AA'| > |\BB|$. To address these issues, we add some auxiliary vertices: add $|\AA'|-|\BB|$ right vertices,
  then add left super-vertices $\ell_1, \ell_2$ and right
  super-vertices  $r_1, r_2$. Connect every pre-existing isolated
  right-vertex to $\ell_1$ and every pre-existing isolated left-vertex to $r_1$. Also add the edges $(\ell_1, r_2)$ and $(r_1, \ell_2)$. The resulting graph
  $G=(L,R,E)$ has equal sides and is a disjoint union of stars without
  isolated vertices, so \Cref{lem:graph-swaps} applies. 
  
  Consider any swap in the graph $(L_i, R_i) \in \calS$. Discard all
  auxiliary vertices contained in $L_i$ and $R_i$ to obtain a pair
  $(\AA'_i, \BB_i)$ with $\AA'_i \subseteq \AA'$ and
  $\BB_i \subseteq \BB$; define
  the candidate solution to be
  $\LS^{\test} := (\LS \setminus \AA'_i) \cup \BB_i$. By \Cref{defn:span-swaps}\ref{cond:nbrs}, if
  $\ell_1$ or $\ell_2 \in L_i$ then $r_2$ or $r_1 \in R_i$,
  respectively. Hence, at least as many auxiliary right-vertices are
  discarded as left-vertices, and the number of centers $|\BB_i|$
  being added is at most the number $|\AA'_i|$ being deleted, ensuring
  that the new candidate solution has size
  $|(\LS \setminus \AA'_i) \cup \BB_i| \leq k$. (If we insist on
  $\LS^{\test}$ having size exactly $k$, we can add back centers from
  $\AA'_i$, which can only decrease the cost.) Since
  $|\AA'_i|, |\BB_i| \leq |L_i| = |R_i| \leq 12/\eps$ by
  \Cref{defn:span-swaps}\ref{cond:swap-small}, the stopping condition of
the local search algorithm ensures
\begin{equation}
  \cost(\LS^{\test}) \geq
  (1-\eps/k)\cost(\LS).
  \label{eq:local-opt}
\end{equation}
  
  \begin{figure}[t]
    \centering
    \includegraphics[width=3in]{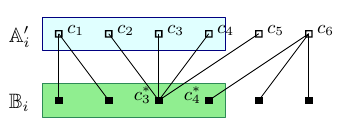}
    \caption{Local Swaps: the colored boxes show one swap
      $(\AA'_i,\BB_i)$ with $4$ centers; the client set
      $D_i$ is precisely $((A_{c_1} \cup A_{c_2} \cup A_{c_3} \cup A_{c_4}) \cup
      B_{c_4^*}) \cap \lsd$. Note that clients in $A_{c_5} \sse B_{c_3^*}$, or in
      $A_{c_6} \setminus B_{c_4^*}$ are not included in $D_i$.}
    \label{fig:swap-to-LS}
  \end{figure}

  We now give an upper bound on the expression on the left-hand side
  by defining a ``test'' assignment of clients to centers in
  $\LS^{\test}$; in turn, this gives an upper bound for
  $\cost(\LS)$. To motivate our assignment, we make a few simple
  observations. (i)~For all $p \in D$,
  $d(p, \pi_{\LS}(p)) \leq d(p, \mu_{\AA}(p))$, because both are
  assignments of clients to $\LS$, and $\pi_{\LS}$ is the
  nearest-center one. (ii)~We can imagine that $\mu_{\BB}$ is a poor
  assignment for clients outside of $\lsd$, given the harsh
  restrictions on $|\BB|$. Finally, (iii)~our bound is non-trivial
  only if $\cost_{\mu_{\BB}}(\lsd) \leq \cost_{\LS}(\lsd)$, in which
  case $\mu_{\BB}$ is better than $\pi_{\LS}$ on $\lsd$ \emph{in
    aggregate}.  We emphasize that, while the observation~(i) holds
  for each client, the observation~(iii) is only an aggregate bound.
  With these observations in mind, we give our ``test'' assignment,
  beginning with the clients assigned to $\BB_i$.

  \begin{enumerate}[label=(\alph*),wide]
  \item By \Cref{defn:span-swaps}\ref{cond:nbrs}, we have that
    $N(L_i) \sse R_i$: this means clients belonging to clusters $A_c$
    with centers $c \in \AA'_i$ can now be assigned to their
    $\mu_{\BB}$-centers, which have been opened in $\BB_i$. There may
    also be some \emph{dominated} clusters $B_{c^*}'$ such that
    $c^* \in \BB_i$, but the center $c$ of its dominating cluster
    $A_c$ is not in $\AA'_i$: clients in $\lsd$ belonging to (the
    unrestricted) $B_{c^*}$ can be ``upgraded'' to their $\mu_{\BB}$
    assignment.  See \Cref{fig:swap-to-LS} for an example of $D_i$.
    (Although we could also assign the rest of the clients in the
    cluster $B_{c^*} \setminus \lsd$ to their $\mu_{\BB}$ assignment,
    we do not---indeed, recall our intuition that such an assignment
    could actually be a downgrade for clients outside of $\lsd$.)
    Set $D_i \subseteq \lsd$ to be the set of clients in either of
    these two cases:
    \[
      D_i := \{p \in \lsd \mid p \in A_c \text{ s.t. } c \in \AA_i'
      \lor p \in B_{c^*} \text{ s.t. } c^* \in \BB_i \text{ and
      } B_{c^*}' \text{ dominated} \}.
    \]

  \item We augment $D_i$ to $D_i^+$ by assigning some more clients to
    their $\mu_{\BB}$ centers. Start with $D_i^+ = D_i$. Then for
    clients $p \in \lsd \setminus \lsdres$ such that $p \in B_{c^*}$
    and whose restriction $B_{c^*}' \in \calB'$ is strictly dominating
    or empty, choose a unique index $i$ such that $c^* \in \BB_i$ and
    add $p$ into $D_i^+$; there exists such a swap since each right
    vertex appears in a swap;
    \Cref{defn:span-swaps}\ref{cond:alive-repped}. (We only assign
    these clients to a single index because
    $\cost_{\mu_{\BB}}(\lsd) \leq \cost_{\LS}(\lsd)$ is an aggregate
    bound, not a client-wise one.)

  \item Next, consider clients in $D \setminus D_i^+$ that are
    assigned to centers in $\AA'_i$ in the \emph{local-search}
    solution $\calL$; call this set
    $D_i' := \{p \in D \setminus D_{i}^+ \mid p \in C \in \calL, c \in
    \AA'_i\}$. Although these clients can no longer be sent to their
    assignments under $\pi_{\LS}$---those centers are closed---they
    can still be sent to their assignment under $\mu_{\AA}$. Indeed,
    if their $\mu_\AA$-center had closed, the definition of $\calA'$
    would put them in $\lsd$, so the first case defining $D_i^+$ would
    already apply.
  \item Finally, all remaining clients are sent to their $\pi_{\LS}$
    assignments.
  \end{enumerate}
  In summary:
  \begin{gather*}
    d(p, \LS^{\test}) \leq \begin{cases}
      d(p, \mu_{\BB}(p)) & \qquad \forall p \in D_i^+ \\
      d(p, \mu_{\AA}(p)) & \qquad \forall p \in D_i' \\
      d(p,\LS) &\qquad \text{otherwise}
    \end{cases}
  \end{gather*}
  This gives a bound on the cost of the post-swap solution:
  \begin{align}
    \cost(\LS^{\test}) &\leq \cost_{\LS}(D \setminus (D_i^+ \cup D_i')) + \cost_{\mu_{\BB}}(D_i^+) + \cost_{\mu_\AA}(D_i'). \label{eq:swap-cost}
  \end{align}
  Next, combining \Cref{eq:local-opt} and \Cref{eq:swap-cost} and rearranging:
  \begin{align}
    \cost_{\LS}(D_i^+)  \leq \cost_{\mu_{\BB}}(D_i^+) + \cost_{\mu_\AA}(D_i') - \cost_{\LS}(D_i') + \frac{\eps}{k}\cost(\LS). \label{eqtn:DiA-bd}
  \end{align}
  Before we sum \Cref{eqtn:DiA-bd} over all swaps in $\calS$, we
  observe that $\lsd \subseteq \cup_{(L_i, R_i) \in \calS}
  D_{i}^+$. Indeed, all leaves and all right vertices appear among
  some swap in $\calS$, by
  \Cref{defn:span-swaps}\ref{cond:alive-repped}. Now each client in
  $\lsdres$ appears in some dominated cluster among $\calA'$ and
  $\calB'$ (which correspond to leaves), so $p$ appears in some
  $D_i^+$. Additionally, each client $p \in \lsd \setminus \lsdres$
  appears in some $D_i^+$ by the additional round of client
  assignments, since $\calB$ is a partition of $D$ and all right
  vertices (i.e., centers of clusters in $\calB$) appear in some
  swap. In summary, every client in $\lsd$ appears in some $D_i^+$, so:
  \begin{align*}
    \cost_{\LS}(\lsd) &\leq \sum_{(L_i, R_i) \in \calS} \cost_{\LS}(D_i^+) 
                 \leq \sum_{(L_i, R_i) \in \calS}(\cost_{\mu_{\BB}}(D_i^+) + \cost_{\mu_\AA}(D_i') - \cost_{\LS}(D_i')) + 3 \eps \cost(\LS).
  \end{align*}
 Here we use that $|\calS| \leq 3k$. Indeed, discarding empty swaps, every swap in the construction from
\Cref{lem:graph-swaps} contains a right leaf appearing in no other
swap. Thus $|\calS|\leq |R|=|\AA'|+2\leq k+2\leq 3k$.
Taking expectations gives:
  \begin{align}
    \cost_\LS(\lsd) &\leq \mathop{\EE}_{\calS \sim \calD} \bigg[
                 \sum_{(L_i, R_i) \in \calS} \cost_{\mu_{\BB}}(D_i^+) \bigg] +
                 \mathop{\EE}_{\calS \sim \calD} \bigg[\sum_{(L_i, R_i) \in \calS}( \cost_{\mu_\AA}(D_i') - \cost_{\LS}(D_i')) \bigg] + 3\eps\cost(\LS). \label{eq:strongest-ls-bd}  
  \end{align}
  We bound the two expectation terms separately. For the first term,
  we start by observing that every client $p\in\lsd\setminus\lsdres$
  appears exactly once. Indeed, if $B_{\mu_\BB(p)}'$ is dominated, its
  right vertex is a leaf appearing in exactly one swap; otherwise, $p$
  appears only in its designated swap. Next, consider a client $p\in\lsdres$, and let $A_c\in\calA'$
and $B_{c^*}'\in\calB'$ be its clusters.
If $A_c$ is strictly dominating, then $c^*$ is a right leaf.
Whenever $c\in\AA_i'$, condition~\ref{cond:nbrs} gives
$c^*\in\BB_i$, so $p\in D_i$ exactly when $c^*\in\BB_i$;
this occurs in exactly one swap.
If the two clusters are paired, their centers are leaves and
appear together in exactly one swap, by
\Cref{defn:span-swaps}\ref{cond:nbrs}--\ref{cond:left-hubs}.
Otherwise, $B_{c^*}'$ is strictly dominating and $c$ is a left leaf.
Then $p\in D_i$ exactly when $c\in\AA_i'$, which occurs in at most
$(1+8\eps)$ swaps in expectation by \Cref{lem:graph-swaps}.
Thus every client in $\lsd$ belongs to at most $(1+8\eps)$ sets
$D_i^+$ in expectation, giving
  \begin{align}
    \mathop{\EE}_{\calS \sim \calD} \bigg[ \sum_{(L_i, R_i) \in
    \calS}\cost_{\mu_{\BB}}(D_i^+) \bigg]  &\leq (1 + 8\eps)
                                           \cost_{\mu_\BB}(\lsd). \label{eq:exp-bd} 
  \end{align} 
  This bounds the first expectation in (\ref{eq:strongest-ls-bd}). To
  bound the second expectation, a client $p \in D$ belongs to $D_i'$
  only if it belongs to a local-search cluster $C \in \calL$ whose
  center $c$ lies in $\AA'_i$. If $c$ is a hub of a left-star (center of
  a dominating cluster $A_c \in \calA'$), then by
  \Cref{defn:span-swaps}\ref{cond:left-hubs}, $p \in D_i'$ for at most
  one swap. If instead $c$ is a left-leaf (i.e., the center of a
  dominated cluster $A_c \in \calA'$ or an originally isolated node), then $p \in D_i'$ only if
  $c \in \AA'_i$, and we can again use \Cref{lem:graph-swaps}. Overall:
  \begin{align}
    \mathop{\EE}_{\calS \sim \calD} \bigg[ \sum_{(L_i, R_i) \in \calS}(\cost_{\mu_\AA}(D_i') - \cost_{\LS}(D_i')) \bigg]  &=  \mathop{\EE}_{\calS \sim \calD} \bigg[ \sum_{(L_i, R_i) \in \calS} \Big( \sum_{p \in D} (d(p, \mu_{\AA}(p)) - d(p, \LS)) \mathbf{1}\{p \in D_i'\} \Big) \bigg] \nonumber \\
                                                                                                                          &= \sum_{p \in D} (d(p, \mu_{\AA}(p)) - d(p, \LS)) \mathop{\EE}_{\calS \sim \calD} \Big[\sum_{(L_i, R_i) \in \calS} \mathbf{1}\{p \in D_i'\} \Big] \nonumber \\
                                                                                                                          &\stackrel{(\star)}{\leq} (1 + 8\eps) \sum_{p \in D} (d(p, \mu_{\AA}(p)) - d(p, \LS)) \\
                                                                                                                          &\leq (1 + 8\eps) (\cost(\mu_{\AA}) - \cost(\LS)). \label{eq:exp-bd2}
  \end{align}
  Observe that the inequality $(\star)$ relies on each of the terms
  $d(p, \mu_{\AA}(p)) - d(p, \LS)$ being non-negative; this holds since $d(p, \mu_{\AA}(p))$ is the distance to some center in $\LS$, whereas $d(p,\LS)$ is
  the distance to the closest such center. Combining
  \Cref{eq:strongest-ls-bd,eq:exp-bd,eq:exp-bd2} gives
  \begin{align*}
    \cost_{\LS}(\lsd) &\leq (1 + 8\eps) (\cost_{\mu_\BB}(\lsd) + \cost(\mu_{\AA}) - \cost(\LS)) + 3\eps \cost(\LS).
  \end{align*}
  Since $\cost(\LS)\leq\cost(\mu_\AA)$, the additional term is
$O(\eps)\cost(\mu_\AA)$. Rearranging gives the desired result.
\end{proof}

\subsection{Constructing the Distribution over Swap Collections}
\label{sec:constr-swaps}

We show how to construct the swap collections satisfying
\Cref{lem:graph-swaps}.  The intuition is quite natural: to ensure
conditions~\ref{cond:swap-small} and \ref{cond:nbrs} of
\Cref{defn:span-swaps}, we force hubs of large left-stars to not
appear in any swaps. We compensate for their lack by allowing leaves
on the left to appear more than once. Indeed,
conditions~\ref{cond:alive-repped} and \ref{cond:left-hubs} mean that
leaves on the left are the only left-vertices that may appear in multiple
swaps. Since at most $\eps$ fraction of the nodes on the left can be
left-hubs of large stars, it is intuitive that left leaves appear
$\approx 1+\eps$ times in expectation. However, implementing this idea
requires some care, as we now see.

Since we can handle single edges $\{u,v\}$ using swaps $(\{u\},\{v\})$
of unit size, let us focus on the case where each star has at least
$2$ leaves.  We additionally assume that $|L| = |R| > 12 /\eps$, else
the collection $\{(L, R)\}$ consisting of a single swap has all of the required
properties. Hence, we have at least $6/\eps$ leaves on both the left
and right, since each star has at least two leaves.  Since
$|L| = |R|$, we get
\begin{equation}
   |H_L| + |L \setminus H_L| = |H_R| + |R \setminus H_R| \quad
   \implies \quad \sum_{v \in H_L} (|S(v)| - 1) = \sum_{w \in H_R} (|S(w)| - 1). \label{eq:budget-balance}
\end{equation}

We begin by defining a partition $\Lambda_R$ of
right-leaves $R \setminus H_R$ into pieces of size between
$1/\eps$ and $3/\eps$. Then we iteratively match these pieces to
subsets of left leaves via a randomized matching process. %
\begin{enumerate}
\item For each hub $v$ of a big left-star, partition its right-leaves $S(v)$ arbitrarily into sets of size
  between $1/\eps$ and $2/\eps$. (This can be done by greedily forming
  sets of size $1/\eps$, adding the remainder of size less than
  $1/\eps$ to the final greedily computed set of size $1/\eps$.) Add
  each such set to $\Lambda_R$.
\item Next, for remaining hubs of small left-stars $v$ (with
  $1 < |S(v)| < 2/\eps$), greedily form sets of size between $1/\eps$
  and $3/\eps$ by taking unions of their right-leaves $S(v)$. Add each
  such set to $\Lambda_R$. (In the edge case that the total number of
  remaining leaves is less than $1/\eps$, merge into the last formed
  set---this is possibly a set containing the leaves of a big star. If the combined size is at least $3/\eps$, split it into two sets
  of size at least $1/\eps$ while keeping right-leaves belonging to the same star in the same set.)
\end{enumerate}
Combined, these two steps yield the partition $\Lambda_R$ of the right
leaves. (Note that if there are no big left-stars, our assumption that
$|L| = |R| > 12/\eps$ ensures at least $6/\eps$ leaves, so the second
step then still outputs a valid partition.) For each
$Z \in \Lambda_R$, define its hub set 
\begin{equation}
  H(Z) := \{ v \in H_L \mid S(v) \subseteq Z, |S(v)| < 2/\eps \}
\end{equation}
to be the set of hubs of small left-stars with all leaves in $Z$. The
swap $(L_i, R_i)$ involving $Z$ will include $Z \subseteq R_i$ and
$H(Z) \subseteq L_i$. Define the \emph{deficit} of $Z$ to be
\begin{equation}
    \deficit(Z):= |Z| - |H(Z)|.
\end{equation}
For a swap $(L_i, R_i)$ with $Z \subseteq R_i$ (and hence
$H(Z) \sse L_i$), this deficit $\deficit(Z)$ is the difference between
the number of vertices that $Z$ contributes to $R_i$ and $L_i$.

We will randomly match sets $Z \in \Lambda_R$ with sets of leaves on
the left. (Note that all small left hubs are already accounted for by
the inclusion of $H(Z)$.) For $Y \subseteq L \setminus H_L$, we define
a complementary notion of \emph{excess}. Define
\begin{equation}
 H(Y) := \{ w \in H_R \mid S(w) \cap Y \neq \emptyset \}
\end{equation}
to be the set of hubs of right-stars with any of their leaves in
$Y$. Any swap $(L_i, R_i)$ with $Y \subseteq L_i$ must include
$H(Y) \subseteq R_i$. Then, the \emph{excess} of $Y$ is the difference
between the number of vertices $Y$ contributes to $L_i$ and $R_i$,
i.e.,
\begin{equation}
    \excess(Y) := |Y| - |H(Y)|.
\end{equation}

All of our swaps will be of the form $(Y \cup H(Z), Z \cup H(Y))$,
where $Y \subseteq L \setminus H_L$ and $Z \in \Lambda_R$: ensuring
that $|Y \cup H(Z)| = |Z \cup H(Y)|$ is equivalent to
$\deficit(Z) = \excess(Y)$. Moreover, since \Cref{lem:graph-swaps} requires
minimizing occurrences of left leaves in swaps, we want to maximize the
excess contributions from $Y \subseteq L \setminus H_L$. As such, we
want to minimize the number of times hubs of right-stars must appear,
and hence to pack their leaves as efficiently as possible. Our
algorithm to match all the $Z \in \Lambda_R$ to sets
$Y \subseteq L \setminus H_L$ appears as \Cref{alg:make-swaps}. 

\bigskip
\begin{algorithm}[h]
  \caption{CreateSwaps}
  \label{alg:make-swaps}
  order the leaves within each right-star \;
  order right-stars and concatenate their left-leaf lists into a
  sequence $y_0,\ldots,y_{b-1}$, where $b=|L\setminus H_L|$. \;
  order the $Z$'s in $\Lambda_R$ arbitrarily. \;
  choose a uniform random starting offset $r_0\in\{0,\ldots,b-1\}$. \;
  initialize index $r \gets r_0$. \;
  \ForEach{$Z \in \Lambda_R$}
  { \label{line:find-Y}take shortest consecutive block $Y = \{y_r, y_{r+1}, \ldots\}$
    such that $\excess(Y)=\deficit(Z)$ \;
    set $r \gets r + |Y| \pmod b$ \;
    add swap $(Y \cup H(Z), Z \cup H(Y))$ to $\calS$.
  }
\end{algorithm}
\bigskip

\subsection{Proving the Swap Properties}
\label{sec:swap-props}

\begin{claim}
  \label{cla:Y-exists}
  For any index $r \in \{0,1, \ldots, b-1\}$ and any set $Z \in \Lambda_R$, there
  exists a consecutive block $Y$ beginning at $r$ with
  $\excess(Y) = \deficit(Z)$. Moreover,
  $|Y| \leq |Y \cup H(Z)| = |H(Y) \cup Z| < 7/\eps \leq 12/\eps$.
\end{claim}

\begin{proof}
  Let $w_0 \in H_R$ be the hub of the right-star corresponding to leaf
  $y_r \in S(w_0)$, and let $S'(w_0)$ be the subsequence of $S(w_0)$
  starting at index $r$. Then, $w_0$'s star can contribute any value
  between $0$ and $|S'(w_0)| - 1$ to the excess of $Y$ by taking a
  consecutive block of its leaves. If $\deficit(Z) \leq |S'(w_0)| - 1$,
  we are done. If the next right hubs are ordered $w_1, w_2, \ldots$,
  then each such star can contribute any value $v$ between $0$ and
  $|S(w_i)| - 1 \geq 1$ to the excess of $Y$, while contributing $v+1$
  leaves. In the worst case, $|S'(w_0)| = 1$, and each subsequent
  $|S(w_i)| = 2$, which would cause us to need $1 + 2\deficit(Z)$ new
  leaves before reaching $\excess(Y) = \deficit(Z)$. By construction, each
  $Z \in \Lambda_R$ has $|Z| < 3/\eps$, so
  $\deficit(Z) \leq |Z| < 3/\eps$, and hence
  $1 + 2\deficit(Z) < 6/\eps$. Since we have at least $6/\eps$ left
  leaves, this ensures such a $Y$ exists. To prove the second part of
  the claim, note that
  $|Y \cup H(Z)| \leq 1+ 2\deficit(Z) + |H(Z)| \leq 2|Z|+1 < 7/\eps \leq 12/\eps$.
\end{proof}

\begin{claim}
  \label{cla:loops}
  The index $r$ wraps back to $r_0$ over the course of the swap
  creation, and hence every leaf on the left belongs to a swap.
\end{claim}
\begin{proof}
  By definition of $H(Z)$, we know that the sum of the deficits in
  $\Lambda_R$ is 
  \begin{equation}
    \sum_{Z \in \Lambda_R} \deficit(Z) = \sum_{Z \in \Lambda_R} (|Z| - |H(Z)|) \stackrel{(\star)}{\geq} \sum_{v \in H_L} (|S(v)| - 1)
    =  \sum_{w \in H_R} (|S(w)| - 1). \label{eq:deficit}
  \end{equation}
  To see inequality $(\star)$, use that each left-star either is small
  and includes its hub in one $H(Z)$, or else it is big and its hub
  never appears in any $H(Z)$. On the other hand,
  $\sum_{w \in H_R} (|S(w)| - 1)$ is the maximum possible total excess
  of the $Y$ sets created if $r$ iterates all the way
  back to $r_0$, and covering a smaller region of the cycle would
  give a smaller total excess. (Indeed, this maximum contribution is
  also only achievable if each right-hub is opened exactly once; each
  right hub needs to be opened at least once to ensure
  \Cref{defn:span-swaps}\ref{cond:alive-repped}.) Since the sum of the
  excesses must equal the sum of the deficits, this means the counter
  must wrap around.
\end{proof}

We claim that the random set $\calS$ is a spanning collection of swaps
satisfying \Cref{lem:graph-swaps}. For any $(L_i, R_i) \in \calS$, let
$L_i = Y_i \cup H(Z_i)$, and $R_i = H(Y_i) \cup Z_i$. To verify the
conditions of \Cref{defn:span-swaps}:
\begin{enumerate}[nosep, label = (\alph*)]
\item $|L_i| - |R_i| = \excess(Y_i) - \deficit(Z_i) = 0$ by construction;
  \Cref{cla:Y-exists} shows that $|L_i| \leq 12/\eps$.
\item Our construction ensures that $N(H(Z_i)) \sse Z_i$, and
  $N(Y_i) = H(Y_i)$, which ensures that $N(L_i) \sse R_i$; this is
  exactly why we did not include hubs of large left-stars in $H(Z_i)$.
\item Each right leaf is contained in some $Z \in \Lambda_R$; the hub
  of a small left star appears in $H(Z)$ when $Z$ contains its
  leaves. \Cref{cla:loops} shows that each leaf on the left is matched
  to some $Z$ during the algorithm, and hence belongs to some $Y$;
  this means its hub on the right appears in $H(Y)$.
\item Each right leaf appears in a unique $Z \in \Lambda_R$;
  if it belongs to a small left-star, then the left-hub belongs to
  only that $H(Z)$. Large left-hubs do not appear in any swaps. 
\end{enumerate}

Finally, we prove that each leaf in $L$ belongs to at most
$(1 + 8\eps)$ swaps in expectation. The idea is natural: using a
complementary argument to that in \Cref{cla:loops}, we show that the
counter $r$ does not go much further after completing a full cycle.
Let $\Delta := |R \setminus H_R| - |H_L| = |L \setminus H_L| - |H_R|$.
Since for all hubs $v \in H_L$ of a small left-star,
$v \in H(Z)$ for some $Z \in \Lambda_R$ by construction, we can
bound the total deficit $D_{\mathrm{total}} := \sum_{Z \in \Lambda_R} \deficit(Z)$ \emph{from above} as
follows:
\begin{align}
  \sum_{Z \in \Lambda_R} \deficit(Z)
  &= \sum_{v \in H_L} (|S(v)| - 1)  + |H_{L, \mathrm{big}}|
  \leq \Delta + (\eps/2)|R \setminus H_R| 
  \leq (1 + \eps)\Delta. \label{eq:deficit2}      
\end{align}
The first inequality uses that each big left-star has at least $2/\eps$ leaves, so there are at most $(\eps/2)|R \setminus H_R|$ big left-hubs forming the set $H_{L, \mathrm{big}}$; the second inequality uses that $|R \setminus H_R| \geq 2|H_L|$, which implies $|R \setminus H_R| \leq 2\Delta$.

We now give a \emph{lower bound} on the total excess---i.e., the sum
of $\excess(Y_i)$---over the blocks $Y_1, \ldots, Y_m$ formed during
the first cycle of the index $r$ around the $b = |L \setminus H_L|$
left-leaves. For this calculation, if the last block crosses the end
of the first cycle, truncate $Y_m$ at that point.
By construction, each $Z \in \Lambda_R$ has $|Z| \geq 1/\eps$;
since each small left-star has at least two leaves,
$|H(Z)| \leq |Z|/2$, and hence every deficit satisfies
$\deficit(Z) \geq 1/(2\eps)$.
Since the sum of all deficits is
$D_{\mathrm{total}} \leq (1+\eps)\Delta$, the total number of sets
in $\Lambda_R$ (and in particular the number of blocks $m$ in the
first cycle) is at most
\begin{equation}
  m \leq |\Lambda_R| \leq \frac{D_{\mathrm{total}}}{1/(2\eps)}
  \leq 2\eps(1+\eps)\Delta.
\end{equation}
As $r$ completes a single cycle, the blocks $Y_1, \ldots, Y_m$
cover all $|L \setminus H_L|$ left-leaves exactly once. Moreover,
since each right-star $w \in H_R$ appears as a contiguous interval
of leaves in the cyclic ordering, a right-hub $w$ appears in multiple
$H(Y_i)$ sets only when a boundary between consecutive blocks, or
the starting offset, falls strictly inside $S(w)$. There are at most
$m$ such boundaries in total, and each boundary splits at most one
right-star, contributing at most $1$ extra hub appearance across all
$m$ blocks. Therefore, $\sum_{i=1}^m |H(Y_i)| \leq |H_R| + m$,
and the total excess generated in the first cycle is at least
\begin{align}
  \sum_{i=1}^m \excess(Y_i)
  = \sum_{i=1}^m |Y_i| - \sum_{i=1}^m |H(Y_i)|
  &\geq |L \setminus H_L| - (|H_R| + m) \nonumber \\
  &= \Delta - m
  \geq (1 - 2\eps - 2\eps^2)\Delta. \label{eq:excess-bd}
\end{align}
Let $\Lambda_{\mathrm{rem}} \subseteq \Lambda_R$ denote the remaining
unprocessed sets after the first cycle of $r$, excluding the set
corresponding to the possibly unfinished block. Since
$\excess(Y_i) = \deficit(Z_i)$ for each completed block, the remaining
deficit to be covered, including the unfinished part of the last
block, is at most
\begin{equation}
  D_{\mathrm{total}}- \sum_{i=1}^m \excess(Y_i)
  \leq (1+\eps)\Delta - (1 - 2\eps - 2\eps^2)\Delta 
  = (3\eps + 2\eps^2)\Delta.
\end{equation}
By \Cref{cla:Y-exists}, each remaining block $Y_Z$ requires at most
$|Y_Z| \leq 1 + 2\deficit(Z)$ leaves. The same bound holds for the
unfinished block, with its deficit replaced by its remaining deficit:
each newly encountered right-star except possibly the last contributes
at least two leaves, and a hub already encountered in its prefix
contributes no new hub appearance within that block.
Since $\deficit(Z) \geq 1/(2\eps)$ for each
$Z \in \Lambda_{\mathrm{rem}}$, we have
\[
  |\Lambda_{\mathrm{rem}}|
  \leq 2\eps\left(D_{\mathrm{total}}
    - \sum_{i=1}^m \excess(Y_i)\right).
\]
Summing over the remaining blocks and including the possibly
unfinished block, the total number of left-leaf visits after the
first cycle is at most
\begin{align}
  |\Lambda_{\mathrm{rem}}|
  + 2\left(D_{\mathrm{total}} - \sum_{i=1}^m \excess(Y_i)\right) + 1
  &\leq (2 + 2\eps)\left(D_{\mathrm{total}}
    - \sum_{i=1}^m \excess(Y_i)\right) + 1 \nonumber \\
  &\leq 2(1+\eps)(3\eps + 2\eps^2)\Delta + 1 \\
  &\leq 8\eps|L \setminus H_L|,
\end{align}
where the last inequality uses $\eps \leq 1/200$,
$\Delta = |L \setminus H_L| - |H_R|$, and
$|L \setminus H_L| \geq 6/\eps$.
In conclusion, all visits after the first cycle lie among the first
$\lfloor 8\eps b\rfloor$ leaves starting at $r_0$. This is of course at most $b$, so no leaves appear more than twice. Moreover, the probability a leaf is covered twice is at most the probability it appears in the first $8\eps$ fraction of leaves in the ordering. This probability is in turn at most $8\eps$ since $r_0$ is a uniform random starting offset in the cyclic ordering of left-leaves. This implies that each left leaf belongs to at most $(1 + 8\eps)$ swaps in expectation
for $(L_i, R_i) \in \calS$, completing the proof of~\Cref{lem:graph-swaps}.


%% file: app-local-search.tex
\section{Proofs for the Local-Search Guarantee}
\label{sec:app-proofs-local-search}

We now prove \Cref{lem:new-maps-cost} showing the existence of a nice
cross-free pair:

\NicePair*

\subsection{Defining Maps $\mu_\AA, \mu_\BB$}

We start by recalling from
\Cref{thm:structural}\ref{cond:struct-subs}, parts
\ref{cond:struct-substitution-weak} and \ref{item:unique-preimage},
which ensure that any $c^* \in \OPT_R$ satisfies
$\sigma(c^*) \in \Gamma(c^*)$---i.e., its substitute belongs among its
satellites---and also the preimage of $\sigma(c^*)$ under $\sigma$ is exactly $c^*$. 
Using these properties, we define a ``reverse'' map $h: \Gamma(\OPT_R) \to \OPT_R$:
\begin{itemize}[nosep]
\item If $c = \sigma(c^*)$ for some $c^* \in \OPT_R$, then $h(c) = c^*$.
\item Else, let $h(c)$ be any $c^*$ such that $c \in \Gamma(c^*)$.
\end{itemize}
Since a local search center can be a satellite for several centers in
$\OPT_R$, the map $h$ allows us to associate each satellite to an
optimal center in some canonical way. Our assignments will also carefully treat clients in $\lsd$ belonging to some satellite cluster; define \margindef{$D_{\Gamma}$}
\[
D_\Gamma:=\{p \in \lsd:\pi_\LS(p)\in\Gamma(\OPT_R)\} \subseteq \lsd.
\]
We can now define assignment maps $\mu_{\BB}$ and $\mu_{\AA}$. 


\textbf{Map $\mu_\BB$:} We describe $\mu_{\BB}$ in two parts: assignments of clients in $\lsd$ and assignments of clients outside of $\lsd$.
\begin{enumerate}[label=($\BB$\arabic*)]
\item Reassign
the clients $p \in D_\Gamma$ as follows: For each center $c \in \Gamma(\OPT_R)$, and for any client $p \in \lsd$
  with $\pi_\LS(p) = c$, define $\mu_\BB(p) := h(c) \in \OPT_R$.
\item For clients $p \in \lsd \setminus D_{\Gamma}$, maintain
  $\mu_\BB(p) = \pi_\BB(p)$, the nearest-center mapping for $\BB$. (Notice that, for these clients, $\mu_{\BB}(p) = \pi_{\OPT}(p)$.)
\item For clients $p \in D \setminus \lsd$, assign the client to $\BB$ arbitrarily. The assignment will not be relevant.
\end{enumerate}
In words, the centers in $\OPT_R$ ``absorb'' clients in $D_\Gamma$ in their
satellites' clusters; if a local search center is a satellite for
multiple optimal centers, it chooses one according to the map $h$.
Let $\calB$ be the clustering associated with $(\BB, \mu_\BB)$.

\textbf{Map $\mu_\AA$:} The approach to define $\mu_\AA$ is similar to
that of the standard local search argument from \Cref{cor:3-approx},
but with some additional care. Since each center $c^* \in \OPT_R$
satisfies $h(\sigma(c^*)) = c^*$, its cluster in the clustering
$\calB$ has absorbed the $D_0$-clients from cluster $C_{\sigma(c^*)}$---and
perhaps others---into its cluster in $\calO$. It may also have lost
clients due to other satellites intersecting $C^*$ but being absorbed
elsewhere. The clusters centered at $\OPT_{\Lambda}^{\mathrm{early}}$
only lose clients in $\lsd$.

We now define $\mu_\AA$ as follows:
\begin{enumerate}[label=($\AA$\arabic*)]

\item For a center $c^* \in \OPT_{\Lambda}^{\mathrm{early}}$ with corresponding cluster $C^* \in \calO$, let center $c$ be the ``non-satellite''
  local search center closest to $c^*$ in
  $\LS \setminus \Gamma(\OPT_R)$. For all points $p$ in 
$C^* \setminus D_{\Gamma}$, define $\mu_\AA(p) = c$.

\item For center $c^* \in \OPT_R$ with clusters $B_{c^*} \in \calB$ and $C^* \in \calO$, do
  the following: For each client $p$ in $B_{c^*} \cap D_\Gamma$---i.e., a
  client in $\lsd$ from a satellite cluster absorbed by $B_{c^*}$---maintain
  $\mu_\AA(p) = \pi_\LS(p)$. For any other client $p$ in
  $C^* \setminus D_{\Gamma}$, define $\mu_\AA(p) = \sigma(c^*)$.
\item For $c^* \in \OPT_S$ with $C^* \in \calO$, for all $p \in C^*$ define $\mu_{\AA}(p) = \sigma(c^*)$.
\item For the remaining centers $c^* \in \OPT \setminus (\OPT_{\Lambda}^{\mathrm{early}} \cup \OPT_{R} \cup \OPT_S) = \OPT_P \cup \OPT_{\Lambda}^{\mathrm{late}}$ with cluster $C^* \in \calO$, let $c$ be the 
  closest local search center to $c^*$. For all $p \in C^*$, define $\mu_{\AA}(p) = c$.
\end{enumerate}
This defines the clustering $\calA$; the cluster around center
$c \in \LS$ is denoted $A_c$.

\begin{figure}[t]
  \centering
  \begin{subfigure}[t]{.45\textwidth} 
    \centering \includegraphics[width=3in,page=1]{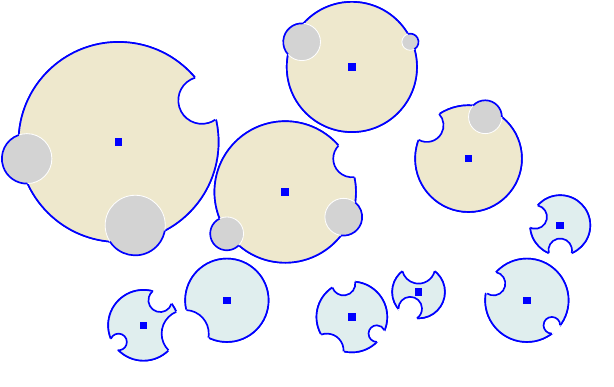}
  \end{subfigure}
  \qquad
  \begin{subfigure}[t]{.45\textwidth} 
    \centering \includegraphics[width=3in,page=2]{images/mua-mub}
  \end{subfigure}
  \caption{Left: the map $\mu_\BB$ on clients in $D_0$. The top four clusters correspond
    to $\OPT_R$ and absorb $D_0$-clients from the satellites (with the
    absorbed local search clusters shown in grey), whereas the rest of
    the clusters only lose clients. Right: the map $\mu_\AA$ on clients in $D_0$. Each distinct color is a cluster of $\calA$; the clusters of
    $\calB$ corresponding to $\OPT_R$ dominate the clusters in
    $\calA$, whereas clusters of $\calA$ are dominating for the remainder centered in $\OPT_{\Lambda}^{\mathrm{early}}$.}
  \label{fig:mua-mub}
\end{figure}

\begin{claim}
  $(\calA, \calB)$ is cross-free with respect to $\lsd$ and $|\BB| + |\mu_{\AA}(D \setminus \lsd)| \leq k$.
\end{claim}
\begin{proof}
First we show that $|\BB| + |\mu_{\AA}(D \setminus \lsd)| \leq k$. For this, note that 
\[
\mu_{\AA}(D \setminus \lsd) = \pi_{\LS}(\OPT_P \cup \OPT_{\Lambda}^{\mathrm{late}}) \cup \sigma(\OPT_S).
\]
So, $|\BB| + |\mu_{\AA}(D \setminus \lsd)| \leq |\OPT_R| + |\OPT_{\Lambda}^{\mathrm{early}}| + |\OPT_P| + |\OPT_{\Lambda}^{\mathrm{late}}| + |\OPT_S| = |\OPT| \leq k$, as needed.

Now, let $\calA' := \{A \in \calA \mid A \subseteq \lsd\}$, $\widehat{D}_0$ the set of points in clusters in $\calA'$, and $\calB' := \{B \cap \widehat{D}_0 \mid B \in \calB\}$, as in \Cref{defn:gen-cross-free}. We need to show that $\calA'$ and $\calB'$ are cross-free. Let $A_c \in \calA'$ and $B_{c^*}' \in \calB'$ be non-empty, with $C$ and $C^*$ the corresponding clusters centered at $c$ and $c^*$ in $\calL$ and $\calO$. 

First, suppose $c^* \in \OPT_{\Lambda}^{\mathrm{early}}$. The cluster $B_{c^*}'$ is a subset of $C^* \setminus D_{\Gamma}$ by construction of $\mu_{\BB}$. ($C^* \setminus D_{\Gamma}$ is the restriction of $B_{c^*} \in \calB$ onto $\lsd$.) The clients in $C^* \setminus D_{\Gamma}$ are all sent to the same center $c' = \pi_{\LS \setminus \Gamma(\OPT_R)}(c^*)$ under $\mu_{\AA}$. Then if $c = c'$, $B_{c^*}' \subseteq A_c$ and otherwise $B_{c^*}' \cap A_c = \emptyset$. 

Otherwise, $c^* \in \OPT_{R}$. In this case, the cluster $B_{c^*}'$ is a subset of $C^* \cup D_{\Gamma}$. The clients in $B_{c^*}'$ are assigned to centers $c \in \Gamma(c^*)$ with $h(c) = c^*$ under $\mu_{\AA}$, so $B_{c^*}'$ and $A_c$ are trivially cross-free unless $c \in \Gamma(c^*)$ with $h(c) = c^*$. But, if $c \in \Gamma(c^*)$ with $h(c) = c^*$, then 
\[
A_c \subseteq
\bigl((C^* \setminus D_{\Gamma}) \cup (C \cap \lsd)\bigr)
\cap \widehat{D}_0
\subseteq B_{c^*}'
\]
by construction of $\mu_{\AA}$. Here we are relying on the fact that
$A_c \subseteq \widehat{D}_0$, since $A_c \in \calA'$ is non-empty. This ensures cross-freeness, completing the proof.
\end{proof}

\subsection{Bounding Costs}

\begin{lemma}[Cost for Clustering $\calB$] \label{lem:muB bd}
    $\cost_{\mu_{\BB}}(\lsd) \leq \cost_{\OPT}(\lsd) + O(\eps) \sopt$.
\end{lemma}

\begin{proof}
  Since on $\lsd$, assignments $\pi_\OPT$ and $\mu_\BB$ differ only on the clients in
  $D_\Gamma$, consider cluster $C \in \calL$ with center
  $c \in \Gamma(\OPT_R)$; let $c^* = h(c)$, and let $C^*$ be its
  cluster in $\calO$. We now bound $d(p, c^{*}) - d(p,\OPT)$ for
  clients $p \in C \cap \lsd$.

\begin{figure}[t]
  \begin{center}
    \includegraphics[width=4in]{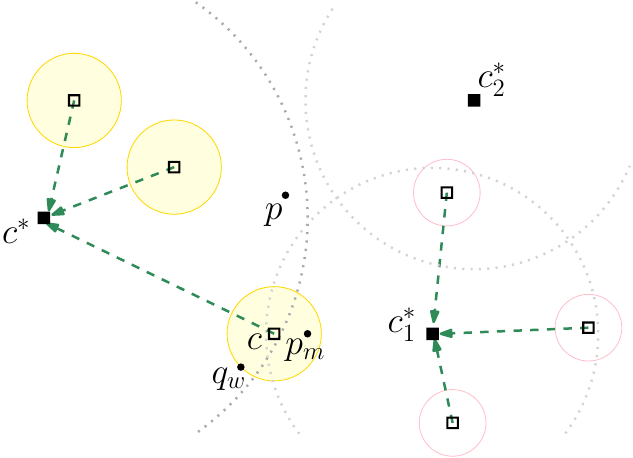}
  \end{center}
  \caption{Bounding the cost for $\mu_\BB$. The dotted balls around each
    center have radius $(1+\newkappa+2\eps) \per(C^*)$, and the green
    arrows indicate the mapping $h: \Gamma(\OPT_R) \to \OPT_R$. For
    the center $c^*$, the yellow balls around centers in $\Gamma(c^*)$
    have radius $\eps \per(C^*)$; the ball around center $c$ has a
    witness $q_w \in C^*$. }
  \label{fig:sats}
\end{figure}

  First, consider client $p_{m} \in D$ with minimum distance to center
  $c$. 
  Since $c$ is a satellite for $c^*$, there exists a witness client
  $q_w \in C^*$ such that $d(q_w, c) \leq \eps \per(C^*)$.  Now, using
  the fact that $p_m$ is the closest client to $c$, and hence
  $d(p_m,c) \leq d(q_w,c) \leq \eps \per(C^*)$, along with the
  triangle inequality,
  \begin{align}
    d(p_m, c^*) &\leq d(p_m,c)+d(c,q_w)+d(q_w,c^*) \implies d(p_m, c^*)
                  \leq d(q_w,c^*) +
    2\eps \per(C^*) \label{eq:5} \\
    d(q_w,c^*) &\leq d(q_w,c)+d(c,p_m)+d(p_m, \OPT) \implies d(p_m, \OPT)
               \geq d(q_w,c^*) - 2\eps \per(C^*). \label{eq:6}
  \end{align}

  Moreover, parts \ref{cond:struct2-not-close-new} and \ref{cond:struct2-props-remaining}\ref{cond:struct2-per-approx-avg}
   of \Cref{thm:struct2} imply that 
  $d(c^*, c) \geq \eta \avg_\OPT(C^*) \geq \frac{\eta}{1 + 2\eps}
  \per(C^*)$; since $d(q_w, c) \leq \eps \per(C^*)$, the triangle
  inequality implies that
  \[ d(q_w, c^*) \geq \Big(\frac{\eta}{1 + 2\eps} - \eps\Big)
    \per(C^*). \]
  Now using~(\ref{eq:5}) and (\ref{eq:6}), 
  \begin{equation}
    \frac{d(p_{m}, c^{*})}{d(p_{m}, \OPT)} \leq
    \frac{d(q_w,c^*) +
      2\eps \per(C^*)}{d(q_w,c^*) -
      2\eps \per(C^*)} \stackrel{(\star)}{\leq} 
    \frac{\eta +
      \eps(1+2\eps)}{\eta - 3\eps(1+2\eps)} \leq 1 + 5\eps, \label{eqtn:close-ratio} 
  \end{equation}
  where $(\star)$ follows by substituting the minimum possible value
  for $d(q_w,c^*)$ in the preceding expression; the last inequality
  uses that $\eps \leq \nf1{200}$.

  Now, to bound $d(p, c^{*}) - d(p, \OPT)$ for an arbitrary client
  $p \in C \cap \lsd$:
  \begin{align*}
    d(p, c^{*}) - d(p, \OPT)
    &\leq \big( d(p,c) + d(c,p_{m}) +
      d(p_{m}, c^{*}) \big) - \big( d(p_{m}, \OPT) -
      d(p_{m},c) - d(c,p)\big) \\
    &\leq 4\,d(p,c) + \big(d(p_{m}, c^{*}) - d(p_m, \OPT)\big) \leq
      4\,d(p,c) + 5\eps \cdot d(p_m, \OPT) \\
    &\leq (4+10\eps)\,d(p,c) + 5\eps \cdot d(p, \OPT),
  \end{align*}
  where we used the triangle inequality and $d(p_m,c) \leq d(p,c)$ for
  the first two inequalities, \eqref{eqtn:close-ratio} for the next
  inequality, and the triangle inequality to infer
  $d(p_m, \OPT) \leq d(p_m, c) + d(c,p) + d(p, \OPT) \leq 2d(p,c) +
  d(p, \OPT)$ for the last inequality. Summing this over all points in
  $D_\Gamma$,  and using $\eps \leq \nf1{200}$, 
  \begin{align}
    \sum_{p \in \lsd} \big( d(p, \mu_\BB(p)) - d(p, \pi_\OPT(p)) \big)
    &\leq \sum_{c \in \Gamma(\OPT_R)} \sum_{p
      \in C \cap \lsd} \big( 5\,d(p,\LS) +5\eps\, d(p, \OPT) \big).
  \end{align}
  Finally, by \Cref{thm:struct2}\ref{cond:struct2-props-remaining}\ref{cond:struct2-hidden-satellites}
the total cost in $\LS$ of these satellite clusters is small, bounding the first term and  
  completing the proof. 
\end{proof}

\begin{lemma}[Cost for Clustering $\calA$] \label{lem:muL cost bd}
We have
\begin{align*}
\cost(\mu_\AA) \leq{}& \cost_\LS(\lsd)
 +(1+\eta)\cost_\OPT(D_S)
 +2\cost_\OPT(D_\Lambda^{\early}) \\
&+\sum_{c^*\in\olamlate\cup\OPT_P}
 \bigl(\cost_\OPT(C^*)+|C^*|\per(C^*)\bigr) \\
&+\sum_{c^*\in\OPT_R}
 \bigl(\cost_\OPT(C^*\setminus D_\kappa)
       +|C^*\setminus D_\kappa|\per(C^*)\bigr)
 +O(\eps)\sopt.
\end{align*}
\end{lemma}

\begin{proof}
The clients in $D_S$ are assigned according to $\sigma$, so their
total cost is at most $(1+\eta)\cost_\OPT(D_S)+O(\eps)\sopt$ by
\Cref{thm:structural}\ref{cond:struct-subs}%
\ref{cond:struct-substitution-strong}.

Next, consider $c^*\in\olamlate\cup\OPT_P$. Its clients are assigned
to $\pi_\LS(c^*)$. For every $q\in\core_0(C^*)$, the triangle
inequality gives
\[
d(c^*,\LS)\leq d(c^*,q)+d(q,\LS)
\leq\per(C^*)+d(q,\LS).
\]
Averaging over $\core_0(C^*)$, which contains at least
$\eps|C^*|$ clients, yields
\[
d(c^*,\LS)\leq\per(C^*)+
\frac{\cost_\LS(\core_{6\eps}(C^*))}{\eps|C^*|}.
\]
Consequently,
\[
\cost_{\mu_\AA}(C^*)\leq
\cost_\OPT(C^*)+|C^*|\per(C^*)
+\eps^{-1}\cost_\LS(\core_{6\eps}(C^*)).
\]
Summing over $\olamlate\cup\OPT_P$, the last terms total at most
$\eps^4\sopt$ by
\Cref{thm:struct2}\ref{cond:struct2-hidden-core-gen}.

For $c^*\in\olamearly$ and $p\in C^*\setminus D_\Gamma$, the center
$\pi_\LS(p)$ belongs to $\LS\setminus\Gamma(\OPT_R)$. Since
$\mu_\AA(p)$ is the center in this set closest to $c^*$,
\[
d(p,\mu_\AA(p))\leq d(p,c^*)+d(c^*,\pi_\LS(p))
\leq 2d(p,c^*)+d(p,\LS).
\]
This totals to at most $2\cost_\OPT(D_\Lambda^{\early}) + \cost_{\LS}(D_0 \setminus D_\Gamma)$. The clients in $D_{\Gamma}$ retain their local-search
assignments. Their cost is then $\cost_\LS(D_{\Gamma})$.

Finally, for $c^*\in\OPT_R$ and $p\in C^*\setminus D_\Gamma$,
\Cref{thm:structural}\ref{cond:struct-subs}%
\ref{cond:struct-substitution-weak} gives
\[
d(p,\mu_\AA(p))=d(p,\sigma(c^*))
\leq d(p,c^*)+(1+\eps)\per(C^*).
\]
This gives a total contribution of 
\[
\sum_{c^*\in\OPT_R}
 \bigl(\cost_\OPT(C^*\setminus D_\Gamma)
       +|C^*\setminus D_\Gamma|\per(C^*)\bigr) + O(\eps) \sopt.
\]
We can replace $D_{\Gamma}$ here by $D_{\kappa}$, as needed for our result, only incurring an extra $O(\eps)\sopt$. To see this, consider
$p\in\core_\kappa(C^*)\setminus D_\Gamma$. We have
$d(p,\LS)>\eps\per(C^*)$, since otherwise $\pi_\LS(p)$ would be a
satellite of $c^*$ and hence $p\in D_\Gamma$. Also,
$d(p,c^*)+\per(C^*)\leq3\per(C^*)$. Thus the total error from
these missing core clients is bounded by
\begin{align*}
\sum_{c^*\in\OPT_R}
\sum_{p\in\core_\kappa(C^*)\setminus D_\Gamma}
\bigl(d(p,c^*)+\per(C^*)\bigr)
&\leq\frac3\eps\sum_{c^*\in\OPT_R}
\cost_\LS(\core_\kappa(C^*))\\
&\leq3\eps^4\sopt,
\end{align*}
by \Cref{thm:struct2}\ref{cond:struct2-props-remaining}%
\ref{cond:struct2-hidden-core-R}. Summing these bounds proves the lemma.
\end{proof}

Combining the bounds from \Cref{lem:muB bd,lem:muL cost bd} %
yields \Cref{lem:new-maps-cost}.


%% file: app-submod.tex
\section{Supermodularity of Distance Functions}
\label{sec:supermod}

In this section, we show some functions involving distances from a
point to a set of points are super/submodular. We start with a helper
claim about supermodularity of minima.

\begin{claim}[Supermodularity of Minimums]
  \label{sec:simple-fact}
  Suppose  $a_i \in \mathbb{R}_{\geq 0}$ for all
  $i \in U \cup \Gamma$. For any $X \sse U$, define
  $h(X) := \min_{i \in X \cup \Gamma} \{ a_i \}$. Then
  $h$ is monotone non-increasing and supermodular. 
\end{claim}

\begin{proof}
  The function $h$ is monotone non-increasing since we take minima
  over larger sets. To prove supermodularity, consider $X \sse Y$ and
  $e \not \in Y$. Let $m_X$ and $m_Y$ be the minima over
  $X \cup \Gamma$ and $Y \cup \Gamma$ respectively; note that
  $m_X \geq m_Y$ by using monotonicity. Suppose $a_e \geq m_Y$, then
  we have
  \[ h(X) - h(X+e) \geq 0 = h(Y) - h(Y+e),\]
  with the inequality using monotonicity of $h$. Else we have $a_e < m_Y
  \leq m_X$, and hence
  \[ h(X) - h(X+e) = m_X - a_e \geq m_Y - a_e = h(Y) - h(Y+e), \]
  completing the proof of supermodularity.
\end{proof}

\begin{corollary}[Distance Supermodularity] \label{cor:dist-supmod}
  For any point $p$ and set of points $\Gamma$ (which may depend on
  $p$), the function $h(X) := d(p, X \cup \Gamma)$ is monotone
  non-increasing and supermodular.
\end{corollary}

\begin{proof}
  Note that $d(p,X \cup \Gamma) = \min_{i \in X \cup \Gamma} d(p,i)$
  and use \Cref{sec:simple-fact}.
\end{proof}

\begin{claim}[Submodularity of Threshold Function] \label{cla:submod-thresh}
Suppose we have $a_{ij} \in \mathbb{R}_{\geq 0}$ for all $i \in U \cup \Gamma$ and
  $j \in L$, and define
  $h_j(X) = \min_{i \in X \cup \Gamma} \{a_{ij}\}$. Then, for any $t \geq 0$, the function
  \begin{equation}
      N_X(t) := |\{j \in L \mid h_j(X) \leq t\}|
  \end{equation}
  is monotone non-decreasing and submodular. 
\end{claim}
\begin{proof}
To prove submodularity, first note that, for any $X$,
\begin{equation}
    N_X(t) = \sum_{j \in L} \mathbf{1}(h_j(X) \leq t). 
\end{equation}
Hence, it suffices to show that $\mathbf{1}(h_j(X) \leq t)$ is submodular and monotone non-decreasing for each $j$. For simplicity write $h$ to denote $h_j$  and $\mathbf{1}(X)$ to denote $\mathbf{1}(h(X) \leq t)$. Since $h$ is monotone non-increasing, $\mathbf{1}(X)$ is monotone non-decreasing. Now fix $X \subseteq Y$ and $e \not\in Y$. We want to show that
\begin{equation*}
\mathbf{1}(X + e) - \mathbf{1}(X) \geq \mathbf{1}(Y + e) - \mathbf{1}(Y).
\end{equation*}
The only way for this inequality to be violated is when $\mathbf{1}(Y) = 0$ and $\mathbf{1}(Y + e) = 1$.  The former condition implies $\mathbf{1}(X) = 0$ since $\mathbf{1}(\cdot)$ is monotone non-decreasing. Moreover, opening up the definition of $h$, $\mathbf{1}(Y) = 0$ and $\mathbf{1}(Y + e) = 1$ combined imply that $a_{ej} \leq t$. Hence, $\mathbf{1}(X + e) = 1$, completing the proof.
\end{proof}

\begin{claim}[Minimum of a Submodular Function and a Constant] \label{cla:const-min-submod}
Let $g$ be a monotone non-decreasing and submodular function, and let $r \in \RR$. Then, $f = \min\{g, r\}$ is also monotone non-decreasing and submodular.
\end{claim}
\begin{proof}
Since $g$ is monotone non-decreasing, evidently $f$ is as well. Now, for submodularity, consider $X \sse Y$ and
  $e \not \in Y$. Suppose first that $g(Y) \geq r$. Then,  since $f$ is monotone non-decreasing,
  \[
  f(Y + e) - f(Y) = 0 \leq f(X + e) - f(X).
  \]
  Otherwise, suppose that $g(Y) < r$. If $g(X + e) \geq r$, then, again since $g$ is monotone non-decreasing
  \[
  f(Y + e) - f(Y) = r - f(Y) \leq r - f(X) = f(X + e) - f(X).
  \]
  Finally, suppose that $g(Y) < r$ and $g(X + e) < r$. Then, by submodularity 
\[
  f(Y + e) - f(Y) \leq g(Y + e) - g(Y) \leq  g(X + e) - g(X) = f(X +
  e) - f(X). \qedhere
  \]
\end{proof}

\begin{claim}[Minimums of Minimums] \label{cla:minofmin}
  Suppose we have $a_{ij} \in \mathbb{R}_{\geq 0}$ for all $i \in U \cup \Gamma$ and
  $j \in L$, and
  $h_j(X) = \min_{i \in X \cup \Gamma} \{a_{ij}\}$. Let $r \in \mathbb{Z}^+$ with $r \leq |L|$. Define
  \[ g(X) = \min_{R \sse L : |R| = r} \Big\{ \sum_{j \in R} h_j(X)
    \Big\}. \]
  Then $g$ is monotone non-increasing, and supermodular. 
\end{claim}

\begin{proof}
Define helper function $N_X(t) := |\{j \in L \mid h_j(X) \leq t\}|$. Then, note that we have 
\begin{equation} \label{eq:integral}
    g(X) = \int_{0}^{\infty} (r - \min\{r, N_X(t)\}) dt.  
\end{equation}
Note that $N_X(t)$ is monotone non-decreasing and submodular for any fixed $t$, by~\Cref{cla:submod-thresh}. Hence, $r - \min\{r, N_X(t)\}$ is monotone non-increasing. The integral of a monotone non-increasing function is monotone non-increasing, so $g$ is monotone non-increasing. 

Since $N_X(t)$ is monotone submodular, \Cref{cla:const-min-submod} implies that $\min\{r,N_X(t)\}$ is monotone submodular. Hence, $r - \min\{r, N_X(t)\}$ is supermodular for each $t$, and integrating preserves supermodularity.
\end{proof}

\begin{claim} \label{cla:lifting-sub}
Let $U$ and $\overline{U}$ be sets and $\proj: \overline{U} \to U$. For $f : 2^U \to \mathbb{R}$, let $f_\uparrow : 2^{\overline{U}} \to \mathbb{R}$ such that, for $\overline{X} \subseteq \overline{U}$, $f_\uparrow(\overline{X}) = f(\proj(\overline{X}))$. Then, if $f$ is monotone non-decreasing and submodular, so is $f_\uparrow$. 
\end{claim}
\begin{proof}
Fix $\overline{X} \subseteq \overline{Y} \subseteq \overline{U}$. Note that this implies that $\proj(\overline{X}) \subseteq \proj(\overline{Y}) \subseteq U$. Hence, $f_\uparrow(\overline{X}) = f(\proj(\overline{X})) \leq f(\proj(\overline{Y})) = f_\uparrow(\overline{Y})$, using that $f$ is monotone non-decreasing.

For submodularity, let $e \in \overline{U}$. Using submodularity of $f$, we have 
\begin{align*}
    f_\uparrow(\overline{Y} + e) - f_\uparrow(\overline{Y}) &= f(\proj(\overline{Y} + e)) - f(\proj(\overline{Y})) \\
    &= f(\proj(\overline{Y})  + \proj(e)) - f(\proj(\overline{Y})) \\
    &\leq f(\proj(\overline{X}) + \proj(e)) - f(\proj(\overline{X})) \\
    &= f(\proj(\overline{X} +e)) - f(\proj(\overline{X})) \\
    &= f_\uparrow(\overline{X} + e) - f_\uparrow(\overline{X}). \qedhere
\end{align*}
\end{proof}


%% file: cost-bd-lemma.tex
\section[Proof of the win-win lemma]{Proof of \Cref{lem:winwin-numerical}}
\label{sec:winwin-technical-proof}
In this section we prove \Cref{lem:winwin-numerical}. We recall the statement of the lemma.
\winwin*
\begin{proof}[Proof of \Cref{lem:winwin-numerical}]
Fix an optimal cluster $C^*$, and abbreviate $x=x(C^*)$,
$y_\kappa =y_\kappa(C^*)$, and $\kappa=\newkappa$; recall that
\[
x :=\frac{\per(C^*)}{\avg_\OPT(C^*)},
\qquad
y_\kappa :=
\frac{\cost_\OPT(\core_\kappa(C^*))
      +|\core_\kappa(C^*)|\per(C^*)}{\cost_\OPT(C^*)}.
\]
In this notation, (\ref{eq:radius-diameter-condition}) is the inequality $(2+\kappa)x
\geq(2+6\eps)$ and \Cref{tab:winwin-bounds} is as follows:
\begin{center}
\(
\begin{array}{c|c|c}
 &\LS&\TT\\ \hline
\OPT_S&1+\eta&1+\eta\\[4pt]
\olamearly&3&1+2x/e\\[4pt]
\olamlate\cup\OPT_P&1+x&1+2(1+\kappa)x/e\\[4pt]
\OPT_R&2+x-y_\kappa&1+x
\end{array}
\)
\end{center}

Let $t$ be $x$ truncated to the interval $[\eta,1]$:
\[
t:=\max\{\eta,\min\{1,x\}\}.
\]
By \eqref{eq:per-bd}, $x\leq1+2\eps$, so $x\leq t+2\eps$.

Increasing $\kappa$ enlarges the core and improves the residual
bound, but increases the late-cluster bound. We set $\kappa$ to balance these bounds. Specifically, using the convenient bound
$(1-a)2/e<13/20$, we choose $\kappa$ to make the latter bound
tight at $t$:
\begin{equation}\label{eq:winwin-kappa-choice}
at+\frac{13}{20}(1+\kappa)t=\eta,
\qquad
1+\kappa=\frac{89-12t}{65t}.
\end{equation}
Since $\eta\leq t\leq1$, this gives
\[
\kappa\in\left[\frac{12}{65},\frac{23}{65}\right]
\subseteq\left[\frac1\eta-1,\frac12\right].
\]

\mypara{The technical condition.}
Suppose $x\geq\eta/(1+\eps)$. The definition of $t$ then gives
$t\leq(1+\eps)x$, and hence
\[
(2+\kappa)x
\geq\frac{(2+\kappa)t}{1+\eps}
=\frac{89+53t}{65(1+\eps)}
\geq\frac{153\eta}{65(1+\eps)}
>2+6\eps,
\]
where the last inequality uses $\eps<1/200$.
This proves \eqref{eq:radius-diameter-condition}.

\mypara{The first three rows.}
The weighted coefficient in the $\OPT_S$ row is exactly
$1+\eta$. For the early and late rows, respectively, the
weighted coefficients minus $1$ satisfy
\begin{align*}
2a+(1-a)\frac{2x}{e}
&\leq 2a+\frac{13}{20}t+O(\eps)
 \leq\eta+O(\eps),\\
ax+(1-a)\frac{2(1+\kappa)x}{e}
&\leq at+\frac{13}{20}(1+\kappa)t+O(\eps)
 =\eta+O(\eps).
\end{align*}
Here we used $x\leq t+2\eps$, $2a+13/20=\eta$, and
\eqref{eq:winwin-kappa-choice}.

\mypara{The residual row.}
Suppose $c^*\in\OPT_R$. By
\Cref{thm:struct2}\ref{cond:struct2-props-remaining}%
\ref{cond:struct2-per-approx-avg},
$x\geq\eta/(1+\eps)$. Together with the bounds above, this gives
\[
x=t+O(\eps).
\]
The residual coefficient is $1+a+x-ay_\kappa$. To control it, we
need a lower bound on the core saving $y_\kappa$. The key observation
is that too many clients outside the core would make the
average cost too large relative to the percentile.

Let $r:=|\core_\kappa(C^*)|/|C^*|$. All but at most
$\eps|C^*|$ clients have distance at least $\per(C^*)$ from
$c^*$, and clients outside the core have distance greater than
$(1+\kappa)\per(C^*)$. Dividing the resulting cost bounds by
$\cost_\OPT(C^*)$ gives
\[
1\geq\bigl((r-\eps)+(1+\kappa)(1-r)\bigr)x,
\qquad
y_\kappa\geq(r-\eps)x+rx=(2r-\eps)x.
\]
The first inequality bounds the core contribution:
\[
rx\geq\frac{(1+\kappa)x-1-\eps x}{\kappa}.
\]
Substituting into the second yields
\[
y_\kappa\geq
\frac{2((1+\kappa)x-1)}{\kappa}
-\eps x\left(1+\frac2\kappa\right)
=
\frac{2((1+\kappa)x-1)}{\kappa}-O(\eps),
\]
using $x\leq1+2\eps$ and $\kappa\geq12/65$.
Since $x=t+O(\eps)$, we conclude that
\[
1+a+x-ay_\kappa
\leq
1+a+t-\frac{2a}{\kappa}\bigl((1+\kappa)t-1\bigr)
+O(\eps).
\]

It therefore remains to show
\[
a+t-\frac{2a}{\kappa}\bigl((1+\kappa)t-1\bigr)\leq\eta.
\]
Our choice of $\kappa$ gives $t=89/(77+65\kappa)$.
Substituting this expression, multiplying by the positive
denominator, and collecting terms reduces the desired
inequality to
\[
\left(\frac12+\frac1{2000}\right)\kappa
+\frac{2a^2}{\kappa}
\geq2a-\frac1{2000}.
\]
This follows from the arithmetic--geometric mean inequality
\[
\frac{\kappa}{2}+\frac{2a^2}{\kappa}\geq2a.
\qedhere
\]
\end{proof}